\documentclass[11pt]{article}

\usepackage[margin=1in]{geometry}
\usepackage[utf8]{inputenc}
\usepackage[T1]{fontenc}
\usepackage{graphicx}

\usepackage{amsmath}
\usepackage{amssymb}

\usepackage{amsthm}
\newtheorem{definition}{Definition}[section]
\newtheorem{lemma}[definition]{Lemma}
\newtheorem{proposition}[definition]{Proposition}

\newtheorem{theorem}[definition]{Theorem}
\theoremstyle{definition}
\newtheorem{assumption}[definition]{Assumption}
\theoremstyle{remark}
\newtheorem{remark}[definition]{Remark}
\newtheorem{notation}[definition]{Notation}

\usepackage{authblk}

\usepackage{hyperref}
\hypersetup{
    colorlinks=true,
    linkcolor=black,
    citecolor=black,
    urlcolor=blue,
    pdftitle={A principled closure framework for higher-order SIS epidemic models on networks}
}

\usepackage{siunitx}      
\usepackage{bm}           
\usepackage{dcolumn}      

\usepackage{cleveref}

\newcommand{\prob}[1]{\langle#1\rangle}

\newcommand{\set}[1]{\{#1 \}}

\newcommand{\overbracket}[2]{\overbrace{#1}^{\textrm{#2}}}
\newcommand{\underbracket}[2]{\underbrace{#1}_{\textrm{#2}}}

\title{A principled closure framework for higher-order SIS epidemic models on networks}

\author[1]{Kevin Teo\thanks{Contact: \texttt{teo.ke@northeastern.edu}}}
\author[2]{Peter L. Simon}
\author[1]{Istv\'an Zolt\'an Kiss}

\affil[1]{Network Science Institute, Northeastern University London, London, \break E1W 1LP, United Kingdom}
\affil[2]{Institute of Mathematics, E\"otv\"os Lor\'and University, Budapest, Hungary}

\date{}

\begin{document}

\maketitle


\begin{abstract}
\noindent
Susceptible-infected-susceptible (SIS) epidemic models on networks are governed by hierarchical moment equations in which the dynamics of smaller subsystems depend on the state of larger ones. Moment closure approximations, which truncate this hierarchy by expressing higher-order state probabilities in terms of lower-order ones, are therefore essential for obtaining tractable reduced systems. Higher-order network models, which extend the pairwise contact structure to include group interactions encoded by hyperedges, introduce a combinatorial explosion of possible closure configurations, making systematic derivation significantly harder. Consequently, existing higher-order SIS models are typically derived heuristically, where structural and dynamical assumptions underpinning their closures are not always immediately apparent from~the formulation alone. We develop a principled bottom-up derivation of higher-order SIS dynamics, building systematically from node-level equations through pairs, triplets, and three-body interactions. Central to our approach is a network-dependent closure operator that generates topologically appropriate approximations from local pairwise and triadic structure. Within this unified framework, we recover three existing higher-order SIS models---Burgio et al.'s maximal clique model, Malizia et al.'s pair-based and inter-order models---as special cases, each arising under specific and identifiable combinations of topological and dynamical assumptions. This derivation makes explicit assumptions that are not immediately visible from heuristic approaches: for instance, the homogeneity assumption in Burgio et al. operates at the node level, stronger than the population-level average statement; and in Malizia et al., the inter-order overlap parameter is insufficient alone to express the model within our framework, with~the original derivation implicitly invoking additional structural assumptions that collapse the requisite parameterisation down to the single overlap parameter, even as the resulting equations perform~well against simulations. Our framework thus offers both a rigorous foundation for higher-order epidemic modeling and a constructive pathway for understanding the assumptions implicit in heuristically derived mean-field closures, and provides a principled method of generating new models.
\end{abstract}

\section{Introduction} \label{sec:introduction}

The spread of infections, ideas, and behaviors through a population has traditionally been modeled as a stochastic dynamical process on a contact network~\cite{anderson1991infectious, newman2002spread}, with variants of the Susceptible--Infected--Susceptible (SIS) epidemic model as the canonical setting for simulation and analytical study~\cite{pastor2015epidemic, kiss2017mathematics}. Over the past several years, interest has grown in addressing a fundamental limitation of traditional networks: their restriction to pairwise edges~\cite{benson2021higher, bick2023higher}. Certain interactions --- the coauthorship of a paper, a conversation among several people, a group activity --- cannot be reduced to combinations of pairwise interactions, since the number of participants alters the qualitative character of the interaction; these naturally live as irreducible higher-order interactions, and projecting them onto pairwise links discards critical, dynamically relevant information~\cite{battiston2020networks, st2025defining}.

Epidemic models defined on higher-order networks, which retain both pairwise edges and grouped hyperedges, exhibit qualitatively new phenomena absent from their pairwise counterparts, governed by the strength of the group interaction~\cite{majhi2022dynamics}. Iacopini et al.~\cite{iacopini2019simplicial} first established bistability and discontinuous (``explosive'') transitions in simplicial SIS contagion, and de Arruda et al.~\cite{deArruda2020} extended the framework to general higher-order social contagion on hypergraphs. Subsequent work has explored the role of degree heterogeneity~\cite{malizia2025disentangling}, pair-level correlations~\cite{malizia2025pair}, overlap between higher-order edges~\cite{malizia2025hyperedge}, and the effects of nested hyperedge structures~\cite{burgio2024triadic, malizia2026nested}. This progress has been built on a common methodological foundation: compact mean-field models --- small systems of ordinary differential equations in the global infected fraction and one or two macroscopic parameters such as the mean degrees. These models are analytically tractable and have been reconciled with stochastic simulation closely enough to underpin the field's quantitative understanding. Their governing equations, however, are not derived from a fully general node-level description of the stochastic dynamics. They are obtained instead from an intermediate-level description in which a collection of structural and dynamical assumptions has already been imposed: node-state independence, neighbor correlations, hyperedge homogeneity, and the order of closure. The resulting implicit assumptions about topology and dynamics are not always made explicit, which limits the generalizability of the resulting models.

This research gap is a natural consequence of the explosion in combinatorial complexity that arises when considering higher-order interactions in conjunction with pairwise interactions. A first-principles derivation starts from the exact microscopic rate equation for the infection probability at a single node~\cite{sharkey2008deterministic, sharkey2011deterministic}, which is a function of the joint state probabilities on every local sub-graph (a \emph{motif}) in which the node participates. Each motif corresponds to a distinct infection pathway, and the rate equation for its joint state probability depends in turn on motifs of larger size, generating a hierarchy of coupled equations. Moment-closure techniques address such hierarchies by approximating higher-order motif probabilities as functions of lower-order ones~\cite{Kuehn2016}, and have a long history on pairwise networks~\cite{rand1999correlation, keeling1999effects, house2009motif, wuyts2022mean}. On higher-order networks, however, the number of motif isomorphism classes grows faster than exponentially in the motif size~\cite{hedge1983enumeration}. The bookkeeping is already infeasible at motifs of four or five nodes, and existing closures have been constructed manually under heavy structural restrictions on the underlying network to limit the number of dynamically relevant cases.

Our central contribution is a \emph{closure operator} that, given any local motif specified by its adjacency data, returns a pseudo-probability for the joint node-state on that motif as a product of lower-order joint probabilities. The operator uses the Kirkwood superposition approximation~\cite{kirkwood1935statistical, matsuda2000physical} as its basis and augments it with two topology-aware operators: a \emph{hyperedge operator} that preserves triplet-level joint probabilities, and an \emph{independence operator} that decomposes higher-order joint probabilities into lower-order marginals. Together, these account for the presence and absence of specific hyperedges and pairwise edges within any motif and its isomorphs. This yields a closed set of rate equations for SIS dynamics that retain dependencies on any motif of size four and five, where prior closures avoided most of them altogether.

To deploy this closure operator, we first write the exact microscopic SIS rate equations at the node, pair, and triplet levels on a generic higher-order network, imposing no structural assumptions beyond the existence of pairwise edges and hyperedges. We then apply the operator to each term involving joint probabilities over more than three nodes. The resulting equations recover existing models as special cases, both a consistency check and a precise account of where each prior model's structural assumptions enter, and in one case expose an implicit assumption in the equations of~\cite{malizia2025pair} and suggest a refinement. Identifying which of these assumptions can be relaxed and which are~essential to the mean-field approach is an important direction for future research in higher-order epidemics.

The remainder of the paper is organized as follows. Section~\ref{sec:preliminaries} fixes the network model and the macroscopic topological descriptors. Section~\ref{sec:model-and-methods} provides the exact microscopic rate equations and introduces the closure operator. Section~\ref{sec:canonical-closures} demonstrates the closure operator to derive several commonly used closures in the literature, while section~\ref{sec:mean-field-approx} lays out the mean-field assumptions used. Section~\ref{sec:results} derives existing models~\cite{burgio2024triadic, malizia2025pair, malizia2026nested} as special cases of our general framework. Section~\ref{sec:discussion} discusses extensions to other higher-order dynamical processes and open directions.

\section{Preliminaries}\label{sec:preliminaries}

We consider static, undirected, unweighted contact structures that include both pairwise and triplet~interactions.

{
\begin{definition} \label{def:higher-order-networks}
    A \emph{higher-order network} $\mathcal{G} = (\mathcal{N}, \mathcal{E}, \mathcal{H})$ consists of a finite set of nodes~$\mathcal{N}$, a set of edges $\mathcal{E} \subseteq \binom{\mathcal{N}}{2}$, and a set of hyperedges $\mathcal{H} \subseteq \binom{\mathcal{N}}{3}$, where $\binom{\mathcal{N}}{k}$ denotes the collection of all $k$-element subsets of~$\mathcal{N}$. We write $e=\set{ij}$ for a generic edge and $h=\set{ijk}$ for a generic hyperedge. We denote the total number of nodes as $|\mathcal{N}|=N$.
\end{definition}
}

\begin{definition}\label{def:adjacency-tensors}
  The \emph{adjacency matrix}
  $\mathbf{A} \in \set{0,1}^{N \times N}$ and the \emph{hyperedge
  adjacency tensor}
  $\mathbf{H} \in \set{0,1}^{N \times N \times N}$ of a higher-order
  network~$\mathcal{G} = (\mathcal{N}, \mathcal{E}, \mathcal{H})$ are defined by
  \[
    A_{ij} =
    \begin{cases}
      1 & \text{if } \set{ij} \in \mathcal{E}, \\
      0 & \text{otherwise},
    \end{cases}
    \qquad
    H_{ijk} =
    \begin{cases}
      1 & \text{if } \set{ijk} \in \mathcal{H}, \\
      0 & \text{otherwise}.
    \end{cases}
  \]
  Since the network is undirected, $\mathbf{A}$ is symmetric
  ($A_{ij} = A_{ji}$) and $\mathbf{H}$ is symmetric under all
  permutations of its indices
  ($H_{ijk} = H_{ikj} = H_{jik} = \cdots$).  We set $A_{ii} = 0$ for all~$i$ and $H_{ijk} = 0$
  whenever any two indices coincide. For convenience, we also denote the complements as
  $\bar{A}_{ij} = 1 - A_{ij}$ and
  $\bar{H}_{ijk} = 1 - H_{ijk}$.
\end{definition}

\begin{remark}\label{rem:hyperedge-edge-independence}
  At this stage we make no assumption about any correlation (or lack thereof) between the distribution of hyperedge and edges among nodes and pairs of nodes. The presence of edges between pairs within a triplet provides no information about whether those three nodes also form a~hyperedge.
\end{remark}

\begin{notation} \label{not:summation}
  We write $\sum_{i,j,k \in \mathcal{N}}$ for the triplet sum $\sum_{i \in \mathcal{N}} \sum_{j \in \mathcal{N}}
  \sum_{k \in \mathcal{N}}$ with all indices running independently. When mutual distinctness among the indices is required we write
  $\sum_{\substack{i,j,k \in \mathcal{N} \\ \textup{distinct}}}$. Ordered restricted sums are written $\sum_{\substack{x < y \\ x,y \neq i,j}}$, meaning the sum runs
  over all $x, y \in \mathcal{N}$ with $x < y$ and $x, y \notin \set{i,j}$ (and similarly for larger exclusion sets).
  Unless stated otherwise, all sums run over~$\mathcal{N}$. As a convention, we use letters $i,j,k,l$ for fixed node indices, and $x,y,z,w$ for summation indices. 
\end{notation}

\pagebreak

\begin{remark} \label{rem:summation-convention}
    An ordered double sum can be rewritten as an unordered double sum with a factor of $\frac{1}{2}$ to account for double-counting:
    \[
        \sum_{\substack{x < y \\ x,y \neq i,j}} f(x,y) 
        = \frac{1}{2} \sum_{\substack{x,y \\ y \neq x,\; x,y \neq i,j}} f(x,y)
    \]
    for any function $f$ symmetric in its arguments. The unordered form is used commonly within the existing~literature. 
\end{remark}

\begin{definition}\label{def:mean-degrees}
  The \emph{degree} and \emph{hyperedge degree} of node~$i$ are
  \[
    k_{1,i} = \sum_{j} A_{ij},
    \qquad
    k_{2,i} = \frac{1}{2}
      \sum_{\substack{j,k \\ \textup{distinct}}} H_{ijk}.
  \]
  The \emph{mean degree} and \emph{mean hyperedge degree} are
  \[
    k_1 = \frac{1}{N}\sum_{i} k_{1,i}
        = \frac{1}{N}\sum_{i,j} A_{ij}
        = \frac{2|\mathcal{E}|}{N},
    \qquad
    k_2 = \frac{1}{N}\sum_{i} k_{2,i}
        = \frac{1}{2N}\sum_{i,j,k} H_{ijk}
        = \frac{3|\mathcal{H}|}{N}.
  \]
\end{definition}

\begin{definition}\label{def:clustering}
  The \emph{global clustering coefficient} of~$\mathcal{G}$ is
  \[
    \phi
    = \frac{
        \sum_{\substack{i,j,k \\ \textup{distinct}}}
          A_{ij}\, A_{ik}\, A_{jk}
      }{
        \sum_{\substack{i,j,k \\ \textup{distinct}}}
          A_{ij}\, A_{ik}
      },
  \]
  the proportion of connected triples that close into triangles.
  Consequently, the proportion of connected triples that form open wedges is
  \[
    \frac{
        \sum_{\substack{i,j,k \\ \textup{distinct}}}
          A_{ij}\, A_{ik}\, (1 - A_{jk})
      }{
        \sum_{\substack{i,j,k \\ \textup{distinct}}}
          A_{ij}\, A_{ik}
      }
    = 1 - \phi.
  \]
\end{definition}

\subsection{State probabilities}\label{subsec:state-prob}

Each node $i \in \mathcal{N}$ has an associated state $X_i \in \set{S_i, I_i}$, denoting whether $i$ is susceptible ($S$) or infected ($I$).  More generally, any group of nodes $i, j, \ldots, k \in \mathcal{N}$ has a joint state $(X_i, X_j, \ldots, X_k)$ with each $X_\ell \in \set{S_\ell, I_\ell}$ independently.  The joint state is defined for any collection of nodes, regardless of whether they are connected by edges or
hyperedges.

\begin{notation}\label{not:state-prob}
  Throughout, angled brackets $\prob{\,\cdot\,}$ denote \emph{probabilities}. For a single node,
  \[
    \prob{X_i} \equiv \mathbb{P}(X_i),
  \]
  is the probability that node~$i$ is in state~$X_i$. For a group of nodes $i, j, \ldots, k \in \mathcal{N}$,
  \[
    \prob{X_i X_j \ldots X_k} \equiv \mathbb{P}(X_i, X_j, \ldots, X_k),
  \]
  is the joint probability that nodes $i, j, \ldots, k$ simultaneously occupy states $X_i, X_j, \ldots, X_k$ respectively. Each $X_\ell$ denotes a generic state of node~$\ell$; specific states are substituted as needed, e.g. $\prob{S_i I_j I_k}$ for the probability that $i$ is susceptible while $j$ and $k$ are infected. 
\end{notation}

\section{Model and Methods} \label{sec:model-and-methods}

\subsection{SIS dynamics and microscopic rate equations}\label{subsec:sis-dynamics}

The Susceptible--Infected--Susceptible (SIS) model describes the evolution of a disease spreading process on a contact network~\cite{hethcote2000mathematics, sharkey2008deterministic, kiss2017mathematics}, where we treat the global state $\mathbf{X}(t) \in \set{S, I}^\mathcal{N}$ as a continuous-time Markov chain: each allowed transition (recovery of an infected node, or infection of a susceptible node through an edge or hyperedge) occurs as an independent Poisson process at a prescribed rate, and the resulting distribution evolves according to the Kolmogorov forward equation. The microscopic rate equations stated below describe the evolution of a node or a group of nodes, obtained by taking the marginal probabilities of the Kolmogorov forward equation over the nodes outside the group of interest.

On a contact network described by a higher-order network $\mathcal{G} = (\mathcal{N}, \mathcal{E}, \mathcal{H})$, nodes are either susceptible or infected at any given time. A susceptible node can become infected through (i) a pairwise interaction with an infected node at rate $\beta_1$, or (ii) a triplet interaction with two infected nodes at rate $\beta_2$. The presence of pairwise (edge) and triplet (hyperedge) interactions are encoded by $\mathcal{E}$ and $\mathcal{H}$ respectively. Note that hyperedge infections require both hyperedge neighbors to be infected. Finally, infected nodes can recover back to a susceptible state at a rate $\gamma$. 

The state probabilities evolve according to a hierarchy of coupled rate equations. We state the equations for the first three moments; higher-order terms will be addressed by the closure approximation developed in Subsection~\ref{subsec:closure}.

\subsubsection{First moment}

The probability that node~$i$ is infected evolves as
\begin{equation} \label{eq:dI}
    \frac{d}{dt}\prob{I_i} = 
    \underbracket{-\gamma \prob{I_i}}{recovery} + 
    \underbracket{\beta_1 \sum_{x\not=i} A_{ix} \prob{S_i I_x}}{pairwise infection} +
    \underbracket{\beta_2 \sum_{\substack{x<y \\x,y \not= i}} H_{ixy}\prob{S_i I_x I_y}}{hyperedge infection}.
\end{equation}
The pairwise infection term couples the single-node probability $\prob{I_i}$ to pair probabilities $\prob{S_i I_x}$, while the hyperedge infection term introduces triplet probabilities $\prob{S_i I_x I_y}$. The pre-factors $\gamma$, $\beta_1$, and $\beta_2$ clearly distinguish recovery, pairwise infection, and hyperedge infection terms, respectively. 

As seen from equation~\eqref{eq:dI}, evaluating the first moment equation requires the pair and triplet level state probabilities. These are the second and third moments of the system; we develop the rate equations subsequently.

\subsubsection{Second moment}\label{subsubsec:second-moment}
Here we introduce the rate equation for the pair probabilities, i.e. the second moment. For a fixed pair of nodes $i,j$, we consider how interactions between $i$ and $j$ as well as with external nodes $x,y \not \in \set{i,j}$ drive the system.

The evolution of $\prob{S_i I_j}$ is
\begin{equation}\label{eq:dSI}
\begin{aligned}
  \frac{d}{dt}\prob{S_i I_j} = 
    & \overbracket{-\gamma\,\prob{S_i I_j}
    + \gamma\,\prob{I_i I_j}}{recovery} \\
    & \underbracket{- \beta_1 A_{ij}\,\prob{S_i I_j}}{internal pairwise infection} \\
    & \underbracket{- \beta_1 \sum_{x \neq i,j} A_{ix}\,\prob{S_i I_j I_x}
    + \beta_1 \sum_{x \neq i,j} A_{jx}\,\prob{S_i S_j I_x}}{external pairwise infection} \\
    & \underbracket{- \beta_2 \sum_{x \neq i,j} H_{ijx}\,\prob{S_i I_j I_x}}{mixed hyperedge infection} \\
    & \underbracket{- \beta_2 \sum_{\substack{x < y \\ x,y \neq i,j}}
      H_{ixy}\,\prob{I_x I_y S_i I_j}
    + \beta_2 \sum_{\substack{x < y \\ x,y \neq i,j}}
      H_{jxy}\,\prob{S_i S_j I_x I_y}}{external hyperedge infection}.
\end{aligned}
\end{equation}
Here, the equation is organized into $5$ lines, each reserved for a particular type of interaction, depending on which nodes drive the infection --- the fixed nodes $i$ and/or $j$ (internal), dummy nodes $x$ and $y$ that are summed over (external), or both (mixed). The first line is the recovery term. The second line is the \emph{internal} pairwise infection term driven by node $j$. The third line is the \emph{external} pairwise infection term driven by a dummy node $x$ that is summed over. The fourth line describes the \emph{mixed} hyperedge infection term driven by one of the fixed nodes ($j$ in this case) and a dummy node $x$ that is summed over. The fifth and final line describes the purely \emph{external} hyperedge infection term driven by two dummy nodes $x$ and $y$ that are both summed over. 

This equation is valid for any pair $i, j$, regardless of whether they are connected by an edge or share a hyperedge---although as equation~\eqref{eq:dI} implies, in practice we only need to track the case where $A_{ij}=1$. We have written the equation in this way for generality; in practice, we can multiply the left-hand side and right-hand side of~\eqref{eq:dSI} by $A_{ij}$ to achieve the relevant form of the equation. 

\subsubsection{Third moment}\label{subsubsec:third-moment}

We also require an expression for the triplet probabilities, i.e. the third moment. For a fixed triplet $i, j, k$, we consider the evolution of the triplet under internal interactions between $i, j, k$, as well as interactions with external nodes $x,y \not\in \set{i,j,k}$. This yields the following equations:

\begin{equation} \label{eq:dSSI}
\begin{aligned}
    \frac{d}{dt}\prob{S_i S_j I_k} =& 
      \overbracket{\gamma \Big(\prob{I_i S_j I_k} + \prob{S_i I_j I_k} - \prob{S_i S_j I_k} \Big)}{recovery}  +
      \overbracket{\beta_1 \Big(- { [A_{ik} + A_{jk}]} \times \prob{S_i S_j I_k} \Big)}{internal pairwise infection}   +\\
    &  \underbracket{\beta_1 \Big( \sum_{x \not= i,j,k} [-A_{ix} - A_{jx}]\prob{I_x S_i S_j I_k} +A_{kx}\prob{S_i S_j S_k I_x}  \Big)}{external pairwise infection}   +
      \underbracket{0}{internal hyperedge infection}   + \\&
      \underbracket{\beta_2 \Big( \sum_{x \not=i,j,k} [-H_{ikx} - H_{jkx}] \prob{I_x S_i S_j I_k} \Big)}{mixed hyperedge infection}   +\\
    & \underbracket{\beta_2 \Big( \sum_{\substack{x<y \\x,y\not= i,j,k }} [-H_{ixy} -H_{jxy}] \prob{I_x I_y S_i S_j I_k} + H_{kxy} \prob{S_i S_j S_k I_x I_y} \Big)}{external hyperedge infection}
\end{aligned}
\end{equation}

\begin{equation} \label{eq:dSII}
\begin{aligned}
    \frac{d}{dt}\prob{S_i I_j I_k} =& 
      \overbracket{\gamma \Big(\prob{I_i I_j I_k} - 2\times\prob{S_i I_j I_k}\Big)}{recovery} +
      \overbracket{\beta_1 \Big({ [-A_{ij} - A_{ik}]} \prob{S_i I_j I_k} + {A_{jk}}[\prob{S_i I_j S_k} + \prob{S_i S_j I_k}] \Big)}{internal pairwise infection}   + \\&
      \underbracket{\beta_1 \Big(\sum_{x\not=i,j,k} -A_{ix} \prob{I_x S_i I_j I_k} + A_{jx} \prob{I_x S_i S_j I_k} + A_{kx} \prob{I_x S_i I_j S_k}\Big)}{external pairwise infection}   +
      \underbracket{-\beta_2 H_{ijk} \prob{S_i I_j I_k}}{internal hyperedge infection}   + \\&
      \underbracket{\beta_2 \Big(\sum_{x\not= i,j,k} [-H_{ijx} - H_{ikx}] \prob{I_x S_i I_j I_k} + H_{jkx}[\prob{I_x S_i S_j I_k} + \prob{I_x S_i I_j S_k}] \Big)}{mixed hyperedge infection}   + \\&
      \underbracket{\beta_2 \Big( \sum_{\substack{x<y \\ x,y\not= i,j,k}} -H_{ixy}\prob{I_x I_y S_i I_j I_k} + H_{jxy}\prob{S_i S_j I_k I_x I_y} + H_{kxy}\prob{S_i I_j S_k I_x I_y} \Big)}{external hyperedge infection}
\end{aligned}
\end{equation}
As with the pair-level equation, these terms are organized by interaction type. The first lines of equations~\eqref{eq:dSSI} and~\eqref{eq:dSII} contain the recovery terms followed by the \emph{internal} pairwise infection terms. The second lines contain the \emph{external} pairwise infection terms together with the \emph{internal} hyperedge infection term. The latter requires two infectious fixed nodes to activate the hyperedge $\{i,j,k\}$, so it appears only in equation~\eqref{eq:dSII} (where $I_j$ and 
$I_k$ are both infected) and vanishes in equation~\eqref{eq:dSSI} (where only $I_k$ is infected); it cannot appear at the pair level at all, since a hyperedge requires three fixed nodes. The third lines describe the \emph{mixed} hyperedge infection terms, and the fourth and final lines describe the purely \emph{external} hyperedge infection terms.

These equations are valid for any triplet $i, j, k$ regardless of the connectivity among the three nodes. 
The rest of the individual and pair terms can be obtained by
\begin{equation}
\begin{aligned}
    \prob{S_i} =& 1 - \prob{I_i}, \\
    \prob{I_i I_j} =& \prob{I_i} - \prob{I_i S_j}, \\
    \prob{S_i S_j} =& \prob{S_i} - \prob{S_i I_j}, 
\end{aligned}
\end{equation}
and the rest of the triplet terms can be obtained by
\begin{equation}
\begin{aligned}
    \prob{I_i I_j I_k} =& \prob{I_i} - \prob{I_i S_j I_k} - \prob{I_i I_j S_k} - \prob{I_i S_j S_k}, \\
    \prob{S_i S_j S_k} =& \prob{S_i} - \prob{S_i I_j I_k} - \prob{S_i S_j I_k} - \prob{S_i I_j S_k}.
\end{aligned}
\end{equation}

\subsection{Closure problem} \label{subsec:closure-problem}
The rate equations in~\eqref{eq:dI}--\eqref{eq:dSII} exhibit a hierarchical coupling: each equation at the $m$-node level introduces joint probabilities on $m+1$ nodes (for pairwise infections) and $m+2$ nodes (for hyperedge infections). Concretely, the single-node equation depends on pair and triplet probabilities; the pair equation depends on triplet and quadruplet probabilities; the triplet equations depend on quadruplet and quintuplet probabilities. In theory, this hierarchy extends until $m$ equals the network size $N$, making it analytically intractable and computationally impractical. 

To obtain a reasonably small and closed system of equations, we truncate the hierarchy at the level of triplets. This requires approximating all quadruplet $\prob{X_i X_j X_x X_y}$ and quintuplet probabilities $\prob{X_i X_j X_k X_x X_y}$ in terms of singles $\prob{X_i}$, pairs $\prob{X_i X_j}$, and triplets $\prob{X_i X_j X_k}$~\cite{Kuehn2016}. 

The challenge is that the appropriate closure depends on the local topology.  Consider two configurations of three nodes $i,j,k$ with pairwise edges only.  If all three are mutually connected (a triangle), then infection of node~$i$ can spread directly to both~$j$ and~$k$.  If instead only the pairs $\set{i,j}$ and $\set{j,k}$ are connected (a wedge), then infection of~$i$ can reach~$k$ only indirectly through~$j$, producing markedly different short-time dynamics.  The same principle holds for sub-graphs (we use sub-graph and motifs interchangeably) of any size: the topology of the motif spanned by a group of nodes strongly determines the local infection dynamics.

Moreover, the closure itself eliminates information about the external network.  The information that remains---and that plays a decisive role in the dynamics---is the local topology within the lower-order node groupings, characterized by its isomorphism class.  Joint probabilities on sub-graphs of different isomorphism classes therefore require separate approximations.

\subsection{Sub-graph isomorphism classes} \label{subsec:subgraph-isos}
We now classify the sub-graphs that arise on groups of nodes up to isomorphism.  Two sub-graphs are \emph{isomorphic} if there exists a bijection between their node sets under which every edge and hyperedge is preserved~\cite{diestel2017graph}; since isomorphic sub-graphs are topologically identical, they thus admit the same topological closure decomposition, making this the natural equivalence relation.

For a single node~$i$ there is trivially one isomorphism class.  For a pair $i,j$ there are two isomorphism classes: connected ($A_{ij} = 1$) and disconnected ($A_{ij} = 0$).

For triplets $i, j, k$, the combinatorial growth begins. Hyperedges are now relevant, as three is the minimum number of nodes required to form a hyperedge.  The hyperedge attribute has $2$ possibilities: present ($H_{ijk} = 1$, denoted~$\bigcirc$) or absent ($\bar{H}_{ijk} = 1$, denoted~$\otimes$).  At the pairwise level, there are $3$ potential edges, each present or absent, giving $2^3 = 8$ configurations.  However, the symmetry of three nodes means that configurations with the same number of edges are isomorphic, so we can classify by simply counting edges: three
($\triangle$), two ($\wedge$), one ($-$), or zero ($\bullet$). Combined with the hyperedge attribute, this gives $2 \times 4 = 8$ isomorphism classes for triplets.

\begin{figure}[hb!]
    \centering
    \includegraphics[width=0.5\linewidth]{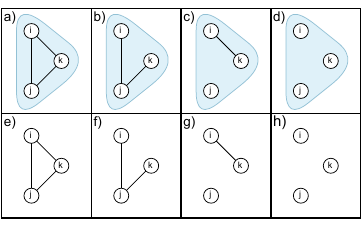}
    \caption{Diagrams of all possible $3$-node sub-graphs with pairwise edges and hyperedges. Hyperedges are shown in blue in panels (a)--(d).}
    \label{fig:triplets}
\end{figure}

\begin{notation}    
    To distinguish state probabilities belonging to different isomorphism classes, we encode the class as a superscript.  The local topologies described by the $\mathbf{A}$ and $\mathbf{H}$ factors fall into eight isomorphism classes: any given set of nodes $\set{i,j,k}$ belongs to one and only one class.  When the specific class is unknown, we treat $\prob{X_i X_j X_k}$ as a superposition over all classes,
    \[
      \prob{X_i X_j X_k}
      = \sum_{\substack{\tau \in \set{\bigcirc, \otimes},\\ \sigma \in \set{\triangle, \wedge, -, \bullet}}}
        f_{\tau\sigma}(\mathbf{A}, \mathbf{H})\,
        \prob{X_i X_j X_k}^{\tau\sigma},
    \]
    where $f_{\tau\sigma}(\mathbf{A}, \mathbf{H})$ is the product of adjacency factors on $\set{i,j,k}$ that selects class $\tau\sigma$.  The prefactors are mutually exclusive and are listed in Table~\ref{tab:isomorphism-class-notation}.  In the wedge and single-edge classes, one node is topologically distinguished (the wedge center or the isolated node respectively); we adopt the convention that the topologically distinguished node occupies the central (second) position in all state probabilities for the $\wedge$ and $-$ classes.  For example, $\prob{S_i I_j S_k}^{\otimes\wedge}$ denotes a wedge with center~$j$.
    
    \begin{table}[t]
    \centering
    \renewcommand{\arraystretch}{1.4}
    \caption{Isomorphism classes for triplets of nodes $i, j, k$.  The topologically distinguished node (the wedge center in the $\wedge$ class and the isolated node in the $-$ class) is positioned in the middle of the state probability, i.e. the second entry of $\prob{X_i X_j X_k}^{\bigcirc \wedge}$ implies node $j$ as the wedge center.}
    \label{tab:isomorphism-class-notation}
    \begin{tabular}{ l l l p{45.7446pt} }
      \hline
      \textbf{$\prob{X_i X_j X_k}$ Prefactor} &
        \textbf{Notation} & \textbf{Notes}  & \textbf{Figure~\ref{fig:triplets} panel}\\
      \hline
      $H_{ijk}\, A_{ij}\, A_{jk}\, A_{ik}$ &
        $\prob{X_i X_j X_k}^{\bigcirc\triangle}$ &
        Symmetric under all permutations  &(a)\\
      $H_{ijk}\, A_{ij}\, A_{jk}\, \bar{A}_{ik}$ &
        $\prob{X_i {X}_j X_k}^{\bigcirc\wedge}$ &
        Symmetric under $i \leftrightarrow k$; center is $j$  &(b)\\
      $H_{ijk}\, \bar{A}_{ij}\, \bar{A}_{jk}\, A_{ik}$ &
        $\prob{X_i {X}_j X_k}^{\bigcirc-}$ &
        Symmetric under $i \leftrightarrow k$; isolated node is $j$  &(c)\\
      $H_{ijk}\, \bar{A}_{ij}\, \bar{A}_{jk}\, \bar{A}_{ik}$ &
        $\prob{X_i X_j X_k}^{\bigcirc\,\bullet}$ &
        Symmetric under all permutations  &(d)\\[3pt]
      $\bar{H}_{ijk}\, A_{ij}\, A_{jk}\, A_{ik}$ &
        $\prob{X_i X_j X_k}^{\otimes\triangle}$ &
        Symmetric under all permutations  &(e)\\
      $\bar{H}_{ijk}\, A_{ij}\, A_{jk}\, \bar{A}_{ik}$ &
        $\prob{X_i {X}_j X_k}^{\otimes\wedge}$ &
        Symmetric under $i \leftrightarrow k$; center is $j$  &(f)\\
      $\bar{H}_{ijk}\, \bar{A}_{ij}\, \bar{A}_{jk}\, A_{ik}$ &
        $\prob{X_i {X}_j X_k}^{\otimes-}$ &
        Disconnected; isolated node is $j$&(g)\\
      $\bar{H}_{ijk}\, \bar{A}_{ij}\, \bar{A}_{jk}\, \bar{A}_{ik}$ &
        $\prob{X_i X_j X_k}^{\otimes\,\bullet}$ &
        Disconnected &(h)\\
      \hline 
      \textbf{$\prob{X_i X_j}$ Prefactor} &
        \textbf{Notation} & \textbf{Notes}  &\\
      \hline
      $A_{ij}$ &
        $\prob{X_i X_j}^{-}$ &
          &\\
      $\bar{A}_{ij}$ &
        $\prob{X_i X_j}^{\, \bullet}$ & 
         Can be ignored (disconnected) &\\
        \hline
    \end{tabular}
    \end{table}
\end{notation}

\begin{remark}\label{rem:shorthand-notation}
  When only part of the isomorphism class is known, we use partial superscripts to denote the known structure, treating the unknown attribute as a superposition over its possible values.  If the hyperedge status is fixed but the pairwise structure is unspecified,
  \[
    \prob{X_i X_j X_k}^{\tau}
    = \sum_{\sigma \in \set{\triangle, \wedge, -, \bullet}}
      f_{\tau\sigma}(\mathbf{A}, \mathbf{H})\,
      \prob{X_i X_j X_k}^{\tau\sigma};
  \]
  conversely, if the pairwise structure is fixed but the hyperedge status is unspecified,
  \[
    \prob{X_i X_j X_k}^{\sigma}
    = \sum_{\tau \in \set{\bigcirc, \otimes}}
      f_{\tau\sigma}(\mathbf{A}, \mathbf{H})\,
      \prob{X_i X_j X_k}^{\tau\sigma},
    \qquad \sigma \in \set{\triangle, \wedge},
  \]
  where $\sigma$ is restricted to $\triangle$ or $\wedge$ since the single-edge and empty classes ($-$ and $\bullet$) are not dynamically relevant.  Once the unknown prefactors are revealed, the mutual exclusivity of the $f_{\tau\sigma}$ selects a single surviving term.  Both shorthands arise naturally when solving the rate equations under structural constraints that specify only partial information on the adjacency structure.
\end{remark}

\begin{remark}
    Since the pair probability that appears in~\eqref{eq:dI} has an $A_{ix}$ prefactor, in practice it is much simpler to only keep track of $\prob{X_i X_j}^{-}$ and ignore $\prob{X_i X_j}^{\,\bullet}$ entirely.
\end{remark}

\subsubsection{Larger sub-graphs}\label{subsubsec:larger-classes}

For sub-graphs on $m$ nodes, the number of adjacency matrices/tensors is $\binom{m}{2} + \binom{m}{3}$ (potential edges plus potential hyperedges), giving $2^{\binom{m}{2} + \binom{m}{3}}$ configurations before accounting for isomorphisms.  For quadruplets ($m = 4$) this is $2^{10} = 1024$; for quintuplets ($m = 5$) it is $2^{20} \approx 10^6$. Even after collapsing isomorphic configurations, the number of distinct classes remains prohibitively large, and manually writing a separate closure for each is impractical. To handle this, we develop a closure operator in the next section, which decomposes any sub-graph into its constituent singles, pairs, and triplets, applying simple topology-dependent rules at each level.

\subsection{Closure Operator} \label{subsec:closure}
We develop a systematic method of generating closure approximations for quadruplet and quintuplet joint probabilities in terms of singles, pairs, and triplets. The closure method rests on two observations about how sub-graph topology governs the joint distribution over node states.

First, statistical dependence between node states is mediated by transmission paths.  If removing an edge---pairwise or hyperedge---from a sub-graph $\mathcal{S}=\set{i,j,\ldots}$ disconnects a node~$i$ from a set of nodes~$\mathcal{S}'=\mathcal{S} \setminus \set{i}$, then the evolution of $X_i$ is assumed to be independent of the states of all the nodes in $\mathcal{S}'$ and the joint probability factorizes:
\begin{equation}\label{eq:disconnection-factorisation}
  \prob{X_i X_{\mathcal{S}'}} = \prob{X_i} \cdot \prob{X_{\mathcal{S}'}}.
\end{equation}
Second, hyperedges represent irreducible three-body interactions.  Removing a pairwise edge from a hyperedge triplet changes the internal pairwise structure but does not destroy the three-node correlation; instead, the state probability changes into the corresponding isomorphism class variant. Without loss of generality, this can be expressed as:
\begin{equation}\label{eq:hyperedge-irreducible}
\begin{aligned}
    \prob{X_i X_j X_k}^{\bigcirc\triangle}
  \xrightarrow{\;\text{remove }  \set{i,j}\;} &
  \prob{X_i {X}_k X_j}^{\bigcirc\wedge}
  \quad \text{(demoted but not decomposed)}. \\ 
  \prob{X_i {X}_k X_j}^{\bigcirc\wedge}
  \xrightarrow{\;\text{remove }  \set{i,k}\;} &
  \prob{X_k {X}_i X_j}^{\bigcirc -} 
  \quad \text{(demoted but not decomposed)} \\
  \prob{X_k {X}_i X_j}^{\bigcirc-}
  \xrightarrow{\;\text{remove }  \set{i,k}\;} &
  \prob{X_i X_j X_k}^{\bigcirc  \, \bullet} 
  \quad \text{(demoted but not decomposed)} \\
\end{aligned}
\end{equation}
Recall that the position of the node labels $i,j,k$ changes to identify the topologically distinct node in the $\wedge$ and $-$ cases.

We proceed in three stages: (i) defining the maximally correlated approximation ansatz, (ii) tagging hyperedge triplets via a hyperedge operator to preserve triplet joint probability, and (iii) demoting (pairwise) correlations based on the presence/absence of edges via an independence operator. The composition of these operations defines the closure.

\subsubsection{Kirkwood superposition approximation} \label{subsubsec:kirkwood}
The Kirkwood superposition approximation provides an ansatz for the joint-probability of an $m$-state joint probability in terms of the product of $(m-1)$ state joint probabilities, which is a starting point for our approximation~\cite{kirkwood1935statistical, pearl2014probabilistic, matsuda2000physical}. 

\begin{definition}
    We define the \emph{Kirkwood superposition approximation} $\kappa \colon \binom{\mathcal{N}}{m} \to \mathbb{R}_{\geq 0}$ for a set of $m$ nodes $\mathcal{S} = \{x_1, x_2, \dots, x_m\} \in \binom{\mathcal{N}}{m}$ as 
    \begin{equation} \label{eq:kirkwood}
        \kappa \bigl(\mathcal{S}=(x_1, \ldots, x_m)\bigr) = \prod_{i=1}^{m-1} \Bigg[\prod_{\substack{s_i \subseteq \mathcal{S} \\ |s_i| = i}} \prob{s_i} \Bigg]^{(-1)^{m-1-i}},
    \end{equation}
    where $s_i$ is a subset of $\mathcal{S}$ of size $i$~\cite{matsuda2000physical}. The function $\kappa$ approximates the $m$-variable joint probability as a combination of $m'<m$-variable joint probabilities.  
\end{definition}

\begin{remark}
    The Kirkwood superposition approximation function is not a true probability density function because it is not guaranteed to be normalized, except in the case where all variables are statistically independent of each other.
\end{remark}

The Kirkwood superposition approximation alone is not sufficient to close our system of equations, as applying $\kappa$ to quintuplet terms from~\eqref{eq:dSII} and~\eqref{eq:dSSI} removes the quintuplet terms but generates new quadruplet terms. To obtain closure, we have to break down the quadruplet terms again, which we can do by iterating $\kappa$ again. Formally, we do this by introducing the truncated Kirkwood superposition approximation.
\begin{definition}
    For a set of $m$ nodes $\mathcal{S}$ with $m \geq 3$, the 
    \emph{truncated Kirkwood approximation} is
    \begin{equation} \label{eq:kirkwood-tilde}
        \tilde{\kappa}(\mathcal{S}) = 
        \Bigg(\prod_{i=1}^{m-2} \Bigg[\prod_{\substack{s \subseteq 
        \mathcal{S} \\ |s| = i}} \prob{s} 
        \Bigg]^{(-1)^{m-1-i}} \Bigg) 
        \times
        \Bigg[\prod_{\substack{s \subseteq \mathcal{S} \\ |s| = m-1}} 
        \kappa(s) \Bigg],
    \end{equation}
    where the $(m-1)$-node joint probabilities are replaced by their Kirkwood approximations~$\kappa$.
\end{definition}

\begin{remark}\label{rem:kirkwood-default-class}
  The Kirkwood approximation retains all pairwise correlations: every pair probability $\prob{X_i X_j}$ appears in the equations~\eqref{eq:kirkwood-quadruple-decomposed} and~\eqref{eq:kirkwood-quintuple-decomposed}. The resulting triplet expression therefore corresponds to the maximally correlated pairwise structure: a triangle.  Since the Kirkwood approximation carries no information about hyperedges at this stage, we assign every triplet emerging from $\kappa$ (or $\tilde{\kappa}$) the default isomorphism class $\otimes \triangle$: all three pairwise edges present, no hyperedge present.
\end{remark}

For our triadic closure specifically, we can close $4$-variable terms with $\kappa$ and $5$-variable terms with $\tilde \kappa$, such that these functions only contain up to $3$-variable terms:

\begin{equation} \label{eq:kirkwood-quadruple-decomposed}
\begin{aligned}
    \kappa \big( \set{ijkl} \big) \approx& \frac{\prob{ijk}^{\otimes\triangle} \prob{ijl}^{\otimes\triangle}  \prob{ikl}^{\otimes\triangle} \prob{jkl}^{\otimes\triangle}   }{\prob{ij} \prob{ik} \prob{il} \prob{jk} \prob{jl} \prob{kl}  } \prob{i} \prob{j} \prob{k} \prob{l},
    \end{aligned}
\end{equation}

\begin{equation} \label{eq:kirkwood-quintuple-decomposed}
\begin{aligned}
    &\tilde \kappa \big( \set{ijklm} \big) \\
    &\approx  \frac{ [\prob{ijk}^{\otimes\triangle} \prob{ijl}^{\otimes\triangle} \prob{ijm}^{\otimes\triangle} \prob{ikl}^{\otimes\triangle} \prob{ikm}^{\otimes\triangle} \prob{ilm}^{\otimes\triangle} \prob{jkl}^{\otimes\triangle} \prob{jkm}^{\otimes\triangle} \prob{jlm}^{\otimes\triangle} \prob{klm}^{\otimes\triangle}]^1}
    { [\prob{ij} \prob{ik} \prob{il} \prob{im} \prob{jk} \prob{jl} \prob{jm} \prob{kl} \prob{km} \prob{lm}]^2} 
    \\ 
    &\qquad  \times [\prob{i} \prob{j} \prob{k} \prob{l} \prob{m}]^3 ,
\end{aligned}
\end{equation}
where we have dropped the state labels $X$ for convenience and brevity. At this stage, every triplet carries the default class $\otimes\triangle$ (Remark~\ref{rem:kirkwood-default-class}), reflecting the maximal pairwise correlation assumed by the Kirkwood approximation.

\subsubsection{Hyperedge operator} \label{subsubsec:hyperedge-operator}
The Kirkwood approximation does not distinguish between triplets connected by a hyperedge and those connected by pairwise edges, or any combination of the two. However, as discussed in Subsection~\ref{subsec:closure-problem}, different isomorphism classes (determined by the hyperedge and pairwise connections) should be assigned to their own joint probability measures, as the correlation structures will yield different~dynamics. 

The hyperedge operator transforms the naive triplet probability $\prob{X_i X_j X_k}^{\otimes \triangle}$ into the hyperedge-class probabilities from Table~\ref{tab:isomorphism-class-notation}.

\begin{definition}\label{def:hyperedge-operator}
  Let $h = \set{ijk} \in \mathcal{H}$ be a hyperedge. The \emph{hyperedge operator} $\lambda_h$ acts on pseudo-probability expressions from the Kirkwood approximation by replacing a triplet probability $\prob{X_i X_j X_k}$ on the nodes $\set{i,j,k}$ with its $\bigcirc$-class counterpart:
  \[
    \lambda_{h=\set{ijk}}\colon \prob{X_i X_j X_k}^\triangle \mapsto \prob{X_i X_j X_k}^{\bigcirc \triangle}.
  \]
  If the triplet is already in a $\bigcirc$ class, $\lambda_h$ acts as the identity. 
\end{definition}
\begin{proposition}\label{prop:lambda-properties}
  The hyperedge operator has the following properties:
  \begin{enumerate}
    \item \emph{Idempotency:} $\lambda_h \circ \lambda_h = \lambda_h$.
    \item \emph{Commutativity:} $\lambda_h \circ \lambda_{h'} = \lambda_{h'} \circ \lambda_h$ for any two hyperedges $h, h' \in \mathcal{H}$.
  \end{enumerate}
\end{proposition}

\begin{proof}
  Idempotency is immediate: if the triplet is already tagged $\bigcirc$, then $\lambda_h$ acts as the identity. Commutativity between distinct hyperedge operators follows from the fact that hyperedges are defined entirely by the triplet of nodes that compose it and are independent of other hyperedges, thus the operator on $h$ has no effect on $h'$ and vice versa, leading to commutativity. 
\end{proof}

\subsubsection{Independence operator}\label{subsubsec:independence-operator}

The Kirkwood approximation treats all pair correlations as present. As established, when an edge $\set{ij}$ is absent ($A_{ij} = 0$), removing that transmission pathway weakens the correlation between $X_i$ and $X_j$. If removing the edge disconnects $i$ from $j$ within the local sub-graph, the factorization~\eqref{eq:disconnection-factorisation} applies. We first describe the effect on pairs and on non-hyperedge triplets, where removing enough edges leads to disconnection and hence factorization, and then extend to hyperedge triplets, where the irreducible three-body correlation~\eqref{eq:hyperedge-irreducible} prevents decomposition.

\begin{definition}\label{def:independence-operator}
  Let $\set{i,j}$ be a pair of nodes with $A_{ij} = 0$. The \emph{independence operator} $\Phi_{ij}$ acts on pseudo-probability expressions according to the following rules.

  \medskip
  \emph{Action on pairs.} The pair probability decomposes into a product of singles:
  \[
    \Phi_{ij}\colon \prob{X_i X_j} \mapsto \prob{X_i}\,\prob{X_j}.
  \]

  \emph{Action on non-hyperedge triplets ($\otimes$).} Removing one edge from a non-hyperedge triplet either demotes it to a lower class or decomposes it into lower-order products. Applied iteratively, one edge at a time, the full demotion chain is:
  \begin{align}
    \prob{X_i X_j X_k}^{\otimes\triangle}
      &\xrightarrow{\;\Phi_{ik}\;}
      \prob{X_i {X}_j X_k}^{\otimes\wedge},
      \label{eq:phi-tri-to-wedge} \\
    \prob{X_i {X}_j X_k}^{\otimes\wedge}
      &\xrightarrow{\;\Phi_{jk}\;}
      \prob{X_i X_j}\,\prob{X_k}.
      \label{eq:phi-wedge-to-edge} 
  \end{align}
  In~\eqref{eq:phi-tri-to-wedge}, removing $\set{ik}$ from the triangle leaves edges $\set{ij}$ and $\set{jk}$, forming a wedge with center~$j$. In~\eqref{eq:phi-wedge-to-edge}, removing $\set{jk}$ from the wedge disconnects node~$i$ from the remaining pair $\set{ij}$, and since there is no hyperedge to maintain the three-node correlation, the triplet decomposes. 

  \medskip
  \emph{Action on hyperedge triplets ($\bigcirc$).} Removing an edge from a hyperedge triplet demotes the pairwise structure exactly as in the $\otimes$ case, but the triplet \emph{never decomposes} --- the hyperedge preserves the three-node correlation regardless of how many pairwise edges are absent:
  \begin{align}
    \prob{X_i X_j X_k}^{\bigcirc\triangle}
      &\xrightarrow{\;\Phi_{ij}\;}
      \prob{X_i {X}_k X_j}^{\bigcirc\wedge},
      \label{eq:phi-h-tri-to-wedge} \\
    \prob{X_i {X}_k X_j}^{\bigcirc\wedge}
      &\xrightarrow{\;\Phi_{ik}\;}
      \prob{X_j {X}_i X_k}^{\bigcirc-},
      \label{eq:phi-h-wedge-to-edge} \\
    \prob{X_i {X}_j X_k}^{\bigcirc-}
      &\xrightarrow{\;\Phi_{ik}\;}
      \prob{X_i X_j X_k}^{\bigcirc\,\bullet}.
      \label{eq:phi-h-edge-to-empty}
  \end{align}
  The demotion chain mirrors the $\otimes$ case step by step, but at every stage the result remains a triplet probability tagged with~$\bigcirc$. As before, if the edge being removed is already absent, the operator has no effect.
\end{definition}

The contrast between the two cases is the central feature of the closure: hyperedge triplets are demoted in their pairwise structure but survive as correlated units, while non-hyperedge triplets decompose into products of lower-order terms once enough edges are removed. Therefore, the hyperedge operator must be applied before the independence operator --- the $\bigcirc$/$\otimes$ distinction must be established before the independence rules can act correctly.

\begin{proposition}\label{prop:phi-properties}
  The independence operator has the following properties:
  \begin{enumerate}
    \item \emph{Idempotency:} $\Phi_{ij} \circ \Phi_{ij} = \Phi_{ij}$.
    \item \emph{Commutativity:} $\Phi_{ij} \circ \Phi_{kl} = \Phi_{kl} \circ \Phi_{ij}$ for any two pairs $\set{i,j}$ and $\set{k,l}$.
  \end{enumerate}
\end{proposition}

\begin{proof}
  Idempotency follows from the fact that each rule either decomposes a term (pair~$\to$ product of singles, non-hyperedge triple~$\to$ lower-order product) or demotes the pairwise structure of a hyperedge triple. In either case, applying $\Phi_{ij}$ a second time encounters a term in which the edge $\set{i,j}$ is already treated as absent, so the operator acts as the identity.

  Commutativity between independence operators follows from the fact that $\Phi_{ij}$ modifies only terms involving the pair $\set{i,j}$. When $\set{i,j} \cap \set{k,l} = \emptyset$, the operators act on disjoint terms. When the pairs share a node, commutativity holds because the transformation rules depend only on which edges are present in the current state of the expression, and removing two edges in either order produces the same final configuration.
\end{proof}
\subsubsection{Closure map}\label{subsubsec:closure-map}

We now compose the operators to define the full closure. For a set of nodes $\mathcal{S}$ with local edge set $\mathcal{E}_\mathcal{S} = \set{e \in \binom{\mathcal{S}}{2} : A_e = 1}$ and local hyperedge set $\mathcal{H}_\mathcal{S} = \set{h \in \binom{\mathcal{S}}{3} : H_h = 1}$, define the set of absent edges $\bar{\mathcal{E}}_\mathcal{S} = \binom{\mathcal{S}}{2} \setminus \mathcal{E}_\mathcal{S}$.

\begin{definition}\label{def:closure-map}
  The \emph{closure map} for a set of nodes $\mathcal{S}$ with $|\mathcal{S}| \geq 4$ is
  \begin{equation}\label{eq:closure-map}
    \mathcal{C}_\mathcal{S,  E_S, H_S} = \Biggl(\,\prod_{e \in \bar{\mathcal{E}}_\mathcal{S}} \Phi_e \Biggr) \circ \Biggl(\,\prod_{h \in \mathcal{H}_\mathcal{S}} \lambda_h \Biggr),
  \end{equation}
  where the products denote composition. By Propositions~\ref{prop:lambda-properties} and~\ref{prop:phi-properties}, all $\lambda$ operators commute internally, and all $\Phi$ operators commute internally, but $\Phi$ and $\lambda$ do not necessarily commute between each other; $\lambda$ should be applied first.
\end{definition}

\begin{remark}\label{rem:closure-algorithm}
  Equation~\eqref{eq:closure-map} can be implemented as a simple algorithm: (1)~construct the Kirkwood approximation $\kappa$ or $\tilde{\kappa}$; (2)~for each hyperedge $h \subseteq \mathcal{S}$, tag the corresponding triplet with its $\bigcirc \triangle$-class annotation via $\lambda_h$; (3)~for each absent edge $e \notin \mathcal{E}_\mathcal{S}$, apply the independence transformation rules from Definition~\ref{def:independence-operator}. The commutativity and idempotency properties guarantee that the result is unique and independent of the internal ordering of hyperedge processing or edge processing.
\end{remark}

\subsubsection{Closure of rate equations}\label{subsubsec:closure-examples}

Since the closure depends on the local network topology, we make this dependence explicit by inserting factors of $1 = A_{xy} + \bar{A}_{xy}$ for each pair and $1 = H_{xyz} + \bar{H}_{xyz}$ for each triplet within the group.  For a quadruplet probability, this gives
\begin{equation}\label{eq:closure-4-expand}
  \prob{X_i X_j X_k X_l} = \Biggl( \prod_{\set{x,y} \in \binom{\set{i,j,k,l}}{2}} \bigl[A_{xy} + \bar{A}_{xy}\bigr] \Biggr) \Biggl( \prod_{\set{x,y,z} \in \binom{\set{i,j,k,l}}{3}} \bigl[H_{xyz} + \bar{H}_{xyz}\bigr] \Biggr) \prob{X_i X_j X_k X_l}.
\end{equation}
Each factor acts as a filter to select whether a given edge or hyperedge is present or absent.  Expanding the products yields $2^{\binom{4}{2}} \times 2^{\binom{4}{3}}$ terms, one for each possible topology on the four nodes. For each hyperedge that is present, we apply the hyperedge operator~$\lambda$; for each edge that is absent, we apply the independence operator~$\Phi$.  Replacing the joint probability by its Kirkwood approximation gives the closure:
\begin{multline}\label{eq:closure-4}
  \prob{X_i X_j X_k X_l} \approx \Biggl( \prod_{\set{x,y} \in \binom{\set{i,j,k,l}}{2}} \bigl[A_{xy} + \bar{A}_{xy}\,\Phi_{xy}\bigr] \Biggr) \\
  \times \Biggl( \prod_{\set{x,y,z} \in \binom{\set{i,j,k,l}}{3}} \bigl[H_{xyz}\,\lambda_{xyz} + \bar{H}_{xyz}\bigr] \Biggr) \bigl[\kappa(\set{X_i,X_j,X_k,X_l})\bigr].
\end{multline}
For quintuplet probabilities, the same construction applies with the iterated Kirkwood approximation:
\begin{multline}\label{eq:closure-5}
  \prob{X_i X_j X_k X_l X_m} \approx \Biggl( \prod_{\set{x,y} \in \binom{\set{i,j,k,l,m}}{2}} \bigl[A_{xy} + \bar{A}_{xy}\,\Phi_{xy}\bigr] \Biggr) \\
  \times \Biggl( \prod_{\set{x,y,z} \in \binom{\set{i,j,k,l,m}}{3}} \bigl[H_{xyz}\,\lambda_{xyz} + \bar{H}_{xyz}\bigr] \Biggr) \bigl[\tilde{\kappa}(\set{X_i,X_j,X_k,X_l,X_m})\bigr].
\end{multline}
The rate equations are then closed by substituting equations~\eqref{eq:closure-4} and~\eqref{eq:closure-5} into the quadruplet and quintuplet terms of~\eqref{eq:dSI}--\eqref{eq:dSII}.  In practice, when the left-hand side of a rate equation fixes certain adjacency values (e.g.\ $A_{ij} = 1$ or $H_{ijk} = 1$), many terms in the product vanish, substantially reducing the number of surviving configurations.  Further simplification may arise from structural constraints on the network; for instance, if all hyperedges are simplices ($H_{ijk} = 1 \implies A_{ij} = A_{jk} = A_{ik} = 1$), the number of non-trivial closures is reduced further.

We illustrate the closure on a quadruplet configuration.


\textit{Example 1: Hyperedge with partial edges.}  Let $A_{ij} = A_{ik} = A_{kl} = 1$, $A_{jk} = A_{il} = A_{jl} = 0$, and $H_{ijk} = 1$ (no other hyperedges). Examining all the triples:
\begin{itemize}
    \item $\set{i,j,k}$: Part of a hyperedge and wedge with center $i$ $\rightarrow \bigcirc \wedge$.
    \item $\set{i,k,l}$: wedge with $k$ in the center $\rightarrow \otimes \wedge$.
    \item $\set{j,k,l}$: edge with node $j$ isolated from $k,l$ $\rightarrow$ $\prob{X_k X_l} \prob{X_j}$.
    \item $\set{i,j,l}$: edge with node $l$ isolated from $i,j$ $\rightarrow$ $\prob{X_i X_j} \prob{X_l}$.
\end{itemize}
After factorizing existing pairs and cancelling terms, we obtain:
\begin{equation}\label{eq:closure-example-2-result}
    \mathcal{C}_{\mathcal{S, \bar{E}, H}}\bigl[\kappa(\set{X_i,X_j,X_k,X_l})\bigr]
  = \frac{\prob{X_j {X}_i X_k}^{\bigcirc\wedge}\,\prob{X_i{X}_k X_l}^{\otimes\wedge}}{\prob{X_i X_k}}.
\end{equation}
The hyperedge triplet survives as $\bigcirc\wedge$ despite the absent edge $\set{j,k}$, while the non-hyperedge triplet $\set{i,k,l}$ survives as $\otimes\wedge$ because it retains two edges.  The remaining triplets decomposed and canceled with pair and single factors.
\begin{figure}[h]
    \centering
    \includegraphics[width=0.18\linewidth]{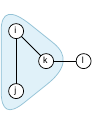}
    \caption{Example of 4 nodes with a hyperedge $\set{i,j,k}$ and edges $\set{i,j}$, $\set{i,k}$ and $\set{k,l}$. }
    \label{fig:4-node-example}
\end{figure}

\section{Canonical closure approximations} \label{sec:canonical-closures}
We derive several canonical closure approximations commonly used in the existing literature using our closure operator framework. The key insight is that closure expressions in the literature are often limited to well-known or simple motifs, but our closure operator framework permits a closure of \emph{any} motif configuration provided its topological structure is known.

\subsection{Triplet closure} \label{subsec:triplet-closure}
While our general framework described by equations~\eqref{eq:dI}--\eqref{eq:dSII} explicitly tracks triplets in all possible configurations (see Table~\ref{tab:isomorphism-class-notation}), many existing models close wedges at the level of pairs. We derive the closures of triangles and wedges in this section for reference.

\begin{lemma} \label{lem:triplet-closure}
    When closing at the level of pairs, the joint probability of a \emph{connected} triplet $\prob{X_i X_j X_k}$ can be approximated as
    \begin{equation} \label{eq:triplet-closure-lemma}
        \prob{X_i X_j X_k} \approx 
        \begin{cases}
            \dfrac{\prob{X_i X_j} \prob{X_j X_k} \prob{X_i X_k}}{\prob{X_i} \prob{X_j} \prob{X_k}} 
            \qquad &A_{ij} = A_{jk} = A_{ik} = 1 \textrm{ (triangle)},
            \\[10pt]
            \dfrac{\prob{X_i X_j} \prob{X_j X_k}}{\prob{X_j}} 
            \qquad &A_{ij} = A_{jk} = 1,\; \bar{A}_{ik} = 1 \text{ (wedge with center $j$)},
        \end{cases}
    \end{equation}
    where, without loss of generality, the edge $\set{ik}$ is taken as the variable edge.
\end{lemma}

\begin{proof}
    The triangle closure is obtained directly from the Kirkwood approximation:
    \[
        \kappa(\set{X_i, X_j, X_k}) = \frac{\prob{X_i X_j} \prob{X_j X_k} \prob{X_i X_k}}{\prob{X_i} \prob{X_j} \prob{X_k}}.
    \]
    The wedge closure follows by applying the independence operator $\Phi_{ik}$ to the triangle expression:
    \[
        \Phi_{ik} \Bigg[ \frac{\prob{X_i X_j} \prob{X_j X_k}\prob{X_i X_k}}{\prob{X_i} \prob{X_j} \prob{X_k}} \Bigg] = \frac{\prob{X_i X_j} \prob{X_j X_k} \prob{X_i}\prob{X_k}}{\prob{X_i} \prob{X_j} \prob{X_k}} = \frac{\prob{X_i X_j} \prob{X_j X_k}}{\prob{X_j}}.
    \]
\end{proof}

\subsection{Pendant closure} \label{subsec:pendant-closure}
A common four-node closure is the \emph{pendant} (or \emph{wine-glass}) closure. The motif consists of a triplet core $\set{i,j,k}$---either a triangle or simplex---connected to an external node~$l$ via a single pairwise edge from one of the core nodes.

\begin{lemma} \label{lem:pendant-closure}
    Consider the pendant motif on nodes $i,j,k,l$, where the core $\set{i,j,k}$ forms a triangle or simplex and the external node~$l$ is connected to the core via the edge $\set{kl}$ only, with $\bar{A}_{il} = \bar{A}_{jl} = 1$. The triadic closure approximation is
    \begin{equation} \label{eq:pendant-closure}
        \prob{X_i X_j X_k X_l} \approx \frac{\prob{X_i X_j X_k}^\sigma \prob{X_i X_k X_l}^{\otimes \wedge} \prob{X_j X_k X_l}^{\otimes \wedge}}{\prob{X_i X_k} \prob{X_j X_k} \prob{X_k X_l}} \prob{X_k},
    \end{equation}
    for $\sigma \in \set{\bigcirc\triangle, \otimes \triangle}$, where the choice of pivotal node~$k$ and attachment edge $\set{kl}$ is without loss of generality.
\end{lemma}
\begin{proof}
    Applying the Kirkwood approximation to the four-node joint probability yields
    \begin{multline*}
        \kappa(\set{X_i, X_j, X_k, X_l}) = 
        \frac{
        \prob{X_i X_j X_k}^{\otimes \triangle}\,
        \prob{X_i X_j X_l}^{\otimes \triangle}\,
        \prob{X_i X_k X_l}^{\otimes \triangle}\,
        \prob{X_j X_k X_l}^{\otimes \triangle}
        }{
        \prob{X_i X_j}\,
        \prob{X_i X_k}\,
        \prob{X_i X_l}\,
        \prob{X_j X_k}\,
        \prob{X_j X_l}\,
        \prob{X_k X_l}
        } \\
        \times\;\prob{X_i}\, \prob{X_j}\, \prob{X_k}\, \prob{X_l}.
    \end{multline*}
    Since the absent edges $\set{il}$ and $\set{jl}$ do not involve any pair within the core $\set{i,j,k}$, the independence operators $\Phi_{il}$ and $\Phi_{jl}$ do not act on $\prob{X_i X_j X_k}^{\otimes\triangle}$ and therefore commute with the hyperedge operator $\lambda_{\set{ijk}}$. We apply the independence operators first.

    The triplet $\prob{X_i X_j X_l}^{\otimes \triangle}$ loses two edges and decomposes: $\prob{X_i X_j X_l}^{\otimes \triangle} \xrightarrow{\Phi_{il}\,\Phi_{jl}} \prob{X_i X_j}\,\prob{X_l}$. The triplets $\prob{X_i X_k X_l}^{\otimes \triangle}$ and $\prob{X_j X_k X_l}^{\otimes \triangle}$ each lose one edge and are demoted to wedges centered at~$k$. The pairs $\prob{X_i X_l}$ and $\prob{X_j X_l}$ decompose into products of singles. After cancellation, the expression reduces to
    \[
        \frac{
        \prob{X_i X_j X_k}^{\otimes \triangle}\, 
        \prob{X_i X_k X_l}^{\otimes \wedge}\, 
        \prob{X_j X_k X_l}^{\otimes \wedge}
        }{
        \prob{X_i X_k}\, \prob{X_j X_k}\, \prob{X_k X_l}
        }\, \prob{X_k}.
    \]
    If the hyperedge $\set{i,j,k}$ is present, applying the hyperedge operator $\lambda_{\set{ijk}}$ tags the core triplet: $\prob{X_i X_j X_k}^{\otimes \triangle} \to \prob{X_i X_j X_k}^{\bigcirc \triangle}$. Since $\lambda$ acts only on this factor, the rest of the expression is unchanged, giving the result for both values of~$\sigma$.
\end{proof}

\begin{remark} \label{rem:pendant-with-wedge-closure}
    When used in conjunction with the wedge closure (Lemma~\ref{lem:triplet-closure}), the pendant closure~\eqref{eq:pendant-closure} further reduces to 
    \[
        \prob{X_i X_j X_k X_l} \approx \frac{\prob{X_i X_j X_k}^\sigma \prob{X_k X_l}}{\prob{X_k}}
    \]
    for $\sigma \in \set{\bigcirc \triangle, \otimes \triangle}$. 
\end{remark}

\begin{remark} \label{rem:pendant-hyperedge-only}
    When the core $\set{i,j,k}$ carries a hyperedge but no pairwise edges ($H_{ijk} = 1$, $\bar{A}_{ij} = \bar{A}_{ik} = \bar{A}_{jk} = 1$), the independence operators $\Phi_{ij}$, $\Phi_{ik}$, and $\Phi_{jk}$ demote the core triplet through the chain $\bigcirc\triangle \to \bigcirc\wedge \to \bigcirc{-} \to \bigcirc\bullet$ without decomposing it (Definition~\ref{def:independence-operator}), while the cross-triplets $\prob{X_i X_k X_l}$ and $\prob{X_j X_k X_l}$ fully decompose. The pendant closure then reduces to
    \[
        \prob{X_i X_j X_k X_l} \approx \frac{\prob{X_i X_j X_k}^{\bigcirc \bullet}\, \prob{X_k X_l}}{\prob{X_k}}.
    \]
    This has the same form as the wedge-reduced pendant closure (Remark~\ref{rem:pendant-with-wedge-closure}), confirming that the pendant structure $\prob{X_i X_j X_k}^\sigma\,\prob{X_k X_l}/\prob{X_k}$ is preserved for $\sigma \in \set{\bigcirc\triangle,\, \otimes\triangle,\, \bigcirc\bullet}$.
\end{remark}

\subsection{Bowtie closure} \label{subsec:bowtie-closure}
Another common closure involving two overlapping triplets is the \emph{bowtie} closure. The motif consists of two triplet cores $\set{i,j,k}$ and $\set{k,l,m}$ (each a triangle or simplex) sharing a single pivotal node~$k$.

\begin{lemma} \label{lem:bowtie-closure}
    Consider the bowtie motif on nodes $i,j,k,l,m$, where $\set{i,j,k}$ and $\set{k,l,m}$ each form a triangle or simplex, sharing only the pivotal node~$k$, with $\bar{A}_{il} = \bar{A}_{im} = \bar{A}_{jl} = \bar{A}_{jm} = 1$. The triadic closure approximation is
    \begin{equation} \label{eq:bowtie-closure}
    \begin{aligned}
        \prob{X_i X_j X_k X_l X_m} \approx{}& 
        \frac{
        \prob{X_i X_j X_k}^\sigma\, \prob{X_k X_l X_m}^{\sigma'}
        }{
        \prob{X_i X_k}^2\,
        \prob{X_j X_k}^2\,
        \prob{X_k X_l}^2\,
        \prob{X_k X_m}^2
        } \\
        &\times\;
        \prob{X_i X_k X_l}^{\otimes \wedge}\,
        \prob{X_i X_k X_m}^{\otimes \wedge}\,
        \prob{X_j X_k X_l}^{\otimes \wedge}\,
        \prob{X_j X_k X_m}^{\otimes \wedge}\,
        \prob{X_k}^3,
    \end{aligned}
    \end{equation}
    for $\sigma, \sigma' \in \set{\bigcirc \triangle, \otimes \triangle}$.
\end{lemma}
\begin{proof}
    Applying the truncated Kirkwood approximation $\tilde{\kappa}$ yields
    \[
        \tilde{\kappa}(\set{X_i, X_j, X_k, X_l, X_m}) = 
        \frac{
        \displaystyle\prod_{\set{w,x,y} \in \binom{\set{i,j,k,l,m}}{3}}\prob{X_w X_x X_y}^{\otimes \triangle} \;\cdot\;
        \prod_{x \in \set{i,j,k,l,m}}\prob{X_x}^3
        }{
        \displaystyle\prod_{\set{x,y} \in \binom{\set{i,j,k,l,m}}{2}}\prob{X_x X_y}^2
        }.
    \]
    The independence operators $\Phi_{il}$, $\Phi_{im}$, $\Phi_{jl}$, $\Phi_{jm}$ act on the four absent cross-edges. Since none of these edges appear within the core triplets $\set{i,j,k}$ or $\set{k,l,m}$, these operators commute with the hyperedge operators $\lambda_{\set{ijk}}$ and $\lambda_{\set{klm}}$; we may therefore apply them first.

    \textit{Action on pair terms.} The four absent-edge pairs decompose into products of singles: $\prob{X_i X_l},\, \prob{X_i X_m},\, \prob{X_j X_l},\, \prob{X_j X_m} \mapsto \prob{X_i}\prob{X_l},\, \prob{X_i}\prob{X_m},\, \prob{X_j}\prob{X_l},\, \prob{X_j}\prob{X_m}$.

    \textit{Action on triplet terms.} Among the $\binom{5}{3} = 10$ triplets, the two core triplets $\set{i,j,k}$ and $\set{k,l,m}$ are unaffected. Four triplets each lose one cross-edge and are demoted to wedges centered at~$k$: for example, $\prob{X_i X_k X_l}^{\otimes \triangle} \xrightarrow{\Phi_{il}} \prob{X_i X_k X_l}^{\otimes \wedge}$. The remaining four triplets each lose two cross-edges, disconnecting one node, and fully decompose: for example, $\prob{X_i X_l X_m}^{\otimes \triangle} \xrightarrow{\Phi_{il}\,\Phi_{im}} \prob{X_l X_m}\,\prob{X_i}$. Canceling the decomposed pair and single factors between numerator and denominator, and applying the hyperedge operators to the core triplets, yields the stated expression. 
\end{proof}

\begin{remark} \label{rem:bowtie-with-wedge-closure}
    When used in conjunction with the wedge closure (Lemma~\ref{lem:triplet-closure}), the bowtie closure~\eqref{eq:bowtie-closure} reduces to
    \[
        \prob{X_i X_j X_k X_l X_m} \approx \frac{\prob{X_i X_j X_k}^{\sigma}\, \prob{X_k X_l X_m}^{\sigma'}}{\prob{X_k}}
    \]
    for $\sigma, \sigma' \in \set{\bigcirc \triangle, \otimes \triangle}$.
\end{remark}

\begin{remark} \label{rem:bowtie-hyperedge-only}
    When both cores $\set{i,j,k}$ and $\set{k,l,m}$ carry hyperedges but no pairwise edges (except those within the respective cores already accounted for), the additional independence operators $\Phi_{ij}$, $\Phi_{ik}$, $\Phi_{jk}$ and $\Phi_{lm}$, $\Phi_{kl}$, $\Phi_{km}$ demote each core triplet through the chain $\bigcirc\triangle \to \bigcirc\wedge \to \bigcirc{-} \to \bigcirc\bullet$ without decomposing them (Definition~\ref{def:independence-operator}), while the four cross-wedges fully decompose. The bowtie closure then reduces to
    \[
        \prob{X_i X_j X_k X_l X_m} \approx \frac{\prob{X_i X_j X_k}^{\bigcirc \bullet}\, \prob{X_k X_l X_m}^{\bigcirc \bullet}}{\prob{X_k}}.
    \]
    This has the same form as the wedge-reduced bowtie closure (Remark~\ref{rem:bowtie-with-wedge-closure}), confirming that the bowtie structure $\prob{X_i X_j X_k}^\sigma\,\prob{X_k X_l X_m}^{\sigma'}/\prob{X_k}$ is preserved for $\sigma, \sigma' \in \set{\bigcirc\triangle,\, \otimes\triangle,\, \bigcirc\bullet}$.
\end{remark}

\section{Mean field approximations} \label{sec:mean-field-approx}

The microscopic rate equations derived in Subsection~\ref{subsec:sis-dynamics} depend on specific node labels and on the full adjacency structure of the network. While technically a closed system of equations, on large networks this still leads to impractical computational complexity. To maintain numerical tractability, we introduce two simplifying assumptions: first on node states, then on network topology.

\subsection{Dynamical homogeneity} \label{subsubsec:dyanmical-homogeneity}
As established in Section~\ref{subsec:subgraph-isos}, the isomorphism class of a sub-graph captures the topological information that governs the local dynamics.  The dynamical homogeneity assumption utilizes this to assert that the dynamical evolution, described by the state joint probabilities, is approximately equal for sub-graphs that are isomorphic.

\begin{assumption} \label{ass:dynamical-homogeneity}
    The joint state probability over any group of nodes depends only on the isomorphism class of the sub-graph they span, not on the specific node labels.  That is,
  \[
  \begin{aligned}
  \prob{X_i } 
  &\approx \prob{X}\\
  \prob{X_i X_j}^{-} 
  &\approx \prob{X X}^{-}\\
          \prob{X_i X_j X_k}^{\sigma\tau}
    &\approx \prob{X X X}^{\sigma\tau}
  \end{aligned}
  \]
  for all $\sigma \in \set{\bigcirc, \otimes}$,  $\tau \in \set{\triangle, \wedge, -, \bullet}$ and $X \in \set{S,I}$, where the right-hand side denotes a common or average value shared by all sub-graphs in the same class.
\end{assumption}
In effect, this assumption replaces node-specific probabilities with class-specific averages, eliminating node labels from the dynamical variables; the only information about the network topology that is retained is the isomorphism class associated with each joint probability.  An unspecified state variable can be expanded over its isomorphism classes using adjacency prefactors.  For pairs,
\[
  \prob{X_i X_j} = A_{ij}\,\prob{X_i X_j}^{-} + \bar{A}_{ij}\,\prob{X_i X_j}^{\,\bullet}.
\]
Since $A_{ij}$ and $\bar{A}_{ij} = 1 - A_{ij}$ are complementary matrices, exactly one term is nonzero for any given pair: the prefactors select the unique isomorphism class that the pair belongs to.  The same principle extends to triplets, where the product of $H$, $\bar{H}$, $A$, $\bar{A}$ prefactors partitions the eight isomorphism classes of Table~\ref{tab:isomorphism-class-notation}, with exactly one term surviving for any given triple.

Applying this expansion to equation~\eqref{eq:dI} gives
    \begin{equation}
    \begin{aligned}
        \frac{d}{dt}\prob{I_i} =& {-\gamma \prob{I_i}} + {\beta_1 \sum_{x\not=i} A_{ix} \prob{S_i I_x}} \ldots\\ 
        & \downarrow \\
        \frac{d}{dt}\prob{I_i} \approx& -\gamma \prob{I} + \beta_1 \underbrace{\sum_{x\not=i} A_{ix} 
        \big( A_{ix} \prob{S I }^- + \bar{A}_{ix} \prob{S I }^\bullet \big) }_{= \prob{S I}^- \sum_{x\not=i}A_{ix}} + \ldots 
    \end{aligned}
    \end{equation}
where the $\bar{A}_{ix}$ term vanishes because $A_{ix}\,\bar{A}_{ix} = 0$, leaving only the connected-class contribution.  The same logic extends to the $\beta_2$ infection term via an expansion similar to that described in equation~\eqref{eq:closure-4}.  After dynamical homogeneity, the only node-specific quantities remaining on the right-hand side are purely topological: sums of adjacency products such as $\sum_{x\not=i} A_{ix}$, the degree of node~$i$.  To decouple the dynamics from node labels entirely, we use the topological average of these quantities.

\subsection{Topological homogeneity}\label{subsubsec:topological-homogeneity}

After applying dynamical homogeneity, the remaining node dependence takes the form of partial sums of adjacency tensors with one or more nodes fixed.  For example, the sum $A_{ij}\sum_{x \neq i,j} H_{ijx}$ counts the number of hyperedges that the edge $\set{i,j}$ belongs to.  Generally, this count is not guaranteed to be the same at every edge for an arbitrary network.  Topological homogeneity assumes that the local count is equivalent to the global average: the total number of edge-in-hyperedge occurrences, divided by the total number of edges.  Concretely,
\[
  A_{ij} \sum_{\substack{x \neq i,j}} H_{ijx}
  \approx
  \frac{\displaystyle\sum_{\substack{i,j,x \\ \text{distinct}}} A_{ij}\, H_{ijx}}
       {\displaystyle\sum_{\substack{i,j \\ \text{distinct}}} A_{ij}}.
\]

The numerator and denominator, in such cases, are typically well-defined quantities, often related to network parameters (or are defined as network parameters themselves), allowing for analytic solutions. More generally, for any sub-graph $\mathcal{S}'$ that can be embedded in a larger sub-graph $\mathcal{S}$, we assume that the number of instances of $\mathcal{S}$ containing a given copy of $\mathcal{S}'$ is approximately equal to the network-wide average, namely the total number of (labeled) copies of $\mathcal{S}$ (e.g. hyperedges) divided by the total number of (labeled) copies of $\mathcal{S}'$ (e.g. edges).

\begin{assumption}\label{ass:topological-homogeneity}
  Let $f = p \cdot q$ be a product of adjacency tensors $\mathbf{A}$, $\bar{\mathbf{A}}$, $\mathbf{H}$, $\bar{\mathbf{H}}$, where $p(i, j, \ldots)$ depends only on fixed nodes and $q(i, j, \ldots, x, y, \ldots)$ contains all dependence on the dummy nodes $x, y, \ldots$ (all indices mutually distinct).  The factor $p$ encodes the sub-graph $\mathcal{S}'$ on the fixed nodes, and $p \cdot q$ encodes the larger sub-graph $\mathcal{S}$ on all nodes.  The assumption states that the local membership count is approximately equal to the network-wide average:
  \begin{equation} \label{eq:sum-promotion}
    p(i, j, \ldots)
    \sum_{x, y, \ldots} q(i, j, \ldots, x, y, \ldots)
    \approx
    \frac{
      \displaystyle\sum_{i, j, \ldots, x, y, \ldots}
        p(i, j, \ldots)\, q(i, j, \ldots, x, y, \ldots)
    }{
      \displaystyle\sum_{i, j, \ldots} p(i, j, \ldots)
    }
  \end{equation}
  for all distinct $i, j, \ldots \in \mathcal{N}$.
\end{assumption}

We refer to the right-hand side as the \emph{promoted sum}, since the partial sum over $x, y, \ldots$ with fixed nodes $i, j, \ldots$ has been promoted to a complete sum over all nodes.  The denominator depends on the fixed sub-graph $\mathcal{S}'$: when a single node is fixed, then $p = 1$ and $\sum_{i} 1 = N$; when an edge is fixed with $p = A_{ij}$, then $\sum_{i,j} A_{ij} = Nk_1$; and so on for larger fixed structures.

Recalling the previous example then, we can see the formula being applied:
\[
  \underbrace{A_{ij}}_{p(i,j)} \; \sum_{\substack{x \neq i,j}} \; \underbrace{H_{ijx}}_{q(i,j,x)}
  \approx
  \frac{\displaystyle\sum_{\substack{i,j,x \\ \text{distinct}}} \overbrace{A_{ij}\, H_{ijx}}^{p(i,j) \,q(i,j,x)} }
       {\displaystyle\sum_{\substack{i,j \\ \text{distinct}}} \underbrace{A_{ij}}_{p(i,j)}}.
\]

\begin{remark}\label{rem:topological-homogeneity}
  Assumption~\ref{ass:topological-homogeneity} can be understood as requiring that all nodes are topologically equivalent: every node has the same degree, the same hyperedge degree, and more generally the same count of every isomorphism class of sub-graph in its neighborhood.
\end{remark}

\section{Recovery of existing models} \label{sec:results}
We now demonstrate that several existing models for higher-order contagion can be derived as special cases of our general framework. Each model is recovered by imposing specific structural constraints on the network~$\mathcal{G}$ and utilizing the closure operator from Subsection~\ref{subsec:closure} to obtain a closed system of equations. We treat the following three models:
\begin{enumerate}
  \item the microscopic maximal clique model of Burgio et al.~\cite{burgio2024triadic},
  \item the pair-based simplicial contagion model of Malizia et al.~\cite{malizia2025pair},
  \item the inter-order overlap model of Malizia et al.~\cite{malizia2026nested}.
\end{enumerate}
The first model operates at the microscopic level and requires only structural constraints on the network. The second and third models are mean-field approximations, which additionally require invoking dynamical and topological homogeneity assumptions. 

For each model, we format the respective subsection in a similar manner. We begin by defining the network and model-specific parameters and notation, establishing a connection between the model's quantities and parameters with our adjacency tensors~$\mathbf{A}$ and~$\mathbf{H}$. Next, we state the target equations from the reference paper, translated into our probability notation. Throughout, we use the \emph{internal}, \emph{mixed}, and \emph{external} terminology for pairwise and hyperedge infection terms in the rate equations~\eqref{eq:dI}--\eqref{eq:dSII}, classifying each term by whether its infectious drivers are the fixed nodes, dummy nodes, or both.

Finally, we present the derivation: supporting lemmas are stated as needed, followed by a theorem that recovers the target equations from our general microscopic rate equations~\eqref{eq:dI}--\eqref{eq:dSII}.

\begin{figure}[h!]
    \centering
    \includegraphics[width=0.7\linewidth]{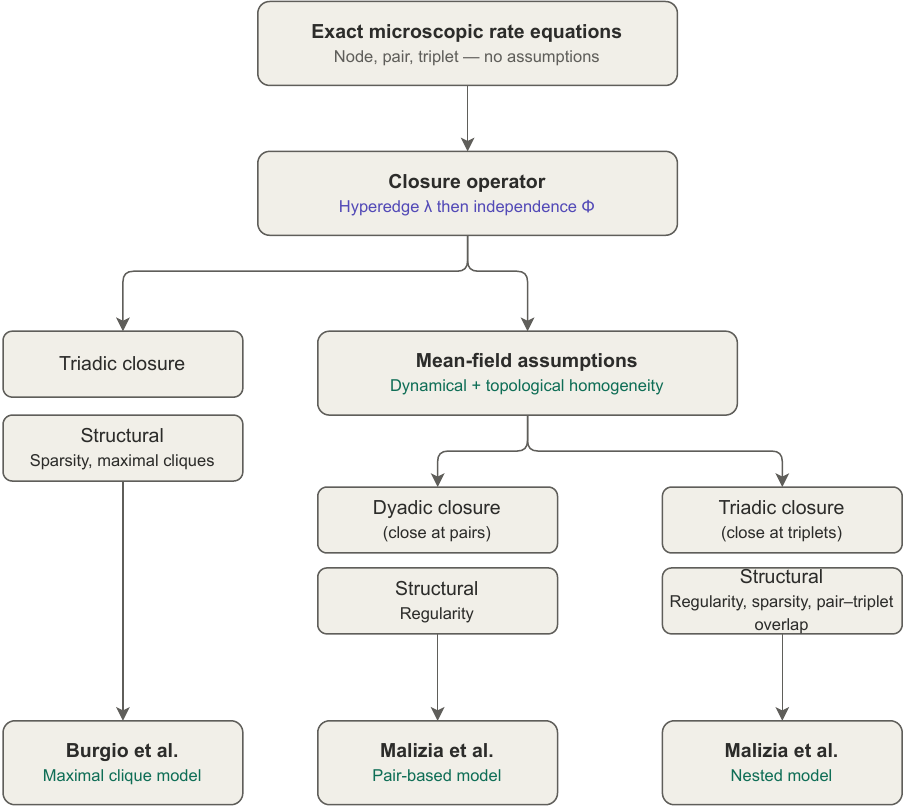}
    \caption{Schematic illustrating the pipeline used in deriving specific higher-order contagion models from the exact microscopic rate equations.}
    \label{fig:derivation-schematic}
\end{figure}

\subsection{Deriving the Burgio maximal clique model} \label{subsec:burgio-derivation}

We show that the microscopic rate equations of the model stated by Burgio et al.~\cite{burgio2024triadic} (hereafter BMC model) can be recovered directly from equations~\eqref{eq:dI}--\eqref{eq:dSII} under structural constraints on the network.

\subsubsection{Model parameters and notation} \label{subsubsec:burgio-notation}
The BMC model is defined on a higher-order network in which the only permitted maximal cliques have size~$2$ or~$3$.  Maximal $2$-cliques are pairwise edges that do not belong to any triangle or hyperedge; maximal $3$-cliques are triangles ($3$~pairwise edges), hyperedges, or simplicial complexes (carrying both a triangle and a hyperedge).  The model further imposes a sparsity constraint: distinct maximal cliques share at most one node.

Rather than the adjacency matrix~$\mathbf{A}$ and hyperedge tensor~$\mathbf{H}$, the BMC model uses tensors that record membership in each type of maximal clique.  Following~\cite{burgio2024triadic}, we define:
\begin{center}
\renewcommand{\arraystretch}{1.3}
\begin{tabular}{ll}
  $B^{(1)}_{ij}$    & maximal $2$-clique: edge $\set{i,j}$ not contained in any $3$-clique \\
  $B^{(1,0)}_{ijk}$ & triangle on $\set{i,j,k}$, with or without a hyperedge \\
  $B^{(0,1)}_{ijk}$ & hyperedge on $\set{i,j,k}$, with or without a triangle \\
  $B^{(1,1)}_{ijk} = B^{(1,0)}_{ijk}\, B^{(0,1)}_{ijk}$ & simplicial complex: both triangle and hyperedge
\end{tabular}
\end{center}
The tensor~$B^{(1,1)}$ is not independent; it is the pointwise product of~$B^{(1,0)}$ and~$B^{(0,1)}$.  These are related to the standard adjacency tensors by
\begin{align}
  A_{ij} &= B^{(1)}_{ij} + \sum_{x} B^{(1,0)}_{ijx}, \label{eq:burgio-A-decomp} \\
  H_{ijk} &= B^{(0,1)}_{ijk}. \label{eq:burgio-H-decomp}
\end{align}
Equation~\eqref{eq:burgio-A-decomp} decomposes each edge into a maximal link contribution and a sum over triangles containing that edge. Equation~\eqref{eq:burgio-H-decomp} identifies the hyperedge tensor directly with~$B^{(0,1)}$.
 
The maximality and sparsity constraints impose algebraic relations among the $B$~tensors that are central to the derivation.
 
\begin{proposition}\label{prop:burgio-constraints}
  The following identities hold for all distinct indices.
 
  \medskip
  \emph{Maximality.}  A maximal $2$-clique edge cannot belong to any $3$-clique:
  \begin{equation}\label{eq:burgio-maximality}
  \begin{aligned}
    B^{(1)}_{ij}\, B^{(1,0)}_{ijk} &= 0, \\
    B^{(1)}_{ij}\, B^{(0,1)}_{ijk} &= 0.
  \end{aligned}
  \end{equation}
 
  \medskip
  \emph{Sparsity.}  Distinct maximal $3$-cliques share at most one node:
  \begin{equation}\label{eq:burgio-sparsity}
  \begin{aligned}
    B^{(1,0)}_{ijk}\, B^{(1,0)}_{ijl} &= B^{(1,0)}_{ijk}\, \delta_{kl}, \\
    B^{(0,1)}_{ijk}\, B^{(0,1)}_{ijl} &= B^{(0,1)}_{ijk}\, \delta_{kl}, \\
    B^{(1,0)}_{ijk}\, B^{(0,1)}_{ijl} &= B^{(1,1)}_{ijk}\, \delta_{kl}.
  \end{aligned}
  \end{equation}
\end{proposition}
 
\begin{proof}
  For maximality: if $B^{(1)}_{ij} = 1$, then $\set{i,j}$ is a maximal $2$-clique and by definition does not appear as an edge in any $3$-clique, so $B^{(\psi)}_{ijk} = 0$ for all $(\psi) \in \set{(1,0),\, (0,1),\, (1,1)}$ and all~$k$.
 
  For sparsity: if $k \neq l$, then $B^{(\psi)}_{ijk}$ and $B^{(\psi')}_{ijl}$ represent two distinct $3$-cliques sharing the pair $\set{i,j}$, violating the sparsity constraint, so the product vanishes.  When $k = l$, the first two identities reduce to $[B^{(\psi)}_{ijk}]^2 = B^{(\psi)}_{ijk}$ since $B^{(\psi)}_{ijk} \in \set{0,1}$.  For the third identity with $k = l$, the product $B^{(1,0)}_{ijk}\, B^{(0,1)}_{ijk} = 1$ exactly when $\set{i,j,k}$ carries both a triangle and a hyperedge, which is the definition of~$B^{(1,1)}_{ijk}$.
\end{proof}
 
We adopt the double-summation convention (Remark~\ref{rem:summation-convention}) throughout this subsection, following~\cite{burgio2024triadic}.

\subsubsection{Target equations} \label{subsubsec:burgio-target-eqs}
The BMC model tracks three levels of dynamical variables: single-node probabilities $\prob{I_i}$, maximal-link pair probabilities $B^{(1)}_{ij}\prob{S_i I_j}$, and maximal $3$-clique triplet probabilities $B^{(\sigma)}_{ijk}\prob{S_i I_j I_k}$ for $\sigma \in \set{(1,0),\, (0,1),\, (1,1)}$.  The rate equations, stated in equations~(2a)--(2c) of~\cite{burgio2024triadic} and translated here into our notation, are as follows.

\begin{equation}\label{eq:burgio-target-order1}
\begin{aligned}
  \frac{d}{dt}\prob{I_i}
  ={}& -\gamma\,\prob{I_i}
  + \beta_1 \sum_{x \neq i} 
    B^{(1)}_{ix}\,\prob{S_i I_x} \\
    &+ \frac{1}{2} \sum_{\substack{x,y \\ y \neq x}}
    \Bigg[ 
      \beta_1 B^{(1,0)}_{ixy}\Big(
        \prob{S_i S_x I_y}
        + \prob{S_i I_x S_y}
        + 2\,\prob{S_i I_x I_y}
      \Big)
      + \beta_2 B^{(0,1)}_{ixy} \prob{S_i I_x I_y}
  \Bigg] 
\end{aligned}
\end{equation}
 
\begin{equation}\label{eq:burgio-target-order2}
\begin{aligned}
\frac{d}{dt}\prob{S_i I_j} &= -(1 + \beta^{(1)}) \prob{S_i I_j} + \prob{I_i I_j}
  - \beta_1 \sum_{x \neq j} B_{ix}^{(1)} \prob{I_j S_i I_x}^{\otimes \wedge}
  + \beta_1 \sum_{x \neq i} B_{jx}^{(1)} \prob{S_i S_j I_x}^{\otimes \wedge} \\
  &\quad - \frac{1}{2} \sum_{x,y} \Big[ B_{ixy}^{(1,0)} \beta_1 \big( \prob{I_j S_i I_x S_y} + \prob{I_j S_i S_x I_y} + 2\prob{I_j S_i I_x I_y} \big)  \\
  &\qquad\qquad + B_{ixy}^{(0,1)} \beta_2 \prob{I_j S_i I_x I_y} \Big]
  + \{i \leftrightarrow j\}
\end{aligned}
\end{equation}

\begin{equation}\label{eq:burgio-target-order3-SSI}
\begin{aligned}
\frac{d}{dt} \prob{S_i S_j I_k} &= -(1 + 2B_{ijk}^{(1,0)} \beta_1) \prob{S_i S_j I_k}
  + \prob{I_i S_j I_k} + \prob{S_i I_j I_k} \\
  &\quad - \beta_1 \sum_{x \neq j,k} B_{ix}^{(1)} \prob{S_i S_j I_k I_x}
  - \beta_1 \sum_{x \neq i,k} B_{jx}^{(1)} \prob{S_i S_j I_k I_x}
  + \beta_1 \sum_{x \neq i,j} B_{kx}^{(1)} \prob{S_i S_j S_k I_x} \\
  &\quad - \frac{1}{2} \sum_{x,y \neq j,k} \Big[ B_{ixy}^{(1,0)} \beta_1 \big( \prob{S_i S_j I_k I_x S_y} + \prob{S_i S_j I_k S_x I_y} + 2\prob{S_i S_j I_k I_x I_y} \big) \\
  &\qquad\qquad + B_{ixy}^{(0,1)} \beta_2 \prob{S_i S_j I_k I_x I_y} \Big]
  - \{i \leftrightarrow j\} + \{i \leftrightarrow k\},
\end{aligned}
\end{equation}

\begin{equation}\label{eq:burgio-target-order3-SII}
\begin{aligned}
\frac{d}{dt} \prob{S_i I_j I_k} &= -(2 + 2B_{ijk}^{(1,0)} \beta_1 + B_{ijk}^{(0,1)} \beta_2) \prob{S_i I_j I_k}
  + B_{ijk}^{(1,0)} \beta_1 (\prob{S_i S_j I_k} + \prob{S_i I_j S_k}) + \prob{I_i I_j I_k} \\
  &\quad - \beta_1 \sum_{x \neq j,k} B_{ix}^{(1)} \prob{S_i I_j I_k I_x}
  + \beta_1 \sum_{x \neq i,k} B_{jx}^{(1)} \prob{S_i S_j I_k I_x}
  + \beta_1 \sum_{x \neq i,j} B_{kx}^{(1)} \prob{S_i I_j S_k I_x} \\
  &\quad - \frac{1}{2} \sum_{x,y \neq j,k} \Big[ B_{ixy}^{(1,0)} \beta_1 \big( \prob{I_j I_k S_i I_x S_y} + \prob{I_j I_k S_i S_x I_y} + 2\prob{I_j I_k S_i I_x I_y} \big) \\
  &\qquad\qquad + B_{ixy}^{(0,1)} \beta_2 \prob{I_j I_k S_i I_x I_y} \Big]
  + \{i \leftrightarrow j\} + \{i \leftrightarrow k\},
\end{aligned}
\end{equation}
where the braces $\set{i \leftrightarrow j}$ denote the terms obtained by taking the term in square brackets $[ \cdots  ]$ (including the summation and factors but excluding the sign), swapping the indices $i$ and $j$ and fixing both nodes in state $S$. 

The BMC model, as stated in~\cite{burgio2024triadic}, closes the hierarchy using the following approximations. Translated into our notation, these are:
\begin{subequations}\label{eq:burgio-target-closures}
\begin{align}
  \prob{X_i X_j X_l}^{\otimes \wedge}
  &\approx \frac{\prob{X_i X_j}\,\prob{X_j X_l}}{\prob{X_j}},
  \label{eq:burgio-closure-triple} \\[4pt]
  \prob{X_i X_j X_l X_h}
  &\approx \frac{\prob{X_j X_l X_h}^{\sigma}\,\prob{X_i X_j}}{\prob{X_j}},
  \label{eq:burgio-closure-quad} \\[4pt]
  \prob{X_i X_j X_l X_h X_k}
  &\approx \frac{\prob{X_i X_j X_l}^{\sigma}\,\prob{X_l X_h X_k}^{\sigma'}}{\prob{X_l}},
  \label{eq:burgio-closure-quint}
\end{align}
\end{subequations}
where in~\eqref{eq:burgio-closure-triple} the open triplet is a wedge with center~$j$; in~\eqref{eq:burgio-closure-quad} the triplet $\set{j,l,h}$ is a maximal $3$-clique connected to the external node~$i$ via the edge $\set{ij}$; and in~\eqref{eq:burgio-closure-quint} the two maximal $3$-cliques $\set{i,j,l}$ and $\set{l,h,k}$ share the pivotal node~$l$.

\subsubsection{Derivation} \label{subsubsec:burgio-derivation}
Since the BMC model is specified at the microscopic level---node labels are retained throughout the target equations---no mean-field approximations are required, and the derivation proceeds entirely through the algebraic manipulation of the adjacency tensors.

The key algebraic step in recovering the BMC equations is substituting the pairwise adjacency matrix with the $B$ tensors (see equation~\eqref{eq:burgio-A-decomp}) into the rate equations. Since the right-hand side of~\eqref{eq:burgio-A-decomp} introduces a three-index tensor, the joint probability must be expanded to accommodate the additional node.
 
\begin{lemma}\label{lem:burgio-substitution}
  For any joint probability $\prob{\ldots S_i I_x \ldots}$ appearing in a sum over~$x$,
  \begin{equation}\label{eq:burgio-substitution}
  \begin{aligned}
    \sum_{x \neq i} A_{ix}\,\prob{\ldots S_i I_x \ldots}
    = \sum_{x \neq i} \Bigg[
      &B^{(1)}_{ix}\,\prob{\ldots S_i I_x \ldots} \\
      &+ \sum_{\substack{y \not=i \\ y \neq x}} B^{(1,0)}_{ixy}\,\frac{1}{2}\Big(
        \prob{\ldots S_i I_x S_y \ldots}
        + \prob{\ldots S_i S_x I_y \ldots}
        + 2\,\prob{\ldots S_i I_x I_y \ldots}
      \Big)
    \Bigg].
  \end{aligned}
  \end{equation}
\end{lemma}
 
\begin{proof}
  Substitute equation~\eqref{eq:burgio-A-decomp} to obtain a $B^{(1)}_{ix}$ term and a $\sum_y B^{(1,0)}_{ixy}$ term.  For the latter, the joint probability must be expanded to include node~$y$ by marginalizing:
  \[
    \prob{\ldots S_i I_x \ldots}
    = \prob{\ldots S_i I_x S_y \ldots}
    + \prob{\ldots S_i I_x I_y \ldots}.
  \]
  Since $x$ and $y$ are both dummy indices summed with symmetric coefficient $B^{(1,0)}_{ixy} = B^{(1,0)}_{iyx}$, we symmetrize:
  \[
    \sum_{\substack{x,y \\ y \neq x}} B^{(1,0)}_{ixy}\,\prob{\ldots S_i I_x S_y \ldots}
    = \sum_{\substack{x,y \\ y \neq x}} B^{(1,0)}_{ixy}\,\frac{1}{2}\Big(
      \prob{\ldots S_i I_x S_y \ldots}
      + \prob{\ldots S_i S_x I_y \ldots}
    \Big).
  \]
  The $\prob{\ldots S_i I_x I_y \ldots}$ term is already symmetric under $x \leftrightarrow y$, contributing $\frac{1}{2} \cdot 2 = 1$.  Combining gives~\eqref{eq:burgio-substitution}.
\end{proof}

\begin{theorem}\label{thm:burgio-order1}
  Under the BMC structural constraints, equation~\eqref{eq:dI} reduces to equation~\eqref{eq:burgio-target-order1}.
\end{theorem}
 
\begin{proof}
  Substituting equations~\eqref{eq:burgio-A-decomp}--\eqref{eq:burgio-H-decomp} into the first-moment equation~\eqref{eq:dI} and applying Lemma~\ref{lem:burgio-substitution} to the $\beta_1$ term gives
  \begin{equation}\label{eq:burgio-order1-derivation}
  \begin{aligned}
    \frac{d}{dt}\prob{I_i}
    =&
      \underbracket{-\gamma\,\prob{I_i}}{recovery}
      + \underbracket{
        \beta_1 \sum_{x \neq i} A_{ix}\,\prob{S_i I_x}
      }{external pairwise infection}
      + \underbracket{
        \beta_2 \sum_{\substack{x < y \\ x,y \neq i}}
          H_{ixy}\,\prob{S_i I_x I_y}
      }{external hyperedge infection} \\
    =&
      \underbracket{-\gamma\,\prob{I_i}}{recovery}
      + \underbracket{\beta_1 \sum_{x \neq i} \Bigg[
        B^{(1)}_{ix}\,\prob{S_i I_x}
        + \sum_{\substack{y \not= i \\ y \neq x}}
          B^{(1,0)}_{ixy}\,\frac{1}{2}\Big(
            \prob{S_i I_x S_y}
            + \prob{S_i S_x I_y}
            + 2\,\prob{S_i I_x I_y}
          \Big)
      \Bigg]}{external pairwise infection} \\
      &+ \underbracket{\frac{\beta_2}{2}
        \sum_{\substack{x,y \\ y \neq x}}
          B^{(0,1)}_{ixy}\,\prob{S_i I_x I_y}}{external hyperedge infection}.
  \end{aligned}
  \end{equation}
  The recovery term is unchanged. The external pairwise infection term splits into maximal links ($B^{(1)}$) and triangle contributions ($B^{(1,0)}$) via Lemma~\ref{lem:burgio-substitution}.  The external hyperedge infection term follows directly from equation~\eqref{eq:burgio-H-decomp} and the double-counting convention $\sum_{x < y} \to \frac{1}{2}\sum_{\substack{x,y \\ y \neq x}}$.  This matches equation~\eqref{eq:burgio-target-order1}.
\end{proof}
 
\begin{theorem}\label{thm:burgio-order2}
  Under the BMC structural constraints, the pair equation~\eqref{eq:dSI} multiplied by $B^{(1)}_{ij}$ reduces to equation~\eqref{eq:burgio-target-order2}.
\end{theorem}

\begin{proof}
  The BMC model tracks only maximal $2$-cliques, not all edges. To restrict to this case, we multiply both sides of equation~\eqref{eq:dSI} by~$B^{(1)}_{ij}$. By Proposition~\ref{prop:burgio-constraints}, $B^{(1)}_{ij}\, A_{ij} = (B^{(1)}_{ij})^2 = B^{(1)}_{ij}$, so the prefactor simplifies.
 
  The recovery, internal pairwise infection, and external hyperedge infection terms pass through directly under Proposition~\ref{prop:burgio-constraints}. The mixed hyperedge infection term contains the prefactor $B^{(1)}_{ij}\, B^{(0,1)}_{ijx}$, which vanishes by the maximality constraint~\eqref{eq:burgio-maximality}: a maximal $2$-clique cannot be contained in a $3$-clique. The external pairwise infection terms are the only contributions requiring expansion; applying Lemma~\ref{lem:burgio-substitution} to the factor $B^{(1)}_{ij}\,A_{ix}$ splits them into a maximal-edge contribution and a maximal-triangle contribution, recovering equation~\eqref{eq:burgio-target-order2} after the symmetry substitution $[I_j \to S_j,\; i \leftrightarrow j]$ is applied to the $j$-neighbor term. The full term-by-term derivation is given in~Appendix~\ref{app:bmc-order-2-derivation}.
\end{proof}

\begin{theorem}\label{thm:burgio-order3}
  Under the BMC structural constraints, the triplet equations~\eqref{eq:dSSI} 
  and~\eqref{eq:dSII} multiplied by $B^{(2)}_{ijk} = B^{(1,0)}_{ijk} + B^{(0,1)}_{ijk} - B^{(1,1)}_{ijk}$ reduce to 
  equations~\eqref{eq:burgio-target-order3-SSI} 
  and~\eqref{eq:burgio-target-order3-SII}.
\end{theorem}

\begin{proof}
  The recovery, internal pairwise infection, internal hyperedge infection (which appears in equation~\eqref{eq:dSII} but not equation~\eqref{eq:dSSI}), and external hyperedge infection terms simplify directly under Proposition~\ref{prop:burgio-constraints}, following the same reasoning as in the proof for the pair equation (Theorem~\ref{thm:burgio-order2}). The mixed hyperedge infection term vanishes by sparsity: the prefactor $B^{(2)}_{ijk}\, B^{(0,1)}_{ikx}$ requires two $3$-cliques to share the pair $\set{i,k}$, which forces $x = j$, but $x \neq j$ in the summation.  The external pairwise infection terms are expanded via Lemma~\ref{lem:burgio-substitution}.  The full term-by-term derivation is given in~Appendix~\ref{app:bmc-order-3-derivation}.
\end{proof}

\begin{remark}\label{rem:burgio-closure}
  The BMC target equations~\eqref{eq:burgio-target-order1}--\eqref{eq:burgio-target-order3-SII} are not closed: quadruplet and quintuplet probabilities remain on the right-hand side. The closures~\eqref{eq:burgio-target-closures} used by Burgio et al.\ are recovered directly from the canonical closure approximations of Section~\ref{sec:canonical-closures}.

  The triplet closure~\eqref{eq:burgio-closure-triple} is the wedge closure of Lemma~\ref{lem:triplet-closure} with center~$j$: the sparsity constraint ensures that open triplets in the BMC model are always wedges, since the third edge connecting nodes in distinct maximal cliques is absent.

  The quadruplet closure~\eqref{eq:burgio-closure-quad} is the pendant closure with wedge reduction (Remark~\ref{rem:pendant-with-wedge-closure}): the sub-graph consists of a maximal $3$-clique $\set{j,l,h}$ connected to an external node~$i$ via the edge $\set{ij}$. By sparsity, no other edges connect $i$ to $\set{l,h}$, giving exactly the pendant configuration with pivotal node~$j$.

  The quintuplet closure~\eqref{eq:burgio-closure-quint} is the bowtie closure with wedge reduction (Remark~\ref{rem:bowtie-with-wedge-closure}): two maximal $3$-cliques $\set{i,j,l}$ and $\set{l,h,k}$ share the pivotal node~$l$. By sparsity, no cross-edges exist between the two cliques except through~$l$, giving the bowtie configuration.
\end{remark}

\begin{remark}
  To pass to a population-level description, Burgio et al.\ introduced averaged quantities such as $\prob{X} = \frac{1}{N}\sum_i \prob{X_i}$, with analogous averages defining the pair and triplet dynamics. The population rate equations are then obtained by differentiating these averages and substituting the microscopic rate equations. We note that this procedure does not close on its own. After closure the right-hand side contains nonlinear expressions in the node-indexed probabilities --- for example $\prob{S_i I_j}\,\prob{S_i I_k}/\prob{S_i}$ --- and the average of such a ratio cannot be rewritten in terms of the averaged variables alone. Closing the averaged system therefore still requires assuming that the individual probabilities are themselves node-independent, e.g.\ $\prob{S_i} \approx \prob{S}$, i.e. the dynamical homogeneity assumption (Assumption~\ref{ass:dynamical-homogeneity}). The final BMC results are unaffected; the difference is only that dynamical homogeneity states this assumption at the outset rather than invoking it implicitly to resolve the averaged equations.
\end{remark}

\subsection{Deriving the pair-based simplicial model} \label{subsec:pair-based}

We show that the mean-field pair-based simplicial contagion model by Malizia et al.~\cite{malizia2025pair} (hereafter PBS model) can be recovered from equations~\eqref{eq:dI}--\eqref{eq:dSII} under specific structural constraints on the network, as well as closing the hierarchy at the level of pairs instead of triplets. 

\subsubsection{Model parameters and notation} \label{subsubsec:pair-based-notation}

The PBS model is defined on a \emph{regular} higher-order network $\mathcal{G} = (\mathcal{N, E, H})$ wherein every node has fixed degree $k_1$ and fixed hyperedge degree $k_2$. The simpliciality constraint requires that every hyperedge contain a complete (pairwise) triangle:
\begin{equation} \label{eq:simplicial-constraint}
    H_{ijk} = 1 \implies A_{ij} A_{ik} A_{jk} = 1.
\end{equation}
Two additional global structural parameters govern the topology: $\phi \in [0, 1]$, the probability that a connected triplet closes to form a triangle, and $\delta \in [0,1]$, the probability that a triangle carries a hyperedge. Note that these parameters are constrained by the values of $k_1$ and $k_2$. 

Due to the structural constraints on the network, the number of possible $4$-node motifs is greatly reduced. Malizia et al.\ introduce further sub-graph notations to encode these structures within the state probabilities.
\begin{remark} \label{rem:pair-based-subgraph-notation}
     To distinguish the sub-graph labeling used by Malizia et al.\ from our notation (i.e. Table~\ref{tab:isomorphism-class-notation}) throughout this subsection, we encode them with \emph{subscripts} on joint probabilities, whereas our framework uses \emph{superscripts}. Specifically, for quadruplet probabilities involving a single 2-simplex $\prob{XYZ}$, we write $\prob{XYZ_\triangle W}_m$ where $m \in \set{0,1,2}$ counts the number of pairwise edges connecting the external node~$W$ to the nodes of the simplex, and the edge $\set{Z, W}$ is implied by their adjacent positioning in the notation.
\end{remark}

\subsubsection{Target equations}\label{subsubsec:pair-target-eqs}

The PBS model tracks five dynamical variables: the single-node densities $\prob{S}$ and $\prob{I}$, and the pair densities $\prob{SS}$, $\prob{SI}$, and $\prob{II}$.  The rate equations, stated in equation~(16) of~\cite{malizia2025pair} and translated into our notation, are:
\begin{subequations}\label{eq:pbs-target}
\begin{align}
  \frac{d}{dt}\prob{S}
  &= \gamma\,\prob{I}
    - \beta_1 k_1\,\prob{SI}
    - \beta_2 k_2\,\prob{ISI}^\bigcirc,
    \label{eq:pbs-target-S} \\[4pt]
  \frac{d}{dt}\prob{I}
  &= -\gamma\,\prob{I}
    + \beta_1 k_1\,\prob{SI}
    + \beta_2 k_2\,\prob{ISI}^\bigcirc,
    \label{eq:pbs-target-I} \\[4pt]
  \frac{d}{dt}\prob{SS}
  &= 2\gamma\,\prob{SI}
    - 2\beta_1(k_1 - 1)\,\prob{SSI}
    - 2\beta_2\,\tfrac{k_2}{k_1}(k_1 - 2)\,
      \prob{IIS_\triangle S},
    \label{eq:pbs-target-SS} \\[4pt]
  \begin{split}\label{eq:pbs-target-SI}
  \frac{d}{dt}\prob{SI}
  &= \gamma\,\prob{II}
    - \gamma\,\prob{SI}
    + \beta_1(k_1 - 1)\,\prob{SSI}
    - \beta_1(k_1 - 1)\,\prob{ISI}
    - \beta_1\,\prob{SI} \\
  &\quad
    + \beta_2\,\tfrac{k_2}{k_1}(k_1 - 2)\,
      \prob{IIS_\triangle S}
    - 2\beta_2\,\tfrac{k_2}{k_1}\,\prob{ISI}^\bigcirc
    - \beta_2\,\tfrac{k_2}{k_1}(k_1 - 2)\,
      \prob{IIS_\triangle I},
  \end{split} \\[4pt]
  \begin{split}\label{eq:pbs-target-II}
  \frac{d}{dt}\prob{II}
  &= -2\gamma\,\prob{II}
    + 2\beta_1(k_1 - 1)\,\prob{ISI}
    + 2\beta_1\,\prob{SI} \\
  &\quad
    + 4\beta_2\,\tfrac{k_2}{k_1}\,\prob{ISI}^\bigcirc
    + 2\beta_2\,\tfrac{k_2}{k_1}(k_1 - 2)\,
      \prob{IIS_\triangle I}.
  \end{split}
\end{align}
\end{subequations}
The partial and empty superscripts on the triplet probabilities here indicate a lack of information at this stage.  As a simplicial model, every pair within a hyperedge is itself an edge, so the presence of a hyperedge forces all three pairwise edges and $\prob{ISI}^{\bigcirc} = \prob{ISI}^{\bigcirc\triangle}$.  Likewise, $\prob{SSI}$ is the sum of all variations of $\prob{SSI}^{\tau\sigma}$ weighted by their density in the network, with $\tau\sigma \in \set{\bigcirc\triangle,\ \otimes\triangle,\ \otimes\wedge}$.
The quadruplet probabilities $\prob{IIS_\triangle S}$ and $\prob{IIS_\triangle I}$ appearing in equations~\eqref{eq:pbs-target-SS}--\eqref{eq:pbs-target-II} are decomposed into sub-configurations according to the number of external edges:
\begin{equation}\label{eq:pbs-target-quadruplet}
\begin{aligned}
  \prob{IIS_\triangle S}
  &= (1-\phi)^2\,\prob{IIS_\triangle S}_0
    + 2\phi(1-\phi)\,\prob{IIS_\triangle S}_1
    + \phi^2\,\prob{IIS_\triangle S}_2, \\[2pt]
  \prob{IIS_\triangle I}
  &= (1-\phi)^2\,\prob{IIS_\triangle I}_0
    + 2\phi(1-\phi)\,\prob{IIS_\triangle I}_1
    + \phi^2\,\prob{IIS_\triangle I}_2.
\end{aligned}
\end{equation}
 
The PBS model closes the hierarchy using a combination of the Kirkwood superposition and the wedge closure (equations~(6)--(7) and~(9)--(10) of~\cite{malizia2025pair}) for triplet densities, and factorized expressions (equation~(17) of~\cite{malizia2025pair}) for quadruplet densities:
\begin{subequations}\label{eq:pbs-target-triplet-closures}
\begin{align}
  \prob{XYZ}^\wedge
    &\approx \frac{\prob{XY}\,\prob{YZ}}
      {\prob{Y}},
    \label{eq:pbs-closure-wedge} \\[3pt]
  \prob{XYZ}^\triangle
    &\approx \frac{\prob{XY}\,\prob{YZ}\,
      \prob{XZ}}
      {\prob{X}\,\prob{Y}\,\prob{Z}},
    \label{eq:pbs-closure-triangle} \\[3pt]
  \prob{XYZ}
    &\approx (1 - \phi)\,\prob{XYZ}^\wedge
      + \phi\,\prob{XYZ}^\triangle.
    \label{eq:pbs-closure-triplet-mixed}
\end{align}
\end{subequations}
The quadruplet closures, stated in equation~(17) of~\cite{malizia2025pair}, factorize each sub-configuration into products of pair and single densities:
\begin{equation}\label{eq:pbs-target-4node-closures}
\renewcommand{\arraystretch}{2.0}
\begin{array}{l@{\qquad}l}
  \prob{IIS_\triangle S}_0
    = \dfrac{\prob{SI}^2\,\prob{II}\,
      \prob{SS}}
      {\prob{S}^2\,\prob{I}^2},
  &\prob{IIS_\triangle I}_0
    = \dfrac{\prob{SI}^3\,\prob{II}}
      {\prob{S}^2\,\prob{I}^2}, \\
  \prob{IIS_\triangle S}_1
    = \dfrac{\prob{SI}^3\,\prob{SS}\,
      \prob{II}}
      {\prob{S}^3\,\prob{I}^3},
  &\prob{IIS_\triangle I}_1
    = \dfrac{\prob{SI}^3\,\prob{II}^2}
      {\prob{S}^2\,\prob{I}^4}, \\
  \prob{IIS_\triangle S}_2
    = \dfrac{\prob{SS}\,\prob{SI}^4\,
      \prob{II}}
      {\prob{S}^4\,\prob{I}^4},
  &\prob{IIS_\triangle I}_2
    = \dfrac{\prob{SI}^3\,\prob{II}^3}
      {\prob{S}^2\,\prob{I}^6}.
\end{array}
\end{equation}

\subsubsection{Derivation}\label{subsubsec:pair-derivation}
 
The derivation proceeds in three stages: we first establish exact network quantities on the regular simplicial complex, then evaluate the partial sums arising under topological homogeneity, and finally recover the target equations.
 
\medskip
\noindent\textit{Network quantities.}
The structural parameters $k_1$, $\phi$, and $\delta$ determine the global counts of each sub-graph type.  Expressing these as sums of adjacency products gives the following.
 
\begin{proposition}\label{prop:pair-network-counts}
  On a regular higher-order network satisfying the simplicial constraint according to equation~\eqref{eq:simplicial-constraint}, the following identities hold:
  \begin{align}
    \sum_{\substack{i,j \\ \textup{distinct}}} A_{ij}
      &= N k_1,
      \label{eq:net-edges} \\
    \sum_{\substack{i,j,k \\ \textup{distinct}}}
      A_{ij}\, A_{ik}
      &= N k_1(k_1 - 1),
      \label{eq:net-triples} \\
    \sum_{\substack{i,j,k \\ \textup{distinct}}}
      A_{ij}\, A_{ik}\, \bar{A}_{jk}
      &= N k_1(k_1 - 1)(1 - \phi),
      \label{eq:net-wedges} \\
    \sum_{\substack{i,j,k \\ \textup{distinct}}}
      A_{ij}\, A_{ik}\, A_{jk}
      &= N k_1(k_1 - 1)\phi,
      \label{eq:net-triangles} \\
    \sum_{\substack{i,j,k \\ \textup{distinct}}}
      H_{ijk}\, A_{ij}\, A_{ik}\, A_{jk}
      &= N k_1(k_1 - 1)\phi\delta
       = 2N k_2,
      \label{eq:net-simplices} \\
    \sum_{\substack{i,j,x,y \\ \textup{distinct}}}
      A_{ij}\, H_{ixy}\, A_{ix}\, A_{iy}\, A_{xy}
      &= N k_1(k_1 - 1)(k_1 - 2)\phi\delta = 2Nk_2 (k_1 - 2).
      \label{eq:net-simplex-ext}
  \end{align}
\end{proposition}
 
\begin{proof}
  Equation~\eqref{eq:net-edges} is the definition of the mean degree on a regular network.  Equation~\eqref{eq:net-triples} counts connected triples centered at~$i$: each of the $N$ nodes has $k_1$ neighbors, from which two distinct neighbors $j,k$ are chosen in order, giving $k_1(k_1-1)$ per node.  Equations~\eqref{eq:net-wedges}--\eqref{eq:net-triangles} partition these triples via $\bar{A}_{jk} + A_{jk} = 1$ and the definition of~$\phi$.  Equation~\eqref{eq:net-simplices} selects triangles carrying a hyperedge via~$\delta$; the factor $2Nk_2$ follows from the fact that every hyperedge $\set{i,j,k}$ is counted once for each of the $N$ nodes and each node counts the pair $\set{j,k}$ twice (once as $(j,k)$ and once as $(k,j)$), giving $2Nk_2$ in total. For equation~\eqref{eq:net-simplex-ext}, each simplicial hyperedge $\set{i,x,y}$ contributes $(k_1 - 2)$ edges from~$i$ to nodes outside the hyperedge (since $i$ necessarily has $2$ edges within the hyperedge), and the sum over~$j$ restricted to these external neighbors gives the stated count.
\end{proof}

\noindent\textit{Partial sums under topological homogeneity.}
After applying dynamical homogeneity, the rate equations~\eqref{eq:dI} and~\eqref{eq:dSI} separate into class-specific probability terms and partial sums of adjacency products with certain nodes fixed.  By Assumption~\ref{ass:topological-homogeneity}, each partial sum is promoted to a complete sum and evaluated using Proposition~\ref{prop:pair-network-counts}.

 \pagebreak
 
\begin{lemma}\label{lem:pair-partial-sums}
  Under topological homogeneity, the following partial sums hold exactly or approximately as indicated:
  \begin{align}
    \sum_{x \neq i} A_{ix}
      &= k_1,
      \label{eq:partial-degree} \\
    \sum_{\substack{x < y \\ x,y \neq i}} H_{ixy}
      &= k_2,
      \label{eq:partial-hypdegree} \\
    A_{ij} \sum_{x \neq i,j} A_{ix}
      &= k_1 - 1,
      \label{eq:partial-edge-degree} \\
    A_{ij} \sum_{x \neq i,j} H_{ijx}\, A_{ix}\, A_{jx}
      &\approx (k_1 - 1)\phi\delta,
      \label{eq:partial-simplex} \\
    A_{ij} \sum_{\substack{x < y \\ x,y \neq i,j}}
      H_{jxy}\, A_{jx}\, A_{jy}\, A_{xy}
      &\approx \tfrac{1}{2}(k_1 - 1)(k_1 - 2)\phi\delta = \frac{k_2}{k_1} (k_1 - 2).
      \label{eq:partial-simplex-ext}
  \end{align}
\end{lemma}
 
\begin{proof}
  Equations~\eqref{eq:partial-degree} and~\eqref{eq:partial-hypdegree} follow directly from regularity.  Equation~\eqref{eq:partial-edge-degree} follows since node~$i$ has $k_1$ neighbors, one of which is~$j$, leaving $k_1 - 1$.
 
  For equation~\eqref{eq:partial-simplex}, we promote with $p = A_{ij}$ from Assumption~\ref{ass:topological-homogeneity}:
  \[
    A_{ij}\!\sum_{x \neq i,j}
      H_{ijx}\, A_{ix}\, A_{jx}
    \;\approx\;
    \frac{
      \displaystyle\sum_{\substack{i,j,x \\ \textup{distinct}}}
        H_{ijx}\, A_{ij}\, A_{ix}\, A_{jx}
    }{
      \displaystyle\sum_{\substack{i,j \\ \textup{distinct}}} A_{ij}
    }
    = \frac{N k_1(k_1 - 1)\phi\delta}{N k_1}
    = (k_1 - 1)\phi\delta.
  \]
  Equation~\eqref{eq:partial-simplex-ext} follows by an analogous promotion using equations~\eqref{eq:net-simplex-ext} and~\eqref{eq:net-edges}, together with the ordered-to-unordered conversion $\sum_{x<y} \to \frac{1}{2}\sum_{\substack{x,y \\ y \neq x}}$.
\end{proof}
 
\begin{remark}\label{rem:pair-partial-approx}
  Equations~\eqref{eq:partial-simplex} and~\eqref{eq:partial-simplex-ext} are approximate: $\phi$ and~$\delta$ are global network properties that need not hold locally around every edge.  However, the product $\phi\delta$ is exact in aggregate since $k_2 = k_1(k_1 - 1)\phi\delta/2$ is fixed by regularity.
\end{remark}
 
\medskip
\noindent\textit{Recovering the target rate equations.}
 
\begin{theorem}\label{thm:pair-rate-eqs}
  On a regular simplicial complex with parameters $k_1$, $\phi$, $\delta$, applying dynamical and topological homogeneity to equations~\eqref{eq:dI} and~\eqref{eq:dSI} and evaluating the partial sums via the identities in Lemma~\ref{lem:pair-partial-sums} yields the PBS target equations~\eqref{eq:pbs-target}.
\end{theorem}
 
\begin{proof}
  We treat each rate equation in turn.
 
  \textit{Step~1: The node-level equation.}
  Equation~\eqref{eq:dI} reads
  \[
    \frac{d}{dt}\prob{I_i}
    = -\gamma\,\prob{I_i}
      + \beta_1 \sum_{x \neq i} A_{ix}\,\prob{S_i I_x}
      + \beta_2 \sum_{\substack{x < y \\ x,y \neq i}}
        H_{ixy}\,\prob{S_i I_x I_y}.
  \]
  Under dynamical homogeneity, $\prob{I_i} = \prob{I}$ and $\prob{S_i I_x} = \prob{SI}$ for all edges $\set{i,x}$.  The $\beta_1$ sum evaluates to $k_1\,\prob{SI}$ by equation~\eqref{eq:partial-degree}.  For the $\beta_2$ sum, the simplicial constraint~\eqref{eq:simplicial-constraint} ensures that every hyperedge $\set{i,x,y}$ carries a complete triangle, so $\prob{S_i I_x I_y}$ belongs to the class $\prob{ISI}^{\bigcirc\triangle}$. Equation~\eqref{eq:partial-hypdegree} then gives $k_2\,\prob{ISI}^{\bigcirc \triangle}$, recovering equation~\eqref{eq:pbs-target-I}, where the simplicial constraint~\eqref{eq:simplicial-constraint} ensures that every hyperedge carries a complete triangle, so $\prob{ISI}^{\bigcirc\triangle} = \prob{ISI}^{\bigcirc}$.
 
  \textit{Step~2: The pair-level equation.}
  Equation~\eqref{eq:dSI} contains: recovery terms, an internal pairwise infection term, external pairwise infection terms, a mixed hyperedge infection term, and external hyperedge infection terms. We evaluate each below.
 
  The recovery terms $-\gamma\,\prob{SI} + \gamma\,\prob{II}$ and the internal pairwise infection term $-\beta_1\,\prob{SI}$ pass through unchanged under dynamical homogeneity.
 
  The external pairwise infection terms contribute
  \[
    -\beta_1(k_1 - 1)\,\prob{ISI}
    + \beta_1(k_1 - 1)\,\prob{SSI},
  \]
  using equation~\eqref{eq:partial-edge-degree}. The generic triple densities $\prob{ISI}$ and $\prob{SSI}$ contain both wedge and triangle contributions, as obtained by splitting equation~\eqref{eq:partial-edge-degree} into the triangle and wedge variant through equations~\eqref{eq:net-wedges} and~\eqref{eq:net-triangles}.
 
  For the mixed hyperedge infection term, given the prefactor $A_{ij}\, H_{ijx}$, the simplicial constraint forces $A_{ix} = A_{jx} = 1$, so the triple is a simplex.  By equation~\eqref{eq:partial-simplex}, the coefficient is $(k_1 - 1)\phi\delta = 2k_2/k_1$.  The two contributions (one from each susceptible node in the simplex being infected) yield $-2\beta_2\,(k_2/k_1)\,\prob{ISI}^{\bigcirc \triangle}$.
 
  For the external hyperedge infection terms, the prefactor $A_{ij}\, H_{jxy}\, A_{jx}\, A_{jy}\, A_{xy}$ involves a simplex $\set{j,x,y}$ connected to~$i$ via the edge $\set{i,j}$, producing a four-node motif.  By equation~\eqref{eq:partial-simplex-ext}, the coefficient is $(k_2/k_1)(k_1 - 2)$.  The resulting terms involve the quadruplet densities $\prob{IIS_\triangle S}$ and $\prob{IIS_\triangle I}$ (depending on the state of~$i$), contributing the remaining terms in~\eqref{eq:pbs-target-SI}.  The equations for $\prob{SS}$ and $\prob{II}$ need not be derived independently: they follow from the marginal identities $\prob{II} = \prob{I} - \prob{SI}$ and $\prob{SS} = \prob{S} - \prob{SI}$ applied to equations~\eqref{eq:pbs-target-I} and~\eqref{eq:pbs-target-SI}.
\end{proof}
 
\medskip
\noindent\textit{Recovering the closure expressions.}

\begin{remark} \label{rem:pbs-closure-triples}
    Since the PBS model hierarchy closes at the pair level, the target triplet closures~\eqref{eq:pbs-target-triplet-closures} follows directly from the canonical triplet closure in equation~\eqref{eq:triplet-closure-lemma} (Lemma~\ref{lem:triplet-closure}) under dynamical homogeneity (Assumption~\ref{ass:dynamical-homogeneity}), which effectively drops the node labels. The proportion of each class---triangle and wedge---among all triplets is determined by $\phi$ and $(1-\phi)$, respectively, which yields equation~\eqref{eq:pbs-closure-triplet-mixed}.
\end{remark}



\begin{lemma}\label{lem:pair-closure-quadruple}
  Consider a quadruplet consisting of a simplex (three nodes with all pairwise edges and a hyperedge) connected to an external node via at least one pairwise edge. Applying the truncated Kirkwood approximation $\tilde{\kappa}$ and the independence operator $\Phi$ for each absent edge yields four cases:
  \begin{equation}\label{eq:pair-closure-quadruple}
    \prob{X Y Z W} \approx
      \begin{cases}
        \dfrac{\prob{X Y}\,\prob{X Z}\,\prob{X W}\,
          \prob{Y Z}\,\prob{Y W}\,\prob{Z W}}
          {\prob{X}^2\,\prob{Y}^2\,\prob{Z}^2\,\prob{W}^2}
        & \textup{if both external edges are present}, \\[10pt]
        \dfrac{\prob{X Y}\,\prob{X W}\,
          \prob{Y Z}\,\prob{Y W}\,\prob{Z W}}
          {\prob{X}\,\prob{Y}^2\,\prob{Z}\,\prob{W}^2}
        & \textup{if only the edge between $X$ and $W$ is present}, \\[10pt]
        \dfrac{\prob{X Y}\,\prob{X Z}\,
          \prob{Y Z}\,\prob{Y W}\,\prob{Z W}}
          {\prob{X}\,\prob{Y}^2\,\prob{Z}^2\,\prob{W}}
        & \textup{if only the edge between $X$ and $Z$ is present}, \\[10pt]
        \dfrac{\prob{X Y}\,\prob{Y Z}\,
          \prob{Y W}\,\prob{Z W}}
          {\prob{Y}^2\,\prob{Z}\,\prob{W}}
        & \textup{if both external edges are absent}.
      \end{cases}
  \end{equation}
  Under topological homogeneity, the four cases occur with densities $\phi^2$, $\phi(1-\phi)$, $\phi(1-\phi)$, and $(1-\phi)^2$ respectively, where $\phi$ is the probability that a connected triplet closes into a triangle.
\end{lemma}

\begin{proof}
  Although the lemma is stated without node labels, the case analysis requires us to identify \emph{which} of the two unresolved edges is present in the ``one edge present'' configurations. We therefore introduce labels as scaffolding for the proof and discard them at the end. Let $\set{j,k,l}$ be the simplex and $i$ be the external node with the known edge $\set{i,j}$; the unresolved edges are $\set{i,k}$ and $\set{i,l}$.

  The truncated Kirkwood approximation on four nodes, $\tilde\kappa$, expresses the joint probability in terms of pair and single probabilities:
  \[
    \tilde\kappa(\set{X_i Y_j Z_k W_l}) = \frac{\prob{X_i Y_j}\,\prob{X_i Z_k}\,\prob{X_i W_l}\,
        \prob{Y_j Z_k}\,\prob{Y_j W_l}\,\prob{Z_k W_l}}
        {\prob{X_i}^2\,\prob{Y_j}^2\,\prob{Z_k}^2\,\prob{W_l}^2}.
  \]
  The simplex fixes $A_{jk} = A_{kl} = A_{jl} = 1$, so the independence operator does not act on the internal pairs. For each absent external edge, $\Phi$ factorizes the corresponding pair, e.g.\ $\Phi_{ik}\colon \prob{X_i Z_k} \mapsto \prob{X_i}\prob{Z_k}$, which then cancels against factors in the denominator. The four combinations of presence/absence of $\set{i,k}$ and $\set{i,l}$ give:
  \begin{equation}\label{eq:pair-closure-quadruple-labelled}
    \prob{X_i Y_j Z_k W_l} \approx
    \begin{cases}
      \dfrac{\prob{X_i Y_j}\,\prob{X_i Z_k}\,\prob{X_i W_l}\,
        \prob{Y_j Z_k}\,\prob{Y_j W_l}\,\prob{Z_k W_l}}
        {\prob{X_i}^2\,\prob{Y_j}^2\,\prob{Z_k}^2\,\prob{W_l}^2}
      & \textup{if } A_{ik} = A_{il} = 1, \\[10pt]
      \dfrac{\prob{X_i Y_j}\,\prob{X_i W_l}\,
        \prob{Y_j Z_k}\,\prob{Y_j W_l}\,\prob{Z_k W_l}}
        {\prob{X_i}\,\prob{Y_j}^2\,\prob{Z_k}\,\prob{W_l}^2}
      & \textup{if } A_{ik} = 0,\; A_{il} = 1, \\[10pt]
      \dfrac{\prob{X_i Y_j}\,\prob{X_i Z_k}\,
        \prob{Y_j Z_k}\,\prob{Y_j W_l}\,\prob{Z_k W_l}}
        {\prob{X_i}\,\prob{Y_j}^2\,\prob{Z_k}^2\,\prob{W_l}}
      & \textup{if } A_{ik} = 1,\; A_{il} = 0, \\[10pt]
      \dfrac{\prob{X_i Y_j}\,\prob{Y_j Z_k}\,
        \prob{Y_j W_l}\,\prob{Z_k W_l}}
        {\prob{Y_j}^2\,\prob{Z_k}\,\prob{W_l}}
      & \textup{if } A_{ik} = A_{il} = 0.
    \end{cases}
  \end{equation}
  Dropping node labels under dynamical homogeneity gives the expressions in~\eqref{eq:pair-closure-quadruple}.

  For the densities, consider the triplet $\set{i,j,x}$ for $x \in \set{k,l}$. By construction $A_{ij} = A_{jx} = 1$, so this triplet is a connected triple that forms a triangle if $A_{ix} = 1$ and a wedge (centered at $j$) if $A_{ix} = 0$. By the definition of $\phi$, these occur with probabilities $\phi$ and $1 - \phi$ respectively, independently for $x = k$ and $x = l$. The four cases of the lemma therefore have densities $\phi^2$, $\phi(1-\phi)$, $\phi(1-\phi)$, and $(1-\phi)^2$.

  Specializing the state labels to $X = I$, $Y = I$, $Z = S$, and $W \in \set{S, I}$, and combining the two ``one edge present'' cases (which share density $\phi(1-\phi)$ and yield the same closure after relabeling), recovers the sub-configurations $\prob{IIS_\triangle S}_m$ and $\prob{IIS_\triangle I}_m$ of equation~\eqref{eq:pbs-target-4node-closures} with weights $(1-\phi)^2$, $2\phi(1-\phi)$, and $\phi^2$ for $m = 0, 1, 2$. This reproduces the decomposition~\eqref{eq:pbs-target-quadruplet} and matches equation~(17)~of~\cite{malizia2025pair}.
\end{proof}


\begin{remark}\label{cor:pair-model-recovery}
  Starting from the microscopic rate equations at the node and pair level~\eqref{eq:dI}--\eqref{eq:dSI}, the pair-based simplicial contagion model of~\cite{malizia2025pair}---i.e.\ the target equations~\eqref{eq:pbs-target} together with the closures~\eqref{eq:pbs-target-triplet-closures} and~\eqref{eq:pbs-target-4node-closures}---is recovered by applying the following ingredients:
  \begin{enumerate}
    \item \emph{Structural constraints on the network.} Regularity (every node has degree $k_1$ and hyperedge degree $k_2$) and the simplicial constraint~\eqref{eq:simplicial-constraint} (every hyperedge carries a complete triangle).
    \item \emph{Closure of the hierarchy at the pair level.} The truncated Kirkwood approximation $\tilde\kappa$, combined with the independence operator $\Phi$, expresses triplet and quadruplet motifs in terms of pairs and singles. For the triplets appearing in~\eqref{eq:dSI}, this yields the wedge and triangle closures of Remark~\ref{rem:pbs-closure-triples}; for the quadruplets, it yields Lemma~\ref{lem:pair-closure-quadruple}.
    \item \emph{Dynamical homogeneity} (Assumption~\ref{ass:dynamical-homogeneity}). State joint probabilities depend only on the isomorphism class of the sub-graph spanned by the nodes, not on the node labels themselves.
    \item \emph{Topological homogeneity} (Assumption~\ref{ass:topological-homogeneity}). Local sub-graph counts (degree, hyperedge membership of an edge, etc.) are replaced by their network-wide averages, evaluated using Lemma~\ref{lem:pair-partial-sums}.
  \end{enumerate}
  Combining these ingredients with Theorem~\ref{thm:pair-rate-eqs} produces the closed system~\eqref{eq:pbs-target} with no further approximation.
\end{remark}

\subsection{Deriving the inter-order model} \label{subsec:inter-order}

We show that the mean-field inter-order overlap model of Malizia et al.~\cite{malizia2026nested} can be recovered from equations~\eqref{eq:dI}--\eqref{eq:dSII} under structural and homogeneity assumptions, together with closing the hierarchy at the level of triplets. Unlike the PBS model of Subsection~\ref{subsec:pair-based}, the inter-order model does not require hyperedges to carry an underlying pairwise triangle; instead, it parametrizes the overlap between pairwise edges and hyperedges through a single global parameter $\alpha \in [0,1]$.

The derivation of this model differs from the previous two in an important respect. As we will demonstrate, the structural and dynamical assumptions \emph{explicitly stated} in~\cite{malizia2026nested} are not, on their own, sufficient to recover the target rate equations from the microscopic system~\eqref{eq:dI}--\eqref{eq:dSII}. Under the stated structural constraints and given network parameters, applying the mean-field approach to the microscopic rate equations leaves the system with extra degrees of freedom that remain unresolved. In order to recover the inter-order model exactly, further parameterizations and extra assumptions must be invoked: (a) $3$ parameters $\theta(m)$ to characterize the exact distribution of edges within hyperedges, (b) a hyperedge dynamical homogeneity assumption that treats all hyperedge triplet probabilities as equivalent regardless of their internal pairwise structure, and (c) a stronger sparsity condition than stated. We make these implicit assumptions explicit alongside the stated ones, and carry out the derivation in two stages: first solving the dynamics on terms of motif size $\leq 3$ and demonstrating where the extra parameters $\theta$ and hyperedge dynamical homogeneity assumption enter, then solving the terms for motifs of size $4$ and $5$ using the stronger sparsity condition.

\subsubsection{Model parameters and notation} \label{subsubsec:inter-order-notation}

The inter-order model is defined on a higher-order network $\mathcal{G} = (\mathcal{N}, \mathcal{E}, \mathcal{H})$ with average degree $k_1$, average hyperedge degree $k_2$ (recall Definition~\ref{def:mean-degrees}), and clustering coefficient $\phi$ (recall Definition~\ref{def:clustering}).

The inter-order model further imposes the following structural/topological sparsity assumptions on $\mathcal{G}$, stated explicitly in~\cite{malizia2026nested}.

\begin{assumption}\label{ass:inter-order-stated-sparsity}
  The network $\mathcal{G}$ satisfies:
  \begin{enumerate}
    \item \emph{Hyperedge pair-disjointedness:} $H_{ijk}\, H_{ijk'} = 0$ for all $i, j \in \mathcal{N}$ and all distinct $k, k' \in \mathcal{N}$. That is, no two hyperedges share two common nodes.
    \item \emph{Triangle edge-disjointedness:} no two triangles share a common edge. Equivalently, among any four distinct nodes, at most one triangle can exist.
  \end{enumerate}
\end{assumption}

These conditions limit the number of distinct four- and five-node motifs that need to be tracked, which is essential for the model to admit a closed system of equations.

Unlike the PBS model, the inter-order model permits hyperedges that do not carry a complete pairwise triangle. The overlap between pairwise and hyperedge structures is captured by a single global parameter $\alpha$.
\begin{definition}\label{def:alpha}
  For a hyperedge $h = \set{i,j,k} \in \mathcal{H}$, let $\epsilon(h) = A_{ij} + A_{ik} + A_{jk}$ denote the number of present pairwise edges among its three node pairs. The \emph{inter-order overlap parameter} $\alpha$ is
  \[
    \alpha
    = \frac{1}{3|\mathcal{H}|}
      \sum_{h \in \mathcal{H}} \epsilon(h).
  \]
\end{definition}

\begin{proposition}\label{prop:alpha-range}
  The inter-order overlap parameter satisfies $0 \leq \alpha \leq 1$.
\end{proposition}

\begin{proof}
  Each hyperedge contributes between $0$ and $3$ to the numerator, since $A_{ij} \in \set{0,1}$.
\end{proof}

\begin{remark}\label{rem:alpha-extremes}
  When $\alpha = 0$, no pair of nodes in any hyperedge is connected by an edge. When $\alpha = 1$, every hyperedge $\set{i,j,k} \in \mathcal{H}$ satisfies $\set{i,j}, \set{i,k}, \set{j,k} \in \mathcal{E}$, so that each hyperedge together with its pairwise edges forms a $2$-simplex.
\end{remark}

\begin{lemma}\label{lem:alpha-tensor}
  The inter-order overlap parameter can be expressed as
  \[
    \alpha
    = \frac{1}{6\,|\mathcal{H}|}
      \sum_{i,j,k} H_{ijk}\, A_{ik}.
  \]
\end{lemma}

\begin{proof}
  Since $\mathbf{H}$ is symmetric under all permutations, each hyperedge $\set{i,j,k} \in \mathcal{H}$ contributes to all $3! = 6$ ordered triples, so $\sum_{i,j,k} H_{ijk} = 6|\mathcal{H}|$. The sum $A_{ij} + A_{ik} + A_{jk}$ is likewise invariant under permutations of its indices, and the three terms $H_{ijk}\,A_{ij}$, $H_{ijk}\,A_{ik}$, $H_{ijk}\,A_{jk}$ each yield the same value when summed over all $i,j,k$ by relabeling dummy indices. Hence
  \[
    \sum_{i,j,k} H_{ijk}\bigl(A_{ij} + A_{ik} + A_{jk}\bigr)
    = 3 \sum_{i,j,k} H_{ijk}\, A_{ik},
  \]
  and dividing by $3 \times 6|\mathcal{H}|$ gives the result.
\end{proof}

As it stands, $\alpha$ captures only the mean number of edges per hyperedge and is insufficient to fully resolve the isomorphism class of each hyperedge triplet. The full distribution of edge counts, which we denote $\theta$, will be required when we carry out the derivation of up to triplet terms---this is the first extra parameterisation that we introduce to the inter-order model.

\begin{definition}\label{def:theta}
  The \emph{inter-order overlap probability distribution} of~$\mathcal{G}$ is the function $\theta \colon \set{0,1,2,3} \to [0,1]$ defined by
  \[
    \theta(m)
    = \frac{1}{|\mathcal{H}|}
      \sum_{\set{i,j,k} \in \mathcal{H}}
      \delta\bigl(A_{ij} + A_{ik} + A_{jk},\; m\bigr),
  \]
  where $\delta$ denotes the Kronecker delta. That is, $\theta(m)$ is the proportion of hyperedges containing exactly $m$ pairwise edges. By construction, $\sum_{m=0}^{3} \theta(m) = 1$.
\end{definition}

\begin{lemma}\label{lem:alpha-from-theta}
  The inter-order overlap parameter is the first moment (i.e. expectation value) of $\theta$:
  \[
    \alpha
    = \frac{1}{3} \sum_{m=0}^{3} m\,\theta(m).
  \]
\end{lemma}

\begin{proof}
  Expanding the expectation gives
  \[
    \mathbb{E}[M]
    = \sum_{m} m\,\theta(m)
    = \frac{1}{|\mathcal{H}|}
      \sum_{\set{i,j,k} \in \mathcal{H}}
      \sum_{m} m\,\delta\bigl(
        A_{ij} + A_{ik} + A_{jk},\, m
      \bigr).
  \]
  For each hyperedge, the Kronecker delta selects $m = A_{ij} + A_{ik} + A_{jk}$, so the inner sum
  reduces to $A_{ij} + A_{ik} + A_{jk} = \epsilon(\set{i,j,k})$. Therefore
  $\mathbb{E}[M]
  = \frac{1}{|\mathcal{H}|}
    \sum_{h \in \mathcal{H}} \epsilon(h)
  = 3\alpha$
  by Definition~\ref{def:alpha}.
\end{proof}

\subsubsection{Target equations and closures} \label{subsubsec:inter-order-target-eqs}

\begin{notation}\label{not:composite-motifs}
  In addition to the isomorphism class notation of Table~\ref{tab:isomorphism-class-notation}, we introduce notation for two composite motifs that arise from the IO closure approximations. Let $\set{i,j,k} \in \mathcal{H}$ be a hyperedge with nodes in states $X_i, Y_j, Z_k$.

  \emph{Hyperedge--pendant motif} (4 nodes). We write $\prob{X_i Y_j Z_k {}_{\bigcirc} W_l}$ for the joint probability over nodes $i, j, k, l$, where $\set{i,j,k}$ form a hyperedge and node $l$ is connected to the \emph{pivotal node} $k$ by a pairwise edge $\set{k,l} \in \mathcal{E}$. The pivotal node (state $Z_k$) always occupies the third entry, and the motif is symmetric in its first two entries.

  \emph{Bowtie motif} (5 nodes). We write $\prob{X_i Y_j Z_k {}_{\bowtie} W_l U_m}$ for the joint probability over nodes $i, j, k, l, m$, where $\set{i,j,k}$ and $\set{k,l,m}$ are two hyperedges sharing the pivotal node $k$ (state $Z_k$, third entry of the first triple). 

  The subscript symbol (${}_{\bigcirc}$ or ${}_{\bowtie}$) marks the boundary between the hyperedge triple and the external structure attached to the pivotal node.
\end{notation}

The inter-order model tracks five dynamical variables: $\prob{I}$, $\prob{SI}$, $\prob{SSS}^{\bigcirc}$, $\prob{SSI}^{\bigcirc}$, and $\prob{ISI}^{\bigcirc}$. The rate equations, stated in equation~(10) of~\cite{malizia2026nested} and translated into our notation, are as follows.

\begin{subequations}\label{eq:io-target-eqs}
\begin{equation}\label{eq:io-target-node}
  \dot{\prob{I}}
  = -\gamma\,\prob{I}
    + \beta_1 k_1\,\prob{SI}
    + \beta_2 k_2\,\prob{ISI}^{\bigcirc},
\end{equation}

\begin{equation}\label{eq:io-target-pair}
\begin{aligned}
  \dot{\prob{SI}}
  &= \beta_1 \bigl[
       (k_1 - 1)\,\bigl(\prob{SSI}
       - \prob{ISI} \bigr)
       - \prob{SI}
     \bigr] \\
  &\quad + \frac{\beta_2 k_2}{k_1} \bigl[
       (k_1 - 2\alpha)\bigl(
         \prob{IIS {}_{\bigcirc} S}
         - \prob{IIS {}_{\bigcirc} I}
       \bigr)
       - 2\alpha\,\prob{ISI}^{\bigcirc}
     \bigr] \\
  &\quad + \gamma \bigl[
       \prob{II} - \prob{SI}
     \bigr],
\end{aligned}
\end{equation}

\begin{equation}\label{eq:io-target-sss}
  \dot{\prob{SSS}}^{\bigcirc}
  = -3\beta_1 (k_1 - 2\alpha)\,
      \prob{SSS {}_{\bigcirc} I}
    - 3\beta_2 (k_2 - 1)\,
      \prob{SSS {}_{\bowtie} II}
    + 3\gamma\,\prob{SSI}^{\bigcirc},
\end{equation}

\begin{equation}\label{eq:io-target-ssi}
\begin{aligned}
  \dot{\prob{SSI}}^{\bigcirc}
  &= \beta_1 \bigl[
       (k_1 - 2\alpha)\bigl(
         \prob{SSS {}_{\bigcirc} I}
         - 2\,\prob{ISS {}_{\bigcirc} I}
       \bigr)
       - 2\alpha\,\prob{SSI}^{\bigcirc}
     \bigr] \\
  &\quad + \beta_2 (k_2 - 1) \bigl[
       \prob{SSS {}_{\bowtie} II}
       - 2\,\prob{ISS {}_{\bowtie} II}
     \bigr] \\
  &\quad + \gamma \bigl[
       2\,\prob{ISI}^{\bigcirc} - \prob{SSI}^{\bigcirc}
     \bigr],
\end{aligned}
\end{equation}

\begin{equation}\label{eq:io-target-isi}
\begin{aligned}
  \dot{\prob{ISI}}^{\bigcirc}
  &= \beta_1 \bigl[
       (k_1 - 2\alpha)\bigl(
         2\,\prob{ISS {}_{\bigcirc} I}
         - \prob{IIS {}_{\bigcirc} I}
       \bigr)
       + 2\alpha\bigl(
         \prob{SSI}^{\bigcirc}
         - \prob{ISI}^{\bigcirc}
       \bigr)
     \bigr] \\
  &\quad + \beta_2 \bigl[
       (k_2 - 1)\bigl(
         2\,\prob{ISS {}_{\bowtie} II}
         - \prob{IIS {}_{\bowtie} II}
       \bigr)
       - \prob{ISI}^{\bigcirc}
     \bigr] \\
  &\quad + \gamma \bigl[
       \prob{III}^{\bigcirc}
       - 2\,\prob{ISI}^{\bigcirc}
     \bigr].
\end{aligned}
\end{equation}
\end{subequations}

Equations~\eqref{eq:io-target-eqs} are not closed: the right-hand sides contain pendant and bowtie probabilities on four and five nodes respectively. The inter-order model closes these using the following approximations, stated in equation~(11) of~\cite{malizia2026nested}. The pairwise open-triple closures take the wedge decomposition (Subsection~\ref{subsec:triplet-closure}):
\begin{equation}\label{eq:io-closure-triples}
  \prob{ISI}
  = \frac{\prob{SI}^2}{\prob{S}},
  \qquad
  \prob{SSI}
  = \frac{\prob{SS}\,\prob{SI}}{\prob{S}}.
\end{equation}

The pendant closures factorize the four-node probability around the pivotal node:
\begin{equation}\label{eq:io-closure-pendant}
\begin{aligned}
  \prob{SSS {}_{\bigcirc} I}
  &= \frac{\prob{SSS}^{\bigcirc}\,\prob{SI}}{\prob{S}},
  &
  \prob{ISS {}_{\bigcirc} I}
  &= \frac{\prob{SSI}^{\bigcirc}\,\prob{SI}}{\prob{S}}, \\[4pt]
  \prob{IIS {}_{\bigcirc} S}
  &= \frac{\prob{ISI}^{\bigcirc}\,\prob{SS}}{\prob{S}},
  &
  \prob{IIS {}_{\bigcirc} I}
  &= \frac{\prob{ISI}^{\bigcirc}\,\prob{SI}}{\prob{S}}.
\end{aligned}
\end{equation}

The bowtie closures factorize the five-node probability around the shared pivotal node:
\begin{equation}\label{eq:io-closure-bowtie}
\begin{aligned}
  \prob{SSS {}_{\bowtie} II}
  &= \frac{\prob{SSS}^{\bigcirc}\,\prob{ISI}^{\bigcirc}}{\prob{S}},
  &
  \prob{ISS {}_{\bowtie} II}
  &= \frac{\prob{SSI}^{\bigcirc}\,\prob{ISI}^{\bigcirc}}{\prob{S}}, \\[4pt]
  \prob{IIS {}_{\bowtie} II}
  &= \frac{\bigl(\prob{ISI}^{\bigcirc}\bigr)^2}{\prob{S}}.
\end{aligned}
\end{equation}

\begin{remark}\label{rem:io-closure-canonical}
  Under dynamical homogeneity (Assumption~\ref{ass:dynamical-homogeneity}), the triplet closures in~\eqref{eq:io-closure-triples} follow from the canonical triplet closure of Lemma~\ref{lem:triplet-closure}. The pendant and bowtie closures in~\eqref{eq:io-closure-pendant}--\eqref{eq:io-closure-bowtie} are recovered from the canonical pendant and bowtie closures of Lemma~\ref{lem:pendant-closure} and Lemma~\ref{lem:bowtie-closure}, with the wedge reductions of Remark~\ref{rem:pendant-with-wedge-closure} and Remark~\ref{rem:bowtie-with-wedge-closure} applied to the open triplets. 
\end{remark}

\subsubsection{Implicit assumptions} \label{subsubsec:io-implicit-assumptions}

Two features of the target equations~\eqref{eq:io-target-eqs} and closures~\eqref{eq:io-closure-pendant}--\eqref{eq:io-closure-bowtie} reveal that assumptions beyond those of Assumption~\ref{ass:inter-order-stated-sparsity} are at work.

First, the pendant and bowtie closures admit only a single configuration each. The pendant closure~\eqref{eq:io-closure-pendant} treats the four-node motif as a hyperedge $\set{i,j,k}$ with a single dangling edge $\set{k,l}$, and no other connections among the four nodes. The bowtie closure~\eqref{eq:io-closure-bowtie} treats the five-node motif as two hyperedges $\set{i,j,k}$ and $\set{k,l,m}$ joined at a single pivotal node $k$, with no other connections. In principle, the stated sparsity assumptions do not preclude additional pairwise edges from forming within these motifs, e.g. a pairwise edge $\set{i,l}$ in the pendant configuration would not violate either condition of Assumption~\ref{ass:inter-order-stated-sparsity} (see figure~\ref{fig:inter-order-motifs}). The absence of any closure for such configurations indicates that the model implicitly forbids them.

Second, the dynamical variables tracked by the inter-order model treat the hyperedge triplet probability $\prob{XYZ}^{\bigcirc}$ as a single object, with no distinction between the four hyperedge isomorphism classes $\set{\bigcirc\triangle, \bigcirc\wedge, \bigcirc-, \bigcirc\bullet}$ identified in Table~\ref{tab:isomorphism-class-notation}. Our general microscopic framework requires these to be tracked separately, since hyperedges with different internal pairwise structures admit different infection pathways. The inter-order model implicitly collapses these four classes to a single common~value.

We make these two implicit conditions explicit in the following assumptions.

\begin{assumption}\label{ass:io-strong-sparsity}
  \emph{(Strong sparsity.)} In addition to the conditions of Assumption~\ref{ass:inter-order-stated-sparsity}, the network $\mathcal{G}$ satisfies:
  \begin{enumerate}
    \item Any pairwise triangle shares at most one node with any distinct hyperedge, and any wedge shares at most one leaf (non-central node) with any triangle or hyperedge.
    \item For any node $k$ that belongs to two distinct hyperedges $\set{i,j,k}, \set{k,l,m} \in \mathcal{H}$, no pairwise edge connects neighbors of $k$ in distinct hyperedges to each other, i.e. edges from $\set{i,j}$ to $\set{l,m}$ are forbidden.
  \end{enumerate}
\end{assumption}

\begin{figure}[hb!]
    \centering
    \includegraphics[width=250pt]{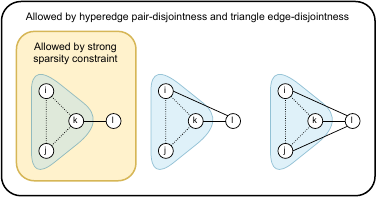}
    \includegraphics[width=430pt]{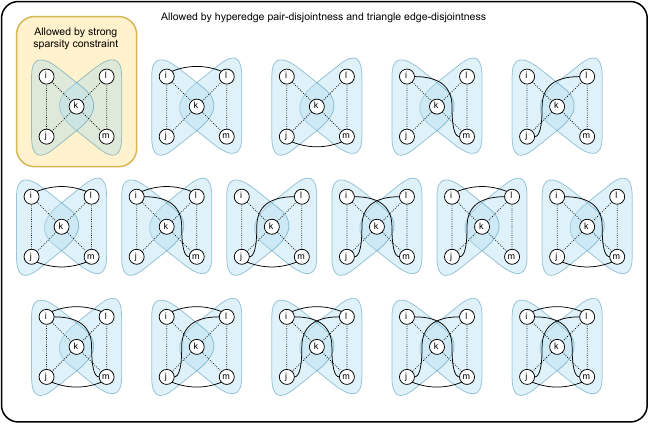}
    \caption{Allowed motif configurations under the explicit (weaker) assumptions of Assumption~\ref{ass:inter-order-stated-sparsity} on hyperedge pair-disjointedness (two hyperedges cannot share more than $1$ node) and triangle edge-disjointedness (two pairwise triangles cannot share more than $1$ node), versus the strong sparsity condition of Assumption~\ref{ass:io-strong-sparsity}. Hyperedges are represented by the filled blue borders; edges are represented by lines. Dotted lines indicate possible edges. The closures generated by the closure operator would produce different approximations for each motif.}
    \label{fig:inter-order-motifs}
\end{figure}

The first condition ensures that the pendant motif is the only four-node configuration containing a hyperedge: triangles and hyperedges interact only via shared single nodes, and the leaves of an external wedge cannot both lie within a hyperedge. The second condition restricts five-node configurations to the bowtie: when a node is shared between two hyperedges, the remaining nodes of one hyperedge cannot be pairwise-connected to those of the other. Figure~\ref{fig:inter-order-motifs} demonstrates the difference in allowed motifs under this strong sparsity assumption.
 
Intuitively, the strong sparsity assumption extends the disjointedness conditions of Assumption~\ref{ass:inter-order-stated-sparsity} in two ways. First, it forbids overlap \emph{between} hyperedges and triangles, not just within each type, and propagates this disjointedness through wedges: a wedge can share at most one leaf with any triangle or hyperedge. Second, it forbids pairwise connections between nodes belonging to distinct hyperedges that share a common node, so that two hyperedges meeting at a single node remain otherwise disconnected.

\begin{assumption}\label{ass:io-class-collapse}
  \emph{(Hyperedge-class collapse.)} For any three nodes $i, j, k$ such that $H_{ijk} = 1$, the joint state probability does not depend on the internal pairwise structure of the hyperedge:
  \[
    \prob{X_i Y_j Z_k}^{\bigcirc\sigma}
    \approx \prob{XYZ}^{\bigcirc}
    \quad \text{for all } \sigma \in \set{\triangle, \wedge, -, \bullet}.
  \]
\end{assumption}

\begin{remark}\label{rem:io-class-collapse-justification}
  The hyperedge-class collapse is a strong assumption: it asserts that the contribution of pairwise edges within a hyperedge to the joint state dynamics can be neglected relative to the hyperedge interaction itself. In practice, the closures~\eqref{eq:io-closure-pendant}--\eqref{eq:io-closure-bowtie} use a single hyperedge triplet probability $\prob{XYZ}^{\bigcirc}$ throughout, so without this collapse the system would not close at the level of the five tracked variables.
\end{remark}

With these assumptions in hand, we now proceed with the derivation. Our strategy, as before, is to start from our microscopic rate equations and derive the corresponding target rate equations. However, what sets this derivation apart is that we require additional parameterizations and assumptions in order to arrive at the exact form of the target rate equations. We demonstrate when and how these parameters and assumptions come into play.

\subsubsection{Network quantities} \label{subsubsec:io-network-quantities}

The derivation of the rate equations proceeds by evaluating each summation term in equations~\eqref{eq:dI}--\eqref{eq:dSII} under the mean-field assumptions, expressing the results in terms of the network parameters $k_1$, $k_2$, $\alpha$, $\theta$, and the isomorphism class-specific probabilities of Table~\ref{tab:isomorphism-class-notation}. 

A key part of the derivation involves expanding adjacency factors in front of joint state probabilities and identifying all isomorphism classes that arise. The resulting partial sums are evaluated under topological homogeneity (Assumption~\ref{ass:topological-homogeneity}), promoting them to complete sums that can be evaluated using $k_1$, $k_2$, $\alpha$, $\theta$, and $\phi$. We first establish four identities relating the products of $\mathbf{H}$, $\mathbf{A}$ and $\bar{\mathbf{A}}$ to the distribution $\theta$; these will then serve as the building blocks for the partial sums over $\leq 3$ nodes that arise in the derivation.

\begin{lemma}\label{lem:theta-tensor-products}
  Since $6|\mathcal{H}| = \sum_{i,j,k \in \mathcal{N}} H_{ijk}$, the following identities hold (all sums over $i,j,k \in \mathcal{N}$):
  \begin{subequations}\label{eq:theta-identities}
  \begin{align}
    \frac{
      \sum_{i,j,k} H_{ijk}\,
        A_{ij}\, A_{ik}\, A_{jk}
    }{
      \sum_{i,j,k} H_{ijk}
    }
      &= \theta(3),
      \label{eq:theta3} \\
    \frac{
      \sum_{i,j,k} H_{ijk}\,
        A_{ij}\, A_{ik}\, \bar{A}_{jk}
    }{
      \sum_{i,j,k} H_{ijk}
    }
      &= \tfrac{1}{3}\,\theta(2),
      \label{eq:theta2} \\
    \frac{
      \sum_{i,j,k} H_{ijk}\,
        A_{ij}\, \bar{A}_{ik}\, \bar{A}_{jk}
    }{
      \sum_{i,j,k} H_{ijk}
    }
      &= \tfrac{1}{3}\,\theta(1),
      \label{eq:theta1} \\
    \frac{
      \sum_{i,j,k} H_{ijk}\,
        \bar{A}_{ij}\, \bar{A}_{ik}\, \bar{A}_{jk}
    }{
      \sum_{i,j,k} H_{ijk}
    }
      &= \theta(0).
      \label{eq:theta0}
  \end{align}
  \end{subequations}
\end{lemma}

\begin{proof}
  For a fixed hyperedge $\set{i,j,k}$, we count how many of the $3! = 6$ ordered triples yield a non-zero contribution to each numerator. The product $A_{ij}\,A_{ik}\,A_{jk}$ is symmetric under all permutations, so all $6$ triples contribute, giving a factor of $6/6 = 1$. The product $A_{ij}\,A_{ik}\,\bar{A}_{jk}$ is symmetric only under $j \leftrightarrow k$, which gives a factor of $2/6 = 1/3$. Similarly, the product $A_{ij}\,\bar{A}_{ik}\,\bar{A}_{jk}$ is symmetric only under $i \leftrightarrow j$, yielding $2/6 = 1/3$. The product $\bar{A}_{ij}\,\bar{A}_{ik}\,\bar{A}_{jk}$ is fully symmetric, yielding $1$. Summing over hyperedges and normalising by $6|\mathcal{H}|$ gives the stated identities.
\end{proof}

With these identities in hand, we now compile the partial sums that appear in the rate equations. Each is conveniently cataloged by the isomorphism class it selects: for a fixed node or edge, the sum counts the number of motifs of a given class incident to it, and under topological homogeneity this local count equals the network-wide average.

\begin{lemma}\label{lem:io-network-quantities}
  Under topological homogeneity, the following partial sums hold:
  \begin{align}
    \bigcirc+\triangle :\qquad
    \sum_{\substack{x<y \\ x,y\not =i}}
    H_{ixy} A_{ix} A_{iy} A_{xy}
    &=k_2 \theta(3)
    \label{eq:io-hyperedge-triangle-condition-on-node} \\
    \bigcirc + \wedge \textrm{ ($x\leftrightarrow y$ symmetric)} :\qquad
    \sum_{\substack{x<y \\ x,y\not =i}}
    H_{ixy} A_{ix} A_{iy} \bar{A}_{xy}
    &=\frac{k_2}{3} \theta(2)
    \label{eq:io-hyperedge-wedge-condition-on-node-symm}\\
    \bigcirc + \wedge \textrm{ (not symmetric)} :\qquad
    \sum_{\substack{x<y \\ x,y\not =i}}
    H_{ixy} A_{ix} \bar{A}_{iy} A_{xy}
    &= \frac{k_2}{3} \theta(2)
    \label{eq:io-hyperedge-wedge-condition-on-node-asymm}\\
    \bigcirc + - \textrm{ ($x\leftrightarrow y$ symmetric)} :\qquad
    \sum_{\substack{x<y \\ x,y\not =i}}
    H_{ixy} \bar{A}_{ix} \bar{A}_{iy} A_{xy}
    &= \frac{k_2}{3} \theta(1)
    \label{eq:io-hyperedge-edge-condition-on-node-symm}\\
    \bigcirc + - \textrm{ (not symmetric)} :\qquad
    \sum_{\substack{x<y \\ x,y\not =i}}
    H_{ixy} \bar{A}_{ix} A_{iy} \bar{A}_{xy}
    &= \frac{k_2}{3} \theta(1)
    \label{eq:io-hyperedge-edge-condition-on-node-asymm}\\
    \bigcirc + \bullet :\qquad
    \sum_{\substack{x<y \\ x,y\not =i}}
    H_{ixy} \bar{A}_{ix} \bar{A}_{iy} \bar{A}_{xy}
    &= k_2 \theta(0)
    \label{eq:io-hyperedge-isonodes-condition-on-node} \\
    \textrm{conditioned on $A_{ij}$ } \bigcirc + \triangle :\qquad
    A_{ij} \sum_{x\not=i,j} 
    H_{ijx} A_{ix} A_{jx}
    &= \frac{2 k_2}{k_1} \theta(3)
    \label{eq:io-hyperedge-triangle-condition-on-edge} \\
    \textrm{conditioned on $A_{ij}$ } \bigcirc + \wedge :\qquad
    A_{ij} \sum_{x\not=i,j} 
    H_{ijx} A_{ix} \bar{A}_{jx}
    &= \frac{2 k_2}{3k_1} \theta(2)
    \label{eq:io-hyperedge-wedge-condition-on-edge} \\
    \textrm{conditioned on $A_{ij}$ } \bigcirc + - :\qquad
    A_{ij} \sum_{x\not=i,j} 
    H_{ijx} \bar{A}_{ix} \bar{A}_{jx}
    &=\frac{2 k_2}{3k_1} \theta(1)
    \label{eq:io-hyperedge-edge-condition-on-edge} \\
    \textrm{conditioned on $A_{ij}$ } \otimes + \triangle :\qquad
    A_{ij} \sum_{x\not=i,j} 
    \bar{H}_{ijx} A_{ix} A_{jx} 
    &= (k_1 - 1) \phi -\frac{2k_2}{k_1} \theta(3)
    \label{eq:io-free-triangle-condition-on-edge} \\
    \textrm{conditioned on $A_{ij}$ } \otimes + \wedge :\qquad
    A_{ij} \sum_{x\not=i,j} 
    \bar{H}_{ijx} A_{ix} \bar{A}_{jx}
    &= (k_1 - 1) (1 - \phi) - \frac{2k_2}{3k_1} \theta(2)
    \label{eq:io-free-wedge-condition-on-edge} 
  \end{align}
\end{lemma}

\begin{proof}
    Each of these partial sums has a natural interpretation: given a small motif fixed by the terms outside the summation, count the number of larger motifs (described by the terms inside the summation) that contain it. For example, equation~\eqref{eq:io-hyperedge-triangle-condition-on-edge} fixes an arbitrary edge $\set{i,j}$ via the $A_{ij}$ prefactor, and the summation counts the number (or in the mean-field case, the expected number) of simplicial complexes (hyperedge $+$ pairwise triangle) that this edge belongs to. 

    To evaluate each of the partial sums, we assume topological homogeneity (Assumption~\ref{ass:topological-homogeneity}) to promote these partial sums into an expression involving a complete sum.
    The definition of $k_1$ and $k_2$, as well as the $\theta$ identities stated in Lemma~\ref{lem:theta-tensor-products} are used to evaluate the complete sums. 
    
    For equations~\eqref{eq:io-hyperedge-triangle-condition-on-node}--\eqref{eq:io-hyperedge-isonodes-condition-on-node}, we must first account for the change in summation conventions from ordered to unordered. Where the dummy indices $x$ and $y$ are symmetric under exchange, this simply introduces a factor of $\frac{1}{2}$ to account for double counting. Furthermore, there is only one fixed node index $i$, and all adjacency terms involve at least one dummy node index, and cannot be factored out of the summation. Thus, recalling the sum promotion formula
    \[
        p(i, j)
    \sum_{x, y} q(i, j, x, y)
    \approx
    \frac{
      \displaystyle\sum_{i, j, x, y}
        p(i, j)\, q(i, j, x, y)
    }{
      \displaystyle\sum_{i, j} p(i, j)
    },
    \]
    we can identify $p(i)=1$ and $q(i,x,y)$ as the terms inside the summation. The expressions $\sum_{i,x,y} q(i,x,y)$ correspond directly to the $\theta$ identities in equation~\eqref{eq:theta-identities}. For example, the derivation for equation~\eqref{eq:io-hyperedge-triangle-condition-on-node} is
   
    \begin{align*}
        \sum_{\substack{x<y \\ x,y\not =i}} H_{ixy} A_{ix} A_{iy} A_{xy} 
        = \frac{1}{2} \sum_{x,y} H_{ixy} A_{ix} A_{iy} A_{xy} 
        &= \frac{1}{2} \dfrac{\sum_{i,x,y} H_{ixy} A_{ix} A_{iy} A_{xy}}{\sum_{i} 1 } \\
        &= \frac{1}{2N} \underbrace{2Nk_2 \theta(3)}_{6|\mathcal{H}| \theta(3) \textrm{ from eq.~\ref{eq:theta3}}}
    \end{align*}

    For equations~\eqref{eq:io-hyperedge-triangle-condition-on-edge}--\eqref{eq:io-hyperedge-edge-condition-on-edge}, the indices $i,j$ are fixed and the index $x$ is the dummy index. Thus, we identify $p(i,j)=A_{ij}$ and $q(i,j,x)$ as the remaining terms inside the summation. Using the sum promotion equation, we have $\sum_{i,j} A_{ij} = N k_1$, and likewise $\sum_{i,j,x} q(i,j,x)$  correspond directly to the $\theta$ identities in equation~\eqref{eq:theta-identities}.

    For the asymmetric cases---equations~\eqref{eq:io-hyperedge-wedge-condition-on-node-asymm} and~\eqref{eq:io-hyperedge-edge-condition-on-node-asymm}---the integrand is not symmetric under $x \leftrightarrow y$, so the ordered-to-unordered conversion $\sum_{x<y} = \frac{1}{2}\sum_{\substack{x,y \\ y \neq x}}$ does not apply directly. Instead, we note that the ordered sum and its mirror (obtained by swapping $x \leftrightarrow y$ in the adjacency factors) together equal the full unordered sum:
    \[
        \sum_{\substack{x<y}} H_{ixy}\, A_{ix}\, \bar{A}_{iy}\, A_{xy}
        + \sum_{\substack{x<y}} H_{ixy}\, \bar{A}_{ix}\, A_{iy}\, A_{xy}
        = \sum_{\substack{x,y \\ y \neq x}} H_{ixy}\, A_{ix}\, \bar{A}_{iy}\, A_{xy}.
    \]
    Since the two ordered sums are equal (by relabelling $x \leftrightarrow y$ and using the symmetry of $\mathbf{H}$), each equals half the promoted unordered sum, yielding the same value as the symmetric case.

    For equations~\eqref{eq:io-free-triangle-condition-on-edge} and \eqref{eq:io-free-wedge-condition-on-edge}, we use the substitution that $\bar{H}_{ijx} = 1 - H_{ijx}$. The term with $H_{ijx}$ is exactly the one derived in equation~\eqref{eq:io-hyperedge-triangle-condition-on-edge} and~\eqref{eq:io-hyperedge-wedge-condition-on-edge}, respectively. This leaves the term with $1$, which is exactly equal to $k_1 (k_1 - 1) \phi$ or $k_1 (k_1-1) (1-\phi)$, where $\phi$ is the global clustering coefficient (Definition~\ref{def:clustering}).
\end{proof}

\subsubsection{Node-level derivation}  \label{subsubsec:io-single}
For the pairwise infection term in equation~\eqref{eq:dI}, the factor $A_{ix}$ selects the connected pair class (Table~\ref{tab:isomorphism-class-notation}), so we drop the node labels and factor the probability out of the sum:
\begin{equation}\label{eq:io-pairwise}
  \sum_{x \neq i} A_{ix}\, \prob{S_i I_x}
  \approx \prob{SI}^{-} \sum_{x \neq i} A_{ix}
  \approx k_1\, \prob{SI}.
\end{equation}
The first step uses the dynamical homogeneity assumption (Assumption~\ref{ass:dynamical-homogeneity}) and the second uses the standard result for the node degree. Since pair probabilities accompanied by an adjacency factor always refer to the connected class, we suppress the $-$ superscript and write simply $\prob{SI}$.

The hyperedge infection term requires more care, as the joint probability depends on the isomorphism class of the sub-graph spanned by $\set{i,x,y}$.

\begin{lemma}\label{lem:io-hyperedge-expansion}
  Under Assumption~\ref{ass:dynamical-homogeneity} (dynamical homogeneity) and Assumption~\ref{ass:topological-homogeneity} (topological homogeneity), the hyperedge infection term in equation~\eqref{eq:dI} satisfies
  \begin{equation}\label{eq:io-hyperedge}
    \begin{aligned}
      \sum_{\substack{x < y \\ x,y \neq i}}
        H_{ixy}\prob{S_i I_x I_y}
      \approx k_2 \Bigl[ &
        \theta(3)\, \prob{SII}^{\bigcirc\triangle}
        + \tfrac{1}{3}\theta(2)\, \prob{I{S}I}^{\bigcirc\wedge}
        + \tfrac{2}{3}\theta(2)\, \prob{S{I}I}^{\bigcirc\wedge}\\
        &+ \tfrac{1}{3}\theta(1)\, \prob{I{S}I}^{\bigcirc-}
        + \tfrac{2}{3}\theta(1)\, \prob{S{I}I}^{\bigcirc-}
        + \theta(0)\, \prob{SII}^{\bigcirc\,\bullet}
      \Bigr].
    \end{aligned}
  \end{equation}
\end{lemma}

\begin{proof}
  We first expand and identify the isomorphism classes. 
  Inserting the identity $(A_{ix} + \bar{A}_{ix})(A_{iy} + \bar{A}_{iy})(A_{xy} + \bar{A}_{xy}) = 1$ into each term of the sum and noting that $H_{ixy}\,\bar{H}_{ixy} = 0$ eliminates the four $\otimes$ cases. Expanding the remaining product yields eight terms grouped by the number of edges present: $1$~triangle term, $3$~wedge terms, $3$~single-edge terms, and $1$~empty term. By dynamical homogeneity, the joint probability in each term depends only on its isomorphism class. However, within the wedge and single-edge groups, the role of node~$i$ differs across the three permutations:
  \begin{itemize}
    \item In the wedge term $A_{ix}\,A_{iy}\,\bar{A}_{xy}$, node~$i$ is the center, giving $\prob{I{S}I}^{\bigcirc\wedge}$; in the other two ($A_{ix}\,\bar{A}_{iy}\,A_{xy}$ and $\bar{A}_{ix}\,A_{iy}\,A_{xy}$), node~$i$ is a leaf, giving $\prob{S{I}I}^{\bigcirc\wedge}$.
    \item In the single-edge term $\bar{A}_{ix}\,\bar{A}_{iy}\,A_{xy}$, node~$i$ is isolated, giving $\prob{I{S}I}^{\bigcirc-}$; in the other two ($A_{ix}\,\bar{A}_{iy}\,\bar{A}_{xy}$ and $\bar{A}_{ix}\,A_{iy}\,\bar{A}_{xy}$), node~$i$ is an endpoint of the edge, giving $\prob{S{I}I}^{\bigcirc-}$.
  \end{itemize}
  Grouping accordingly:
  \begin{equation}\label{eq:io-hyperedge-expanded}
  \begin{aligned}
    \sum_{\substack{x < y \\ x,y \neq i}} &
      H_{ixy}\prob{S_i I_x I_y} \\
    \approx{}&
      \prob{SII}^{\bigcirc\triangle}
      \sum_{\substack{x < y \\ x,y \neq i}}
        H_{ixy}\, A_{ix}\, A_{iy}\, A_{xy} \\
    &+ \sum_{\substack{x < y \\ x,y \neq i}}
        H_{ixy} \bigl(
          A_{ix}\, A_{iy}\, \bar{A}_{xy}\; \prob{I{S}I}^{\bigcirc\wedge}
        + [A_{ix}\, \bar{A}_{iy}\, A_{xy}
        + \bar{A}_{ix}\, A_{iy}\, A_{xy}]\; \prob{S{I}I}^{\bigcirc\wedge}
        \bigr) \\
    &+ \sum_{\substack{x < y \\ x,y \neq i}}
        H_{ixy} \bigl(
          \bar{A}_{ix}\, \bar{A}_{iy}\, A_{xy}\; \prob{I{S}I}^{\bigcirc-}
        + [A_{ix}\, \bar{A}_{iy}\, \bar{A}_{xy}
        + \bar{A}_{ix}\, A_{iy}\, \bar{A}_{xy}]\; \prob{S{I}I}^{\bigcirc-}
        \bigr) \\
    &+ \prob{SII}^{\bigcirc\,\bullet}
      \sum_{\substack{x < y \\ x,y \neq i}}
        H_{ixy}\, \bar{A}_{ix}\, \bar{A}_{iy}\, \bar{A}_{xy}.
  \end{aligned}
  \end{equation}

  The prefactors for each isomorphism class directly map onto the network quantities listed in Lemma~\ref{lem:io-network-quantities}---specifically, those in equations~\eqref{eq:io-hyperedge-triangle-condition-on-node}--\eqref{eq:io-hyperedge-isonodes-condition-on-node}.
\end{proof}

\begin{remark}\label{rem:io-naive-limit}
  If we treat all joint probabilities in equation~\eqref{eq:io-hyperedge} as equivalent, denoting their common value by $\prob{SII}^{\bigcirc}$, then the coefficients regroup as $\theta(3) + (\frac{1}{3} + \frac{2}{3})\theta(2) + (\frac{1}{3} + \frac{2}{3})\theta(1) + \theta(0) = \sum_{m=0}^{3} \theta(m) = 1$, and the right-hand side reduces to $k_2\,\prob{SII}^{\bigcirc}$. This is the naive mean-field result one would obtain by treating all hyperedge infections as equivalent regardless of the internal pairwise structure. The full expression reveals that the dynamics depend not only on the distribution of edges within hyperedges (through~$\theta$), but also on the position of the susceptible node within the sub-graph topology.
\end{remark}

\subsubsection{Pair-level derivation} \label{subsubsec:io-pair}
We now evaluate the summation terms in the pair equation~\eqref{eq:dSI} under dynamical and topological homogeneity, with $A_{ij} = 1$ enforced throughout. The pair equation contains five summation terms: two external pairwise infection terms, one mixed hyperedge infection term, and two external hyperedge infection terms. We treat each in turn.

\begin{lemma}\label{lem:io-pair-beta1}
  Under Assumption~\ref{ass:dynamical-homogeneity} and Assumption~\ref{ass:topological-homogeneity}, with $A_{ij} = 1$, the first external pairwise infection term in equation~\eqref{eq:dSI} satisfies
  \begin{equation}\label{eq:io-pair-beta1}
  \begin{aligned}
    A_{ij} \sum_{x \neq i,j} A_{ix}\, \prob{S_i I_j I_x}
    \approx{}&
      {\frac{2k_2\,\theta(3)}{k_1}}\,
        \prob{SII}^{\bigcirc\triangle}
    + {\frac{2k_2\,\theta(2)}{3k_1}}\,
        \prob{I{S}I}^{\bigcirc\wedge} \\
    &+ \left[
        (k_1 - 1)\phi
        - {\frac{2k_2\,\theta(3)}{k_1}}
      \right] \prob{SII}^{\otimes\triangle} \\
    &+ \left[
        (k_1 - 1)(1 - \phi)
        - {\frac{2k_2\,\theta(2)}{3k_1}}
      \right] \prob{I{S}I}^{\otimes\wedge}.
  \end{aligned}
  \end{equation}
  The four coefficients sum to $k_1 - 1$, the mean number of neighbors of~$i$ excluding~$j$.
\end{lemma}

\begin{proof}
  Since $A_{ij} = 1$ and $A_{ix} = 1$ are enforced by the prefactors, the sub-graph on $\set{i, j, x}$ always contains at least the two edges $\set{i,j}$ and $\set{i,x}$. Inserting $(A_{jx} + \bar{A}_{jx})(H_{ijx} + \bar{H}_{ijx}) = 1$ yields four terms corresponding to triangle or wedge configurations, each with or without a hyperedge. By dynamical homogeneity, the joint probability in each term depends only on the isomorphism class. In the wedge cases ($\bar{A}_{jx}$), node~$i$ is connected to both~$j$ and~$x$, making it the center; its state~$S$ therefore occupies the middle entry. In the triangle cases, the ordering is irrelevant. Grouping accordingly:
  \begin{equation}\label{eq:io-pair-beta1-expanded}
  \begin{aligned}
    A_{ij} \sum_{x \neq i,j} &
      A_{ix}\, \prob{S_i I_j I_x} \\
    \approx{}&
      \prob{SII}^{\bigcirc\triangle}
        \sum_{x \neq i,j}
          A_{ij}\, A_{ix}\, A_{jx}\, H_{ijx}
    + \prob{I{S}I}^{\bigcirc\wedge}
        \sum_{x \neq i,j}
          A_{ij}\, A_{ix}\, \bar{A}_{jx}\, H_{ijx} \\
    &+ \prob{SII}^{\otimes\triangle}
        \sum_{x \neq i,j}
          A_{ij}\, A_{ix}\, A_{jx}\, \bar{H}_{ijx}
    + \prob{I{S}I}^{\otimes\wedge}
        \sum_{x \neq i,j}
          A_{ij}\, A_{ix}\, \bar{A}_{jx}\, \bar{H}_{ijx}.
  \end{aligned}
  \end{equation}

  The corresponding factors for each isomorphism class are precisely the same network quantities listed in Lemma~\ref{lem:io-network-quantities}---specifically, those in equations~\eqref{eq:io-hyperedge-triangle-condition-on-edge}--\eqref{eq:io-free-wedge-condition-on-edge}.

\end{proof}

The second external pairwise infection term in equation~\eqref{eq:dSI}, namely $A_{ij} \sum_{x \neq i,j} A_{jx}\, \prob{S_i S_j I_x}$, follows from Lemma~\ref{lem:io-pair-beta1} by the substitution $i \leftrightarrow j$ with corresponding state relabeling $I_i \to S_i$, $S_j \to S_j$, and yields the same structure with node~$j$ playing the central role.

\begin{remark}
    The corresponding coefficient in the target equation~\eqref{eq:io-target-pair} is $(k_1-1) \prob{I S I}$, where $\prob{ISI}$ is treated as wedges in the closure decomposition. This requires assuming that both $\frac{k_2}{k_1}\ll 1$ and $\phi \ll (1-\phi)$. This is supported in~\cite{malizia2026nested} by asserting that closed triplets of edges are not allowed in the model, which deals with $\phi \ll 1$, but it allows for exceptions for hyperedges in the case where $\alpha =1$.
\end{remark}

\begin{lemma}\label{lem:io-pair-beta2-single}
  Under Assumption~\ref{ass:dynamical-homogeneity} and Assumption~\ref{ass:topological-homogeneity}, with $A_{ij} = 1$, the mixed hyperedge infection term in equation~\eqref{eq:dSI} satisfies
  \begin{equation}\label{eq:io-pair-beta2-single}
  \begin{aligned}
    A_{ij} \sum_{x \neq i,j} H_{ijx}\, \prob{S_i I_j I_x}
    \approx \frac{2k_2}{k_1} \bigl( &
      \theta(3)\, \prob{SII}^{\bigcirc\triangle}
    + \tfrac{1}{3}\theta(2)\, \prob{I{S}I}^{\bigcirc\wedge} \\
    &+ \tfrac{1}{3}\theta(2)\, \prob{S{I}I}^{\bigcirc\wedge}
    + \tfrac{1}{3}\theta(1)\, \prob{S{I}I}^{\bigcirc-}
    \bigr).
  \end{aligned}
  \end{equation}
\end{lemma}

\begin{proof}
  Since $H_{ijx}$ is already a factor in the sum, only four $\bigcirc$ isomorphism classes survive. Inserting the identity $(A_{ix} + \bar{A}_{ix})(A_{jx} + \bar{A}_{jx}) = 1$ and applying dynamical homogeneity, we group by isomorphism class:
  \begin{equation}\label{eq:io-pair-beta2-single-expanded}
  \begin{aligned}
    A_{ij} \sum_{x \neq i,j} &
      H_{ijx}\, \prob{S_i I_j I_x} \\
    \approx{}&
      \prob{SII}^{\bigcirc\triangle}
        \sum_{x \neq i,j}
          A_{ij}\, H_{ijx}\, A_{ix}\, A_{jx} \\
    &+ \prob{I{S}I}^{\bigcirc\wedge}
        \sum_{x \neq i,j}
          A_{ij}\, H_{ijx}\, A_{ix}\, \bar{A}_{jx} \\
    &+ \prob{S{I}I}^{\bigcirc\wedge}
        \sum_{x \neq i,j}
          A_{ij}\, H_{ijx}\, \bar{A}_{ix}\, A_{jx} \\
    &+ \prob{S{I}I}^{\bigcirc-}
        \sum_{x \neq i,j}
          A_{ij}\, H_{ijx}\, \bar{A}_{ix}\, \bar{A}_{jx}.
  \end{aligned}
  \end{equation}
  The solution to the factors corresponding to each isomorphism class is listed in Lemma~\ref{lem:io-network-quantities}---specifically equations~\eqref{eq:io-hyperedge-triangle-condition-on-edge}--\eqref{eq:io-hyperedge-edge-condition-on-edge}.
\end{proof}

\begin{remark}\label{rem:io-pair-beta2-alpha}
  If we treat all joint probabilities in equation~\eqref{eq:io-pair-beta2-single} as equivalent, denoting their common value by $\prob{SII}^{\bigcirc}$, then the coefficient reduces to
  \[
    \frac{2k_2}{k_1} \bigl[
      \theta(3) + \tfrac{1}{3}\theta(2)
      + \tfrac{1}{3}\theta(2) + \tfrac{1}{3}\theta(1)
      + 0 \cdot \theta(0)
    \bigr]
    = \frac{2k_2}{k_1}\,\alpha,
  \]
  where the last equality follows from Lemma~\ref{lem:alpha-from-theta}, since $\theta(3) + \frac{2}{3}\theta(2) + \frac{1}{3}\theta(1) = \frac{1}{3}[3\theta(3) + 2\theta(2) + \theta(1)] = \alpha$. This corresponds directly to the coefficients in equation~\eqref{eq:io-target-pair}.
\end{remark}

We have now established the machinery to close the inter-order model. The node-level and pair-level partial sums in Lemma~\ref{lem:io-network-quantities} demonstrate two things. First, a precise derivation starting from the microscopic rate equations~\eqref{eq:dI}--\eqref{eq:dSII} and using only the explicitly stated structural assumptions (Assumption~\ref{ass:inter-order-stated-sparsity}) requires the additional parameterisation $\theta$ to resolve the isomorphism classes of hyperedge triplets. Second, Figure~\ref{fig:inter-order-motifs} shows that the number of pendant and bowtie closures stated in~\cite{malizia2026nested} do not cover the full range of four- and five-node motifs permitted under Assumption~\ref{ass:inter-order-stated-sparsity}. The strong sparsity Assumption~\ref{ass:io-strong-sparsity} restricts the admissible motifs to exactly those covered by the stated closures, and the hyperedge-class collapse Assumption~\ref{ass:io-class-collapse} allows the $\theta$-resolved expressions to collapse onto the single hyperedge triplet probability $\prob{XYZ}^{\bigcirc}$ tracked by the model. Together, these implicit assumptions are sufficient to close the system.

\begin{theorem}\label{thm:io-recovery}
  Under the stated sparsity assumption (Assumption~\ref{ass:inter-order-stated-sparsity}), the strong sparsity assumption (Assumption~\ref{ass:io-strong-sparsity}), and the hyperedge-class collapse assumption (Assumption~\ref{ass:io-class-collapse}), together with dynamical and topological homogeneity (Assumptions~\ref{ass:dynamical-homogeneity} and~\ref{ass:topological-homogeneity}), the microscopic rate equations~\eqref{eq:dI}--\eqref{eq:dSII} reduce to the inter-order target equations~\eqref{eq:io-target-eqs}, with closures given by~\eqref{eq:io-closure-triples}--\eqref{eq:io-closure-bowtie}.
\end{theorem}

\begin{proof}
  The node-level equation~\eqref{eq:io-target-node} and the recovery, internal pairwise infection, external pairwise infection, and mixed hyperedge infection terms of the pair-level equation~\eqref{eq:io-target-pair} follow from the partial sums of Lemma~\ref{lem:io-network-quantities} together with Assumption~\ref{ass:io-class-collapse} (derivation at the node-level in subsubsection~\ref{subsubsec:io-single} and the pair level in subsubsection~\ref{subsubsec:io-pair}). The remaining terms---the external hyperedge infection term in the pair-level equation, which involves a four-node motif, and the entirety of the triplet-level equations~\eqref{eq:io-target-sss}--\eqref{eq:io-target-isi}, which involve four- and five-node motifs---require additional supporting lemmas analogous to Lemma~\ref{lem:io-network-quantities} but for partial sums conditioned on a fixed hyperedge or a fixed pair of hyperedges. These lemmas, together with their proofs and the term-by-term construction of each rate equation, are provided in Appendix~\ref{app:inter-order-derivations}.
\end{proof}

\section{Discussion} \label{sec:discussion}
We have introduced a closure operator that, given the adjacency data of any local set of nodes, returns a pseudo-probability for the motif's state as a product of lower-order joint probabilities. Applying this operator to four- and five-node motif terms in the exact microscopic SIS rate equations on a higher-order network yields a closed system of equations expressed entirely in terms of node, pair, and triplet level joint probabilities. We close this paper by demonstrating that this framework can (a) recover known existing models by imposing structural constraints and assuming dynamical homogeneity, and (b) identify hidden assumptions of existing models and suggest potential avenues to reconcile them. Specifically, the pair-based model of Malizia et al.~\cite{malizia2025pair} and the maximal-clique model of Burgio et al.~\cite{burgio2024triadic} emerge as direct outputs of our closure operator under the structural and mean-field assumptions made explicit in our derivation.

The exercise also surfaces assumptions that were not stated in the original derivations of these models. The clearest case is the inter-order model of~\cite{malizia2026nested}, whose governing rate equations show good agreement with stochastic simulation but, when derived from the first-principles microscopic rate equation via our closure methodology, are seen to require strictly stronger conditions than those reported: in particular, a sparsity condition that is stronger than the version explicitly stated in their paper, as well as a homogeneity condition on the hyperedge dynamics. In this case, it appears that the minor deviations that would arise from these lost assumptions did not affect the numerical experiments of the paper. Nonetheless, identifying these is useful as they offer potential new avenues forward. We proposed that these assumptions can be lifted by introducing finer topological quantities: in addition to the inter-order overlap parameter $\alpha$, the underlying distribution $\theta$, of which $\alpha$ is the first moment, allows the hyperedge dynamics to be resolved at the level of individual motifs, removing the requirement of hyperedge dynamic homogeneity.

A further consequence of the framework is that it is constructive: the closure procedure does not only recover existing models, it opens up the space of possible models by removing the combinatorial wall of motifs that blocks the analysis of higher-order dynamics. Characterizing the state probabilities is no longer an issue; instead, for mean-field models, the challenge turns to properly characterizing the frequency or density of these motifs on any particular network. Our topological homogeneity framework provides a procedure for identifying the global network parameters required. The procedure has three steps: enumerate all isomorphism classes of motifs of four and five nodes; write the linear relations between their densities using the complement identities $\bar{A}_{ij} = 1 - A_{ij}$ and $\bar{H}_{ijk} = 1 - H_{ijk}$; and count the degrees of freedom of the resulting system to identify the number of global parameters that must be characterized. This makes the framework a generator of new models on novel network topologies, not only a tool for recovering existing ones.

Several limitations should be flagged. Firstly, the Kirkwood superposition approximation underpins all of the closures, and therefore regimes where the Kirkwood superposition approximation is a poor fit will affect the framework downstream. Secondly, the dynamics worked through in this paper are restricted to SIS, but in principle the closure operator is indifferent to the choice of dynamical model: it can be extended to other epidemiological models (SIR, SEIRD, etc.) or other network dynamics entirely (voter, Kuramoto, etc.). Thirdly, the hyperedge size is limited to three in this paper. Larger hyperedges would be a natural extension, particularly to hyperedges of arbitrary size; one would need to generalize the hyperedge operator in the closure procedure, which is further complicated when multiple hyperedge sizes are involved. In principle, however, the basic idea of the construction (using truncated Kirkwood approximations) can be extended to truncate more than two orders below, which would enable larger hyperedges in the model. Fourthly, this model closes at the level of triplets. While this is currently in line with the latest research, in principle the closure could be extended beyond third moments, applying the closure operator to close the system in the same manner. Finally, we have not undertaken computational tests of the framework, although the analytic derivations of existing models serve as a validation.


\section*{Acknowledgments}
K.T acknowledges the PhD studentship support from Northeastern University. The authors declare no competing interests.

\appendix

\section{BMC Model: second-moment derivation} \label{app:bmc-order-2-derivation}
 
We provide the full term-by-term derivation of Theorem~\ref{thm:burgio-order2}. We multiply both sides of equation~\eqref{eq:dSI} by~$B^{(1)}_{ij}$, which restricts attention to maximal $2$-cliques. By Proposition~\ref{prop:burgio-constraints}, $B^{(1)}_{ij}\, A_{ij} = (B^{(1)}_{ij})^2 = B^{(1)}_{ij}$, so the prefactor simplifies throughout.
 
\medskip
\textit{Recovery and internal pairwise infection terms.} The recovery terms $-\gamma\, B^{(1)}_{ij}\prob{S_i I_j} + \gamma\, B^{(1)}_{ij}\prob{I_i I_j}$ and the internal pairwise infection term $-\beta_1\, B^{(1)}_{ij}\prob{S_i I_j}$ pass through unchanged since $B^{(1)}_{ij} \in \set{0,1}$.
 
\textit{Mixed hyperedge infection term.} This term contains the factor $B^{(1)}_{ij}\, H_{ijx} = B^{(1)}_{ij}\, B^{(0,1)}_{ijx}$. By the maximality constraint (equation~\eqref{eq:burgio-maximality}), $B^{(1)}_{ij}\, B^{(0,1)}_{ijk} = 0$ for all~$k$, so the entire mixed hyperedge infection term vanishes.
 
\textit{External hyperedge infection terms.} The terms $\beta_2 B^{(1)}_{ij}\, B^{(0,1)}_{ixy}$ and $\beta_2 B^{(1)}_{ij}\, B^{(0,1)}_{jxy}$ survive because the hyperedges involve $\set{i,x,y}$ or $\set{j,x,y}$, not the pair $\set{i,j}$, so the maximality constraint does not apply. These appear directly after the $\sum_{x<y} \to \frac{1}{2}\sum_{x,y}$ conversion.
 
\textit{Intermediate result.} After Steps above, the intermediate result is
\begin{equation}\label{eq:burgio-order2-intermediate}
\begin{aligned}
  \frac{d}{dt} B^{(1)}_{ij}\prob{S_i I_j}
  ={}&
    \underbracket{
      -\gamma\, B^{(1)}_{ij}\prob{S_i I_j}
      + \gamma\, B^{(1)}_{ij}\prob{I_i I_j}
    }{recovery}
    \underbracket{
      \;-\; \beta_1\, B^{(1)}_{ij}\prob{S_i I_j}
    }{internal pairwise infection} \\
  &\underbracket{
    \;-\; \beta_1 \sum_{x \neq i,j}
      B^{(1)}_{ij}\, A_{ix}\,\prob{S_i I_j I_x}
    \;+\; \beta_1 \sum_{x \neq i,j}
      B^{(1)}_{ij}\, A_{jx}\,\prob{S_i S_j I_x}
  }{external pairwise infection} \\
  &\underbracket{
    \;-\; \frac{\beta_2}{2} \sum_{\substack{x,y \\ y \neq x}}
      B^{(1)}_{ij}\, B^{(0,1)}_{ixy}\,\prob{I_x I_y S_i I_j}
    \;+\; \frac{\beta_2}{2} \sum_{\substack{x,y \\ y \neq x}}
      B^{(1)}_{ij}\, B^{(0,1)}_{jxy}\,\prob{S_i S_j I_x I_y}
  }{external hyperedge infection}.
\end{aligned}
\end{equation}
 
\textit{External pairwise infection terms.} All terms except the external pairwise infection term already match equation~\eqref{eq:burgio-target-order2}. Applying Lemma~\ref{lem:burgio-substitution} to the first external pairwise infection term:
\begin{equation}\label{eq:burgio-order2-beta1}
\begin{aligned}
  -\beta_1 &\sum_{x \neq i,j}
    B^{(1)}_{ij}\, A_{ix}\,\prob{S_i I_j I_x} \\
  =& -\beta_1
    \Bigg[
      \sum_{x \neq i,j}
        B^{(1)}_{ij}\, B^{(1)}_{ix}\,\prob{S_i I_j I_x} \\
      &\quad + \sum_{\substack{x,y \\ y \neq x}}
        B^{(1)}_{ij}\, B^{(1,0)}_{ixy}\,
          \frac{1}{2}\Big(
            \prob{S_i I_j I_x S_y}
            + \prob{S_i I_j S_x I_y}
            + 2\,\prob{S_i I_j I_x I_y}
          \Big)
    \Bigg].
\end{aligned}
\end{equation}
The second external pairwise infection term (infection pressure from neighbors of~$j$) follows by the substitution $[I_j \to S_j,\; i \leftrightarrow j]$ applied to equation~\eqref{eq:burgio-order2-beta1}, i.e.\ swapping the state of~$j$ from~$I$ to~$S$ and exchanging the indices $i$ and~$j$:
\[
  +\beta_1 \Big[I_j \to S_j,\; i \leftrightarrow j\Big].
\]
With this expansion, equation~\eqref{eq:burgio-target-order2} is recovered, matching equation~(2b) of~\cite{burgio2024triadic}.

\section{BMC model: third-moment derivation}\label{app:bmc-order-3-derivation}

We provide the full term-by-term derivation of Theorem~\ref{thm:burgio-order3}. Throughout, we multiply the triplet equations~\eqref{eq:dSSI} and~\eqref{eq:dSII} by $B^{(2)}_{ijk}$, where
\begin{equation}\label{eq:B2-def}
    B^{(2)}_{ijk} = B^{(1,0)}_{ijk} + B^{(0,1)}_{ijk} - B^{(1,0)}_{ijk}\,B^{(0,1)}_{ijk},
\end{equation}
selects triplets that belong to at least one maximal $3$-clique (triangle, hyperedge, or simplicial complex) without double-counting simplicial complexes.

The following identities, which extend Proposition~\ref{prop:burgio-constraints}, are used throughout.
\begin{subequations}\label{eq:B2-identities}
\begin{align}
     B^{(2)}_{ijk}\, B^{(1)}_{ij} &= 0,
     \label{eq:B2-maximality} \\
     B^{(2)}_{ijk}\, B^{(\psi)}_{ijx} &= B^{(\psi)}_{ijk}\, \delta_{xk},
     \quad \text{for } (\psi) \in \set{(1,0),\, (0,1),\, (1,1)},
     \label{eq:B2-sparsity} \\
     B^{(2)}_{ijk}\, A_{ij} &= B^{(2)}_{ijk} \Big[ B^{(1)}_{ij} + \sum_x B^{(1,0)}_{ijx} \Big] = B^{(1,0)}_{ijk}.
     \label{eq:B2-edge}
\end{align}
\end{subequations}
Equation~\eqref{eq:B2-maximality} follows from the maximality constraint: if $B^{(2)}_{ijk} = 1$, then $\set{i,j,k}$ is a maximal $3$-clique and the edge $\set{i,j}$ cannot simultaneously be a maximal $2$-clique. Equation~\eqref{eq:B2-sparsity} follows from sparsity: two distinct $3$-cliques sharing the pair $\set{i,j}$ would require $x = k$. Equation~\eqref{eq:B2-edge} follows by substituting the decomposition~\eqref{eq:burgio-A-decomp} and applying~\eqref{eq:B2-maximality} and~\eqref{eq:B2-sparsity}.

\subsection{Derivation of SSI equation} \label{app:burgio-SSI}

Multiplying equation~\eqref{eq:dSSI} by $B^{(2)}_{ijk}$ and simplifying term by term:

\medskip
\textit{Recovery terms.} The recovery terms pass through unchanged:
\[
    \gamma\, B^{(2)}_{ijk} \Big( \prob{I_i S_j I_k} + \prob{S_i I_j I_k} - \prob{S_i S_j I_k} \Big).
\]

\textit{Internal pairwise infection terms.} The prefactor $B^{(2)}_{ijk}[A_{ik} + A_{jk}]$ simplifies via equation~\eqref{eq:B2-edge}: $B^{(2)}_{ijk}\,A_{ik} = B^{(1,0)}_{ijk}$ and $B^{(2)}_{ijk}\,A_{jk} = B^{(1,0)}_{ijk}$, giving
\[
    -\beta_1 \cdot 2\,B^{(1,0)}_{ijk}\,\prob{S_i S_j I_k}.
\]

\textit{Internal hyperedge infection term.} This term is absent from equation~\eqref{eq:dSSI} (the internal hyperedge $H_{ijk}$ does not contribute to $\prob{S_i S_j I_k}$ since neither susceptible node is infected by the hyperedge when only $k$ is infected).

\textit{Mixed hyperedge infection terms.} The prefactors $B^{(2)}_{ijk}\,H_{ikx}$ and $B^{(2)}_{ijk}\,H_{jkx}$ require two $3$-cliques to share the pair $\set{i,k}$ or $\set{j,k}$ respectively. By the sparsity constraint~\eqref{eq:B2-sparsity}, this forces $x = j$ or $x = i$, but $x \neq i,j$ in the summation. Hence these terms vanish.

\textit{External hyperedge infection terms.} Substituting $H_{ixy} = B^{(0,1)}_{ixy}$, $H_{jxy} = B^{(0,1)}_{jxy}$, and $H_{kxy} = B^{(0,1)}_{kxy}$ via equation~\eqref{eq:burgio-H-decomp}, and converting to unordered sums (Remark~\ref{rem:summation-convention}):
\[
    \beta_2 \bigg( \frac{1}{2} \sum_{\substack{x,y \\ y \neq x,\; x,y \neq i,j,k}} B^{(2)}_{ijk} \Big[ -B^{(0,1)}_{ixy} - B^{(0,1)}_{jxy} \Big] \prob{I_x I_y S_i S_j I_k} + B^{(2)}_{ijk}\, B^{(0,1)}_{kxy}\, \prob{S_i S_j S_k I_x I_y} \bigg).
\]

\textit{External pairwise infection terms.} These are the only terms requiring expansion via Lemma~\ref{lem:burgio-substitution}. The intermediate form is
\[
    \beta_1 \sum_{x \neq i,j,k} B^{(2)}_{ijk} \Big[ -A_{ix} - A_{jx} \Big] \prob{I_x S_i S_j I_k} + B^{(2)}_{ijk}\, A_{kx}\, \prob{S_i S_j S_k I_x}.
\]
Applying Lemma~\ref{lem:burgio-substitution} to each of the three terms and using the bracket notation from the proof of Theorem~\ref{thm:burgio-order2}:
\begin{equation}\label{eq:app-SSI-beta1-expanded}
\begin{aligned}
    &\sum_{x \neq i,j,k} B^{(2)}_{ijk} \Big[ -A_{ix} - A_{jx} \Big] \prob{I_x S_i S_j I_k} + B^{(2)}_{ijk}\, A_{kx}\, \prob{S_i S_j S_k I_x} \\
    = -{\Bigg[}& {\sum_{x \neq i,j,k} B^{(2)}_{ijk}\, B^{(1)}_{ix}\, \prob{S_i S_j I_k I_x}} \\
    &{+ \frac{1}{2} \sum_{\substack{x,y \\ y \neq x}} B^{(2)}_{ijk}\, B^{(1,0)}_{ixy} \Big( \prob{S_i S_j I_k I_x S_y} + \prob{S_i S_j I_k S_x I_y} + 2\,\prob{S_i S_j I_k I_x I_y} \Big) \Bigg]} \\
    &- {\Big[ i \leftrightarrow j \Big]} + {\Big[ I_k \to S_k,\; i \leftrightarrow k \Big]}.
\end{aligned}
\end{equation}

Collecting all terms recovers equation~\eqref{eq:burgio-target-order3-SSI}.

\subsection{Derivation of SII equation}\label{app:burgio-SII}
Multiplying equation~\eqref{eq:dSII} by $B^{(2)}_{ijk}$ and simplifying term by term:

\medskip
\textit{Recovery terms.} The recovery terms pass through unchanged:
\[
    \gamma\, B^{(2)}_{ijk} \Big( \prob{I_i I_j I_k} - 2\,\prob{S_i I_j I_k} \Big).
\]

\textit{Internal pairwise infection terms.} Using equation~\eqref{eq:B2-edge} on each prefactor: $B^{(2)}_{ijk}\,A_{ij} = B^{(1,0)}_{ijk}$, $B^{(2)}_{ijk}\,A_{ik} = B^{(1,0)}_{ijk}$, and $B^{(2)}_{ijk}\,A_{jk} = B^{(1,0)}_{ijk}$. This gives
\[
    \beta_1 \Big( -2\,B^{(1,0)}_{ijk}\,\prob{S_i I_j I_k} + B^{(1,0)}_{ijk} \big[ \prob{S_i I_j S_k} + \prob{S_i S_j I_k} \big] \Big).
\]

\textit{Internal hyperedge infection term.} The prefactor $B^{(2)}_{ijk}\,H_{ijk} = B^{(2)}_{ijk}\,B^{(0,1)}_{ijk} = B^{(0,1)}_{ijk}$, giving
\[
    -\beta_2\, B^{(0,1)}_{ijk}\, \prob{S_i I_j I_k}.
\]

\textit{Mixed hyperedge infection terms.} By the same sparsity argument as in the $\prob{S_i S_j I_k}$ case, the prefactors $B^{(2)}_{ijk}\,H_{ijx}$, $B^{(2)}_{ijk}\,H_{ikx}$, and $B^{(2)}_{ijk}\,H_{jkx}$ each require two $3$-cliques to share a pair, forcing the summation index to equal the excluded node. Hence these terms vanish.

\textit{External hyperedge infection terms.} Substituting $H = B^{(0,1)}$ and converting to unordered sums:
\begin{multline*}
    \beta_2 \bigg( \frac{1}{2} \sum_{\substack{x,y \\ y \neq x,\; x,y \neq i,j,k}} -B^{(2)}_{ijk}\, B^{(0,1)}_{ixy}\, \prob{I_x I_y S_i I_j I_k} \\
    + B^{(2)}_{ijk}\, B^{(0,1)}_{jxy}\, \prob{S_i S_j I_k I_x I_y} + B^{(2)}_{ijk}\, B^{(0,1)}_{kxy}\, \prob{S_i I_j S_k I_x I_y} \bigg).
\end{multline*}

\textit{External pairwise infection terms.} Applying Lemma~\ref{lem:burgio-substitution} to each of the three terms:
\begin{equation}\label{eq:app-SII-beta1-expanded}
\begin{aligned}
    &\sum_{x \neq i,j,k} -B^{(2)}_{ijk}\, A_{ix}\, \prob{I_x S_i I_j I_k} + B^{(2)}_{ijk}\, A_{jx}\, \prob{I_x S_i S_j I_k} + B^{(2)}_{ijk}\, A_{kx}\, \prob{I_x S_i I_j S_k} \\
    = -{\Bigg[}& {\sum_{x \neq i,j,k} B^{(2)}_{ijk}\, B^{(1)}_{ix}\, \prob{S_i I_j I_k I_x}} \\
    &{+ \frac{1}{2} \sum_{\substack{x,y \\ y \neq x}} B^{(2)}_{ijk}\, B^{(1,0)}_{ixy} \Big( \prob{S_i I_j I_k I_x S_y} + \prob{S_i I_j I_k S_x I_y} + 2\,\prob{S_i I_j I_k I_x I_y} \Big) \Bigg]} \\
    &+ {\Big[ I_j \to S_j,\; i \leftrightarrow j \Big]} + {\Big[ I_k \to S_k,\; i \leftrightarrow k \Big]}.
\end{aligned}
\end{equation}

Collecting all terms recovers equation~\eqref{eq:burgio-target-order3-SII}.

\section{Inter-order model: four- and five-node motif derivations}
\label{app:inter-order-derivations}

The lemmas in this appendix derive the closed-form reductions of the pair- and triplet-level rate-equation terms that involve four- or five-node motifs. Throughout, we work under the strong sparsity assumption (\ref{ass:io-strong-sparsity}) together with the dynamical and topological homogeneity assumptions (Assumptions~\ref{ass:dynamical-homogeneity} and \ref{ass:topological-homogeneity}). The strong sparsity assumption restricts the admissible local topologies on four and five nodes to a single isomorphism class in each case --- the pendant motif on four nodes and the bowtie motif on five nodes --- which substantially simplifies the closure.

\begin{lemma}\label{lem:pair-ext-hyper}
  Under Assumptions~\ref{ass:io-strong-sparsity}, \ref{ass:dynamical-homogeneity}, and~\ref{ass:topological-homogeneity}, the pair-level external hyperedge infection term on four nodes,
  \[
    A_{ij} \sum_{\substack{x < y \\ x, y \neq i, j}} H_{ixy}\, \prob{I_x I_y S_i I_j},
  \]
  reduces to
  \[
    \frac{k_2}{k_1}\,(k_1 - 2\alpha)\, \prob{I I S_\bigcirc I}.
  \]
\end{lemma}

\begin{proof}
  Strong sparsity (Assumption~\ref{ass:io-strong-sparsity}) admits only one isomorphism class of four-node sub-graph: the pendant motif. Under dynamical homogeneity (Assumption~\ref{ass:dynamical-homogeneity}), the unspecified joint probability collapses to the single class-specific probability, with the adjacency prefactors selecting that class on both sides of the approximation:
  \[
    A_{ij}\, H_{ixy}\, \prob{I_x I_y S_i I_j} \approx A_{ij}\, H_{ixy}\, \prob{I I S_\bigcirc I}.
  \]

  We next handle the restricted summation. Since $H_{ixy}$ is symmetric in $x, y$ and the state probability is symmetric under $I_x \leftrightarrow I_y$, we first lift the $x < y$ ordering at the cost of a factor~$\tfrac{1}{2}$:
  \[
    A_{ij} \sum_{\substack{x < y \\ x, y \neq i, j}} H_{ixy}\, \prob{I_x I_y S_i I_j} = \frac{1}{2} A_{ij} \sum_{\substack{x \neq y \\ x, y \neq i, j}} H_{ixy}\, \prob{I_x I_y S_i I_j}.
  \]
  The unordered restricted sum then decomposes as an unrestricted sum minus the excluded indices:
  \[
    \sum_{\substack{x \neq y \\ x, y \neq i, j}} f(i, j, x, y) = \sum_{x \neq y} f(i, j, x, y) - \sum_{x} \bigl[f(i, j, x, i) + f(i, j, x, j)\bigr] - \sum_{y} \bigl[f(i, j, i, y) + f(i, j, j, y)\bigr].
  \]

  For the first term, observe that the fixed edge $A_{ij}$ is independent of the hyperedges $H_{ixy}$ enumerated by the sum. By topological homogeneity (Assumption~\ref{ass:topological-homogeneity}), the unordered double sum counts each hyperedge incident to~$i$ twice, giving
  \[
    \frac{1}{2}\, A_{ij} \sum_{x \neq y} H_{ixy} = \frac{1}{2} \cdot 2 k_2 = k_2.
  \]

  The correction terms account for the excluded indices. The cases $x = i$ and $y = i$ vanish since $H_{iiy} = H_{ixi} = 0$ by the diagonal convention of Definition~\ref{def:adjacency-tensors}. The remaining two cases, $x = j$ and $y = j$, each contribute $\sum_x H_{ixj}$ --- the symmetry factor of~$2$ here cancels the prefactor $\tfrac{1}{2}$:
  \[
    \frac{1}{2} \cdot 2\, A_{ij} \sum_{x} H_{ixj} = A_{ij} \sum_{x} H_{ixj} \approx \frac{\sum_{i, j, x} A_{ij}\, H_{ijx}}{\sum_{i, j} A_{ij}} = \frac{2 k_2 \alpha}{k_1}.
  \]

  Combining the two contributions yields
  \[
    k_2 - \frac{2 k_2 \alpha}{k_1} = \frac{k_2}{k_1}\,(k_1 - 2\alpha),
  \]
  as claimed.
\end{proof}

\begin{lemma}\label{lem:triplet-int-hyper}
  Under Assumptions~\ref{ass:io-strong-sparsity}, \ref{ass:dynamical-homogeneity}, and~\ref{ass:topological-homogeneity}, the internal hyperedge contribution to the triplet rate equation for $\prob{S_i S_j I_k}$,
  \[
    H_{ijk}\, A_{ik}\, \prob{S_i S_j I_k},
  \]
  reduces to
  \[
    \alpha\, \prob{S S I}^{\bigcirc}.
  \]
\end{lemma}

\begin{proof}
  Under dynamical homogeneity (Assumption~\ref{ass:dynamical-homogeneity}), triplets containing a hyperedge share a single dynamically distinguishable isomorphism class, with the hyperedge prefactor selecting that class on both sides:
  \[
    H_{ijk} \, A_{ik} \, \prob{S_i S_j I_k} \approx H_{ijk}\, A_{ik} \, \prob{S S I}^{\bigcirc}.
  \]
  The pairwise prefactor $A_{ik}$ is treated using topological homogeneity (Assumption~\ref{ass:topological-homogeneity}): the local count of pairwise edges within a hyperedge is replaced by its network-wide average~$\alpha$, defined by
  \[
    \alpha = \frac{\sum_{i, j, k} H_{ijk}\, A_{ik}}{\sum_{i, j, k} H_{ijk}}.
  \]
  Hence $H_{ijk}\, A_{ik} \approx H_{ijk}\, \alpha$, which absorbs into the hyperedge-tagged triplet probability to give the stated result.
\end{proof}

\begin{lemma}\label{lem:triplet-ext-pair}
  Under Assumptions~\ref{ass:io-strong-sparsity}, \ref{ass:dynamical-homogeneity}, and~\ref{ass:topological-homogeneity}, the external pairwise infection contribution to the triplet rate equation for $\prob{S_i S_j I_k}$,
  \[
    H_{ijk} \sum_{x \neq i, j, k} A_{ix}\, \prob{I_x S_i I_j I_k},
  \]
  reduces to
  \[
    (k_1 - 2\alpha)\, \prob{I I S_\bigcirc I}.
  \]
\end{lemma}

\begin{proof}
  Under strong sparsity (Assumption~\ref{ass:io-strong-sparsity}) and dynamical homogeneity (Assumption~\ref{ass:dynamical-homogeneity}), the only admissible four-node motif is the pendant, with the adjacency prefactors selecting that class on both sides:
  \[
    H_{ijk}\, A_{ix}\, \prob{I_x S_i I_j I_k} \approx H_{ijk}\, A_{ix}\, \prob{I I S_\bigcirc I}.
  \]

  We rewrite the restricted summation as an unrestricted sum minus the excluded indices:
  \[
    H_{ijk} \sum_{x \neq i, j, k} A_{ix} = H_{ijk} \sum_{x} A_{ix} - H_{ijk}\, A_{ii} - H_{ijk}\, A_{ij} - H_{ijk}\, A_{ik}.
  \]

  The self-loop term vanishes since $A_{ii} = 0$. By topological homogeneity (Assumption~\ref{ass:topological-homogeneity}), the unrestricted sum gives $H_{ijk} \sum_{x} A_{ix} = H_{ijk}\, k_1$, and each of the two remaining cross terms satisfies $H_{ijk}\, A_{ik} \approx H_{ijk}\, \alpha$, as in the proof of Lemma~\ref{lem:triplet-int-hyper}. Collecting all contributions yields
  \[
    H_{ijk} \sum_{x \neq i, j, k} A_{ix} \approx H_{ijk}\, (k_1 - 2\alpha),
  \]
  which combines with the homogenized state probability to give the stated result.
\end{proof}

\begin{lemma}\label{lem:triplet-ext-hyper}
  Under Assumptions~\ref{ass:io-strong-sparsity}, \ref{ass:dynamical-homogeneity}, and~\ref{ass:topological-homogeneity}, the external hyperedge infection contribution to the triplet rate equation for $\prob{S_i S_j I_k}$,
  \[
    H_{ijk} \sum_{\substack{x < y \\ x, y \neq i, j, k}} H_{ixy}\, \prob{I_x I_y S_i S_j I_k},
  \]
  reduces to
  \[
    (k_2 - 1)\, \prob{I S S_{\bowtie} I I}.
  \]
\end{lemma}

\begin{proof}
  Strong sparsity (Assumption~\ref{ass:io-strong-sparsity}) admits only one isomorphism class of five-node sub-graph compatible with two hyperedges sharing a single node, namely the bowtie motif. Under dynamical homogeneity (Assumption~\ref{ass:dynamical-homogeneity}), the hyperedge prefactors $H_{ijk}$ and $H_{ixy}$ select this class on both sides of the approximation:
  \[
    H_{ijk}\, H_{ixy}\, \prob{I_x I_y S_i S_j I_k} \approx H_{ijk}\, H_{ixy}\, \prob{I S S_{\bowtie} I I}.
  \]

  For the remaining summation, observe that the outer prefactor $H_{ijk}$ requires node~$i$ to belong to at least one hyperedge, namely $\set{i, j, k}$. Node~$i$ belongs to $k_2$ hyperedges on average, and since $\set{i, j, k}$ is explicitly excluded from the summation, there are $k_2 - 1$ remaining hyperedges incident to~$i$. Hence
  \[
    H_{ijk} \sum_{\substack{x < y \\ x, y \neq i, j, k}} H_{ixy} \approx H_{ijk}\, (k_2 - 1),
  \]
  which combines with the homogenized state probability to give the stated result.
\end{proof}

All other terms in the pair-level and triplet-level rate equations can be obtained via substitutions of state $S \leftrightarrow I$ and node labels and the lemmas within this appendix. Combined with the closures of Remark~\ref{rem:io-closure-canonical}, this completes the derivation of the inter-order model target equations~\eqref{eq:io-target-eqs} and target closures~\eqref{eq:io-closure-triples}--\eqref{eq:io-closure-bowtie}.

\bibliographystyle{unsrt}   

\bibliography{references}

\end{document}